\documentclass[11pt]{article}
\usepackage{graphicx}
\usepackage{amsfonts}
\usepackage{amsmath}
\usepackage{amssymb}
\usepackage{algorithm}
\usepackage{algorithmic}
\usepackage{url}
\usepackage{fancyhdr}
\usepackage{indentfirst}
\usepackage{amsmath, amssymb, mathrsfs, bm}
\usepackage{tabularx, booktabs}
\usepackage{dsfont}
\usepackage{booktabs}
\usepackage{enumerate}
\usepackage{dsfont,footmisc}
\usepackage{comment}
\usepackage{graphicx}
\usepackage{caption}   
\usepackage{tikz}
\usetikzlibrary{arrows.meta, automata, positioning}
\usepackage{subcaption}
\usepackage{subcaption}
\allowdisplaybreaks[4]
\usepackage{stmaryrd}
\usepackage{bbold}
\usepackage{mathrsfs}
\usepackage{graphicx}
\usepackage[export]{adjustbox}
\usetikzlibrary{calc}

\usepackage{setspace}
\usepackage[misc]{ifsym}

\usepackage{amsthm}
\usepackage{color}
\usepackage{natbib}
\usepackage[british]{babel}
\usepackage[colorlinks=true,citecolor=blue]{hyperref}
\usepackage[right,pagewise,displaymath, mathlines]{lineno}

\usepackage{tikz}

\usetikzlibrary{decorations.pathreplacing}
\usepackage{tikz-qtree}
\usepackage{amsmath}
\usepackage{tikz}
\usetikzlibrary{arrows.meta,positioning,fit}

\def\d{\mathrm{d}}

\newcommand{\E}{\mathbb{E}}

\renewcommand{\ge}{\geqslant}
\renewcommand{\le}{\leqslant}
\renewcommand{\geq}{\geqslant}
\renewcommand{\leq}{\leqslant}
\renewcommand{\epsilon}{\varepsilon}

\theoremstyle{plain}
\newtheorem{theorem}{Theorem}
\newtheorem{corollary}{Corollary}
\newtheorem{lemma}{Lemma}
\newtheorem{proposition}{Proposition}

\theoremstyle{definition}
\newtheorem{assumption}{Assumption}
\newtheorem{example}{Example}[assumption]

\theoremstyle{remark}

\theoremstyle{definition}

\renewcommand{\cite}{\citet}
\renewcommand{\cdots}{\dots}

\begin{document} 
	
\title{Robust Bayesian Dynamic Programming for On-policy Risk-sensitive Reinforcement Learning}

\author{
    Shanyu Han\thanks{\scriptsize School of Mathematical Sciences, Peking University, China. \Letter~\texttt{hsy.1123@pku.edu.cn}}
     \and
	Yangbo He\thanks{\scriptsize School of Mathematical Sciences, Peking University, China. \Letter~\texttt{heyb@math.pku.edu.cn}}
    \and
	Yang Liu\thanks{\scriptsize School of Science and Engineering, The Chinese University of Hong Kong, Shenzhen, China. \Letter~\texttt{yangliu16@cuhk.edu.cn}}
   }

\maketitle

\begin{abstract}

We propose a novel framework for risk-sensitive reinforcement learning (RSRL) that incorporates robustness against transition uncertainty. We define two distinct yet coupled risk measures: an inner risk measure addressing state and cost randomness and an outer risk measure capturing transition dynamics uncertainty.
Our framework unifies and generalizes most existing RL frameworks by permitting general coherent risk measures for both inner and outer risk measures. 
Within this framework, we construct a risk-sensitive robust Markov decision process (RSRMDP), derive its Bellman equation, and provide error analysis under a given posterior distribution.
We further develop a Bayesian Dynamic Programming (Bayesian DP) algorithm that alternates between posterior updates and value iteration. 
The approach employs an estimator for the risk-based Bellman operator that combines Monte Carlo sampling with convex optimization, for which we prove strong consistency guarantees. 
Furthermore, we demonstrate that the algorithm converges to a near-optimal policy in the training environment and analyze both the sample complexity and the computational complexity under the Dirichlet posterior and CVaR.
Finally, we validate our approach through two numerical experiments. The results exhibit excellent convergence properties while providing intuitive demonstrations of its advantages in both risk-sensitivity and robustness. Empirically, we further demonstrate the advantages of the proposed algorithm through an application on option hedging.

	\noindent \textbf{Keywords}: risk-sensitive reinforcement learning, Bayesian dynamic programming, robust reinforcement learning, Markov decision processes, coherent risk measures, Laguerre tessellation, optimal transport, convex optimization. 
\end{abstract}

\section{Introduction}
Reinforcement learning (RL) focuses on developing agents that learn optimal policies through interactions with the environment to maximize cumulative rewards or minimize cumulative costs. RL has gained significant attention across multiple domains such as robotics (\cite{gu2017deep}, \cite{brunke2022safe}), finance (\cite{deng2016deep}, \cite{du2020deep}), and games (\cite{silver2018general}). 
However, when applied to real-world tasks, RL typically encounters two key challenges. First, learned policies often optimize only for expected rewards while ignoring rare but potentially catastrophic outcomes, which can introduce substantial risks. Second, when there are discrepancies between the training environment and the real-world deployment environment, the resulting policies tend to suffer severe performance degradation. These issues highlight the lack of risk sensitivity and robustness in conventional RL approaches. Recent research has increasingly emphasized these two aspects, giving rise to two important frameworks: risk-sensitive reinforcement learning (RSRL) and robust reinforcement learning. These two RL frameworks exhibit clear distinctions. RSRL replaces the expectation on rewards with alternative functionals that capture risk, commonly referred to as risk measures. This risk-sensitive viewpoint has also been modeled as an optimization problem under risk-measure-based constraints; for example, \cite{fang2023fairness} study fair individual treatment rule using a Value-at-Risk constraint.  In contrast, robust RL typically accounts for uncertainties in the transition kernel. Some studies have revealed that under certain conditions, the two approaches can be equivalent (\cite{osogami2012robustness}, \cite{shen2013risk}, \cite{chow2015risk}, \cite{zhang2023soft})—for instance, a coherent risk measure can be reformulated as the supremum of expectations over an uncertainty set of transition kernels. Nevertheless, relatively few works investigate both perspectives simultaneously. A remaining challenge lies in handling robustness to model mis-specification and distributional shift under non-expectation-based risk objectives (or preferences). 

Scenarios where both risk sensitivity and robustness are simultaneously required are widespread.
This idea is closely related to worst-case risk problems studied in distributionally robust optimization (DRO), which jointly consider non-expectation-based objectives and model uncertainty (see \cite{kuhn2025distributionally}). 
A representative example arises in financial portfolio management. The performance of a fund manager is typically not evaluated solely by average returns, but rather through risk-adjusted measures such as the Sharpe ratio (expected return divided by standard deviation) or, equivalently, mean–variance utility. At the same time, the future distribution of asset prices is inherently uncertain and can only be estimated from historical data, which highlights the necessity of robustness. Another illustrative example comes from control problems in autonomous driving. In this setting, controlling tail risks is of fatal importance, since extreme events often correspond to catastrophic outcomes such as traffic accidents. On the other hand, the road data used for training typically deviates from real-world driving conditions, which likewise necessitates considerations of robustness.

Broadly speaking, RSRL captures non-expectation-based preferences over uncertainties in states and rewards (or costs), while robust reinforcement learning addresses uncertainties in the transition dynamics. Although these two types of uncertainties could, in principle, be unified into a broader notion of overall uncertainty, modeling them separately offers irreplaceable advantages—most notably enhanced intuitiveness and interpretability. We illustrate this point with a heuristic example. Although the example does not involve sequential decision-making in the RL sense, it clearly highlights the importance of distinguishing between risk-sensitivity and robustness.
Consider a policymaker whose goal is to improve the average income of the bottom 5\% of earners rather than the overall average income—formally, to minimize the Conditional Value-at-Risk (CVaR) at level $0.05$ of a random loss $-X,$ where $X$ denotes the random income with a distribution $P.$ Meanwhile, the income distribution $P$ is unknown; the policymaker places a prior $P\sim\chi$ over plausible distributions. We consider two modeling approaches. In the first, the policymaker integrates both types of uncertainty, with the objective denoted as $\mathrm{CVaR}_{0.05}(-X;X\sim P,P\sim \chi).$ In the second, the policymaker separately models the randomness of income and the uncertainty over its distribution, leading to the objective with double-layered risk measures denoted as $\mathrm{CVaR}_{0.5}(\mathrm{CVaR}_{0.05}(-X;X\sim P);P\sim\chi).$  The first case corresponds to the CVaR under the marginal distribution of  $X$ which can be interpreted as \textit{“the average income of the bottom 5\% of earners as subjectively perceived by the policymaker.” } In contrast, the second case has a different meaning: it can be interpreted as \textit{“ensuring the average income of the bottom 5\% of earners even under the worst 50\% of scenarios where the policymaker’s knowledge about the distribution is most inaccurate.”} Clearly, the latter offers stronger fairness and rationality, and its formulation serves as the main source of inspiration for the framework developed in this paper. 

Specifically, in this paper, we construct a new class of Markov decision processes (MDPs) using a double-layered risk measure similar to the one discussed earlier. This formulation can be viewed as a synthesis of the dynamic risk problem in RSRL and the BRMDP framework (see \cite{lin2022bayesian}) in robust reinforcement learning.
Our primary contributions are threefold. First, we establish a novel measure-theoretic MDP framework centered on double-layered risk measures (inner and outer). The framework is intuitive and provides risk-sensitivity through the inner risk measure and robustness via the outer risk measure. By accommodating general coherent risk measures, this formulation generalizes most existing RL frameworks, including conventional RL, DRP, classical Bayesian RL, and the RL framework in \cite{wang2023bayesian}. Furthermore we provide theoretical foundations such as Bellman equations and error analysis. Second, we develop a Bayesian Dynamic Programming approach for model-free on-policy learning within this framework. To estimate the doubly-nonlinear Bellman operator, we combine Monte Carlo simulation with convex optimization techniques. Third, we theoretically and experimentally validate the proposed Bayesian DP method. Theoretically, we prove (a) strong consistency of the Bellman operator estimator, (b) posterior convergence to the true transition, (c) overall algorithmic convergence, and (d) characterization of both the sample complexity and the computational complexity, including both the number of iterations required for convergence and the cost of each iteration.
Experimentally, we validate the method's risk-sensitivity, convergence properties, and robustness through two simple yet illustrative experiments. In the risk-neutral setting, we benchmark our approach against classical Q-learning (\cite{watkins1992q}) and two representative DRRL frameworks (\cite{liu2022distributionally, neufeld2024robust}). Under a CVaR-based inner risk measure, we compare with iterated CVaR RL (\cite{du2022provably}).

\subsection{Literature}
RSRL aims to incorporate risk factors into the policy learning process, focusing not simply on expectation but rather on other functionals (called risk measures) that account for variability. RSRL can be divided into two categories: static risk problems (SRP) and dynamic risk problems (DRP). The SRP optimization applies a global risk measure to the cumulative cost. \cite{bauerle2011markov} propose using state augmentation to study SRP with a conditional-value-at-risk (CVaR) objective and reduce the problem to an ordinary MDP. 
Subsequent studies have built on this to investigate model-free reinforcement learning under static CVaR (\cite{chow2018risk}, \cite{prashanth2014policy}, \cite{wang2023near}). Other common forms of SRP include the use of mean–variance utility (\cite{di2012policy}, \cite{la2013actor}, \cite{xie2018block}) and the entropy risk measure (ERM; \cite{fei2020risk} and \cite{fei2021risk}). Recently, 
\cite{ni2022policy} develop a policy gradient approach for Entropic Value-at-Risk (EVaR) objectives and \cite{han2025risk} adopt convex scoring functions to handle SRP with a unified class of risk measures including variance, CVaR, ERM, EVaR, and mean-risk utilities. 
SRP is also related to distributional reinforcement learning. While the earliest studies in distributional RL focused on learning the optimal distribution of reward while still optimizing the expectation (see \cite{bellemare2017distributional}), more recent works have examined SRP from a distributional perspective. For example, \cite{kim2024risk} employ policy gradient to study distributional reinforcement learning under static CVaR, and \cite{chen2024provable} apply a general function approximation to investigate distributional reinforcement learning under static Lipschitz risk measures. 
More recent work on SRP can be found in \cite{zhang2021mean}, \cite{wang2024risk} and \cite{ni2024risk}. In contrast, DRP applies recursive risk measures to the cost at each step. The early formulation of DRP, as seen in \cite{mihatsch2002risk}, employs relatively simple piecewise linear functions as risk measures. 
\cite{ruszczynski2010risk} proposes the MDP framework under recursive coherent risk measures and establishes theoretical foundations including Bellman equations, value iteration, and policy iteration. 
Some subsequent work has focused on the DRP problem under iterated CVaR (see \cite{du2022provably} and \cite{chen2023provably}). 
\cite{tamar2016sequential} design an Actor-Critic (AC) algorithm for coherent recursive risk measures, and \cite{coache2024reinforcement} extend it to convex risk measures. More recent work on DRP is available in \cite{coache2024robust}, \cite{coache2023conditionally}, \cite{liang2024regret} and \cite{yu2022risk}. 
In our current work, we focus on DRP with recursive coherent risk measures, due to their inherent time-consistent advantages. 

The most widely studied form of robust reinforcement learning is distributionally robust reinforcement learning (DRRL), which is formulated within the MDP framework augmented with distributionally robust optimization (DRO) techniques. DRRL optimizes for the worst-case performance when the unknown model parameters are in an ambiguity set. 
\cite{shen2013risk} conduct a theoretical study of MDPs under DRO and establish connections with MDPs under static risk measures in SRP. 
\cite{liu2022distributionally} and \cite{neufeld2024robust} develop Q-learning for DRRL, where \cite{liu2022distributionally} employ a KL-divergence-based ambiguity set and \cite{neufeld2024robust} employ a Wasserstein-distance-based ambiguity set. 
\cite{shi2024distributionally} propose a model-based offline DRRL algorithm with near-optimal sample complexity.
More recent studies in DRRL can be found in \cite{badrinath2021robust}, \cite{wang2023model}, \cite{blanchet2023double}, \cite{zhou2023natural}, and \cite{zhang2023soft}.
An alternative method for introducing robustness is the Bayesian risk optimization (BRO) framework (\cite{wu2018bayesian}, \cite{zhou2015simulation}), which quantifies parameter uncertainty through a risk measure over Bayesian posterior distributions. Compared to DRO, BRO offers greater flexibility since DRO's exclusive focus on worst-case scenarios often leads to overly conservative policies. \cite{lin2022bayesian} propose a BRO framework for MDP (termed BRMDP) and a dynamic programming algorithm in finite-horizon scenarios. \cite{wang2023bayesian} develop this approach for infinite-horizon problems and distinguish from the classical Bayesian dynamic programming (Bayesian DP) by applying risk measures rather than expectations over the posterior distributions.
In this work, we employ the BRMDP framework to introduce robustness into our approach. Our work focuses on on-policy learning, which \cite{wang2023bayesian} identify in their concluding remarks as an important future research direction. 
Bayesian DP is first proposed by \cite{strens2000bayesian},  which models transition probability using Bayesian posterior distributions and alternately performs posterior updates and dynamic programming value iterations. Other related Bayesian methods for RL are reviewed in \cite{ghavamzadeh2015bayesian}. 

Although both risk sensitivity and robustness are important concepts, only a few studies have explored their equivalence, and even fewer works in reinforcement learning have examined them jointly. 
Integrating these aspects remains a key challenge in the field. Notable advances include \cite{jaimungal2022robust} and \cite{queeney2023risk}, which address the SRP and DRP, respectively, using DRO to embed robustness. \cite{ni2024robust} propose a robust SRP under CVaR using DRO. Further, \cite{pan2019risk} tackle this integration through the framework of robust adversarial RL (RARL).  The objective of our current work is to provide a robust risk-sensitive RL framework by unifying the DRP and BRMDP with double-layered risk measures.

\section{Preliminaries and Problem Formulation}
\paragraph{Standard MDP.}
Consider a Markov decision process (MDP) with a finite state space $\mathcal{S}=\{s^1,\ldots,s^K\}$ and a finite action space $\mathcal{A} = \{a^1,\ldots,a^B\}$, where $c(s,a,s')$ is a deterministic, state-action-dependent reward function, bounded by $\bar{C}=\max_{s,a,s'}|c(s,a,s')|.$  A Markov policy is a function $\pi:\mathcal{S}\times\mathcal{A}\to[0.1],$ meaning the probability of taking action $a$ at state $s$ and we define the value as $\pi(s|a).$ Denote all the Markov policies by $\Pi.$ 
To model the discrepancy between the training environment and the real
environment, we consider defining the transition probability as a random vector following some distribution $\chi,$ which is inspired by \cite{liu2022distributionally}.  Additionally, in the later section, we select a subjective distribution $\chi$ as the prior and subsequently update it to the posterior using Bayes' theorem.

\paragraph{Transition Probability.}
Below we show our novel definition of the transition probability, which is considered as a ``random vector" on a special ``probability space". We rigorously define them 
based on a measure-theoretic language. 
Subsequently, we formulate all risk measures within this measure-theoretic foundation. 
Denote by $\mathbb{P}_0$ the counting measure on $(\mathcal{S},2^{\mathcal{S}})$. Let $\mathscr{V}_0 = \mathscr{L}_{n_0}(\mathcal{S},2^{\mathcal{S}},\mathbb{P}_0),\mathscr{Y}_0 = \mathscr{L}_{m_0}(\mathcal{S},2^{\mathcal{S}},\mathbb{P}_0)$ with $n_0,m_0\in(1,\infty),$ and $\frac{1}{n_0}+\frac{1}{m_0}=1$ (since $\mathcal{S}$ is finite, $\mathscr{V}_0= \mathscr{Y} _0=  \mathscr{L}_{\infty}(\mathcal{S},2^{\mathcal{S}},\mathbb{P}_0)$). We further define
\begin{equation}\begin{aligned}
    \mathscr{M}&= \left\{m\in\mathscr{Y}_0:\sum_{s'\in\mathcal{S}}m(s')=1,m(s')\geq 0,\forall s'\in\mathcal{S}\right\}\ \text{and}\ \mathscr{P} =\mathscr{M}^{|\mathcal{S}|\times|\mathcal{A}|}. 
 \end{aligned}\notag\end{equation}
The element of $\mathscr{P}$ is represented as $q =\left(q(\cdot|s^1,a^1),q(\cdot|s^1,a^2),\ldots,q(\cdot|s^K,a^B)\right),$ where $q(\cdot|s,a)=(q(s^1|s,a),\ldots,q(s^K|s,a)).$ Denote by $\mathscr{B}( \mathscr{P})$ the Borel $\sigma$-algebra of $\mathscr{P}.$ For any distribution $\chi$ along with the corresponding probability measure $\mathbb{P}^\chi,$ 
let $\mathscr{V}^{\chi}_1 = \mathscr{L}_{n_1}(\mathscr{P},\mathscr{B}( \mathscr{P}),\mathbb{P}^\chi),\mathscr{Y}^{\chi}_1 = \mathscr{L}_{m_1}(\mathscr{P},\mathscr{B}( \mathscr{P}),\mathbb{P}^\chi)$ with $n_1,m_1\in(1,\infty),$ and $\frac{1}{n_1}+\frac{1}{m_1}=1.$ We define the transition probability $p$ as an element on $\mathscr{V}^{\chi}_1,$ which is a ``random vector". Here we emphasize the distinction between $\mathscr{P}$ and the special ``probability space" $\mathscr{V}^{\chi}_1:$ the elements of the former are deterministic vectors (discretely-distributed random variables), which we denote by $q,$ while those of the latter are ``random variables", which we denote by $p.$ 

\paragraph{Inner and Outer Risk Measures.}
We now formally define the inner and outer risk measures, which constitute the core of our robust risk-sensitive RL framework. For any $q\in\mathscr{P},\pi\in\Pi,$ and initial state distribution $\mu_0,$ initial action distribution $\tau_0,$ by Ionescu-Tulcea Theorem (see \cite{klenke2013probability}), there exists a unique probability measure $\mathbb{P}^{q,\pi,\mu_0,\tau_0}$ on $(\Omega,\mathscr{F})=((\mathcal{S})^{\infty},(2^\mathcal{S})^{\infty}),$ such that 
\begin{itemize}
    \item[(1)] $\mathbb{P}^{q,\pi,\mu_0,\tau_0}(S_0=s')=\mu_0(s');$
    \item[(2)] $\mathbb{P}^{q,\pi,\mu_0,\tau_0}(S_1=s'|S_0=s_0)=\sum_{a\in\mathcal{A}}\tau_0(a|s_0)q(s'|s_0,a);$
    \item[(3)] $\mathbb{P}^{q,\pi,\mu_0,\tau_0}(S_{t+1}=s'|S_0=s_0,\ldots,S_t=s_t) = \sum_{a\in\mathcal{A}}\pi(a|s)q(s'|s_t,a),\forall t\geq 1.$
\end{itemize}
Thus there exists a random trajectory $X^{q,\pi,\mu_0,\tau_0}=(S_0,A_0,S_1,\ldots,S_t,A_t,\ldots)$ following $\mathbb{P}^{q,\pi,\mu_0,\tau_0}.$ Define $\mathscr{F}_t=\sigma(S_0,A_0,S_1,A_1,\ldots,A_{t-1},S_t),$ and 
$\mathscr{F}_0\subset\mathscr{F}_1\subset\cdots\subset\mathscr{F}$ is a filtration on $(\Omega,\mathscr{F}).$ For any $q\in\mathscr{P},$ and any $\pi\in\Pi,$ we consider one-step conditional risk measures $\rho_{q,\pi,0},\rho_{q,\pi,1},\ldots,\rho_{q,\pi,t},\ldots$ such that $\rho_{q,\pi,t}:\mathscr{F}_{t+1}\to\mathscr{F}_{t}.$ We further assume $\left\{\rho_{q,\pi,t}\right\}_{t\geq0}$ are stationary Markov risk measures (\cite{ruszczynski2010risk}) w.r.t. process $X^{q,\pi,\mu_0,\tau_0},$ i.e., there exists a risk transition mapping $\sigma:\mathscr{V}_0\times\mathscr{M}\to\mathbb{R}$ such that
\begin{equation}
    \rho_{q,\pi,t} (v(S_t,A_t,S_{t+1})) = \sum_{a\in\mathcal{A}} \pi(a|S_t) \cdot\sigma(v(S_t,a,\cdot),q(\cdot|S_t,a)),\notag
\end{equation}
for all $(S_t,A_t,S_{t+1})_{t\geq0}$ from $X^{q,\pi,\mu_0,\tau_0}.$ Meanwhile, we consider that for any distribution $\chi,$ there exists a risk measure $\beta_{p\sim\chi}$ on $\mathscr{L}_{n_1}(\mathscr{P},\mathscr{B}( \mathscr{P}),\mathbb{P}^{\chi}).$ We refer to $\rho_{q,\pi,t}$ as the inner risk measure and $\beta_{\chi}$ as the outer risk measure. In this paper, both inner and outer risk measures are assumed to be coherent, which is widely used in quantitative finance and operations research; see \cite{delbaen2002coherent}, \cite{ahmed2007coherent}, \cite{jaschke2001coherent} and \cite{FLW24}. Given a set of random variable $\mathcal{X},$ a risk measure $f:\mathcal{X}\to\mathbb{R}$ is coherent if: 
\begin{itemize}
\item[(1)] $f(cX)=cf(X),$ for any $c\geq0$ and $X\in\mathcal{X};$
\item[(2)]  $f(X_1)\geq f(X_2),$ for any $X_1,X_2\in\mathcal{X}$ satisfying $X_1\geq X_2$ a.s.; 
\item[(3)]  $f(X+c) =f(X)+c$ for any $c\in\mathbb{R}$ and $X\in\mathcal{X};$
\item[(4)] $f(X_1+X_2)\leq f(X_1)+f(X_2)$  for any $X_1,X_2\in\mathcal{X}.$
\end{itemize}
 The inner risk measure $\rho_{q,\pi,t}$ characterizes randomness in costs, while the outer risk measure $\beta_{p\sim\chi}$ captures uncertainty in transition probability. The inner risk measure introduces risk-sensitivity to RL, while the outer risk measure ensures robustness against environmental uncertainties. 
In most risk-sensitive RL studies, the adopted risk measures exclusively operate as inner risk measures. In contrast, outer risk measures, which explicitly depend on the posterior of the transition probability, are theoretically meaningful only within Bayesian RL frameworks. To the best of our knowledge, \cite{wang2023bayesian} provide the only existing work that employs an outer risk measure. However, their framework adopts a risk-neutral perspective toward costs, as the inner risk measure simplifies to an expectation. 

\paragraph{Risk Sensitive Robust MDP (RSRMDP) and RL Problem.}
 Below we formalize our robust RL objective in a novel risk-sensitive robust MDP framework.
 In this paper, the total risk is formally defined based on the two risk measures defined above. For any initial state distribution $\mu_0,$ initial action distribution $\tau_0$ and $p\in\mathscr{V}_1^{\chi}$ following distribution $\chi,$ we define the total risk as
\begin{equation}
\begin{aligned}
\text{Risk}(\chi,\pi,\mu_0,\tau_0)
= &\beta_{p\sim\chi}({\rho}_{p,\pi,0}(c(S_0,A_0,S_1)+\\
&\quad\gamma\beta_{p\sim\chi}(\mathscr{\rho}_{p,\pi,1}( c(S_1,A_1,S_2) +\\
&\quad\quad\gamma\beta_{p\sim\chi}(\mathscr{\rho}_{p,\pi,2}\left( c(S_2,A_2,S_3) 
+\cdots\right)))))),
\end{aligned}\notag
\end{equation}
where $(S_t,A_t)_{t\geq 0}$ is from $X^{q,\pi,\mu_0,\tau_0}.$ 
We define the value function $
    V_{\chi,\pi}(s) = \text{Risk}\left(\chi,\pi,\delta_{s},\pi\right),
$
where $\delta_s$ and $\delta_a$ are the point mass on $s$ and $a,$ respectively. We can readily derive Proposition \ref{p0}. Also we assume that there exists a true but unknown training transition probability $\bar{q}\in \mathscr{P},$ i.e., the true training distribution of the transition probability is $\delta_{\bar{q}}.$ 
\begin{proposition}
\label{p0} 
For any $s\in\mathcal{S},$ we have $V_{\chi,\pi}(s)\leq \frac{\bar{C}}{1-\gamma}.$ 
\end{proposition}

Next, we turn to our RL problem. In robust RL, algorithms face a trade-off between optimality and robustness. Achieving optimality in the training environment often conflicts with maintaining robustness in environments different from the training one. For instance, RL algorithms within the DRO framework typically optimize for the worst-case performance, resulting in overly conservative policies that fail to achieve the optimal policy. In contrast, ignoring the discrepancies between training and deployment environments, the RL problem is to optimize 
\begin{equation}
\label{true_RLproblem}
    \min_{\pi\in\Pi} \phi(\text{Risk}\left(\delta_{\bar{q}},\pi,\mu_0,\pi\right)),
\end{equation}
for some initial state distribution $\mu_0$ and some functional $\phi:\mathscr{F}_0\to\mathbb{R}.$ In this paper, we propose a middle ground: at each stage, we incorporate a posterior $\chi$ on the transition probability and learn for the optimal policy based on the risk introduced by uncertainty in the transition probability. This leads to solving an RL problem at each stage, given by:  \begin{equation}
\label{RLproblem}
    \min_{\pi\in\Pi} \phi(\text{Risk}\left(\chi,\pi,\mu_0,\pi\right)).
\end{equation}
Based on this approach, we present an algorithm that combines posterior updates with policy optimization (Algorithm \ref{alg:bayes-dp}). This deals with robustness against model uncertainty with limited interactions in the training environment, while also achieving the near-optimal policy when learning for a longer duration  (Theorem \ref{theorem_convergenceALL}). 
One can control the risk of model mis-specification by setting a higher error tolerance or by adopting an early-stopping strategy, or alternatively achieve the optimal policy in the training environment (i.e., the solution to Problem \eqref{true_RLproblem}), by setting a lower tolerance and training for a longer duration. In this way, a balance between optimality and robustness is attained. 
Figure \ref{framework_com} illustrates a comparison between the proposed stage-wise RL framework and the Q-learning framework, the latter including the standard Q-learning as well as the DRRL Q-learning variants proposed in \cite{liu2022distributionally} and \cite{neufeld2024robust}.
Furthermore, the proposed stage-wise RL framework is suitable for streaming-data scenarios (see, e.g., \cite{wang2023bayesian}).  
In these settings, data are not available all at once but arrive continuously, and the information available before the task starts is limited or unreliable. The agent is required to learn while operating.
Under the stage-wise RL framework used in this paper, the agent continually updates its knowledge of the transition dynamics and manages the risks arising from model uncertainty during decision making, thereby producing adaptive and stable decisions at each stage.
Our empirical results demonstrate the advantage of the proposed approach in these settings. 
Problem \eqref{RLproblem} deeply connects with other RL frameworks:
\begin{itemize}
    \item[(1)] traditional RL emerges when $\chi=\delta_{\bar{q}}$ and $\rho_{q,\pi,t}$ is expectation;
    \item[(2)] risk-sensitive RL for DRP are obtained when $\chi=\delta_{\bar{q}}$ and $\rho_{q,\pi,t}$ is a coherent risk measure;
    \item[(3)] conventional Bayesian RL emerges when $\beta_{\chi}$ and $\rho_{q,\pi,t}$ are both expectation.
\end{itemize}
Furthermore, the framework in \cite{wang2023bayesian} is covered when $\beta_{\chi}$ is Value-at-Risk (VaR) or CVaR and $\rho_{q,\pi,t}$ is expectation. Our framework considers general posterior distributions and general inner/outer risk measures, going beyond all the cases mentioned above. It offers strong adaptability in both risk-sensitivity and robustness while enhancing learning by incorporating prior knowledge.

\begin{figure}[htbp]

\centering
\resizebox{0.8\linewidth}{!}{
\begin{tikzpicture}[
    x=1.4cm, y=1.4cm,
    >=Stealth,
    stagebox/.style={draw,dashed,rounded corners,minimum height=1.1cm},
    stepbox/.style={draw,minimum height=1.2cm},
    lab/.style={font=\scriptsize,align=center}
]

\node[font=\small\bfseries] at (1.5,6.7) {Stage-wise RL framework};
\node[font=\small\bfseries] at (1.5,2.3) {Q-learning framework};

\draw[stagebox] (0,2.6) rectangle (2.8,6.4);
\draw[stagebox] (3.2,2.6) rectangle (6.2,6.4);

\node[lab] at (1.4,6.2) {stage 1};
\node[lab] at (1.4,5.3) {\shortstack{Update posterior $\chi_{(1)}$\\using observations}};
\node[lab] at (1.4,4.15) {Value Iterations\\until convergence};
\node[lab] at (1.4,3.1) {\shortstack{Optimal policy $\pi^*_{(1)}$ \\ for Problem \eqref{RLproblem}}};

\draw[->] (1.4,5.0) -- (1.4,4.5);
\draw[->] (1.4,3.9) -- (1.4,3.4);

\node[lab] at (4.7,6.2) {stage 2};
\node[lab] at (4.7,5.3) {\shortstack{Update posterior $\chi_{(2)}$\\using observations}};
\node[lab] at (4.7,4.15) {Value Iterations\\until convergence};
\node[lab] at (4.7,3.1) {\shortstack{Optimal  policy $\pi^*_{(2)}$\\ for Problem \eqref{RLproblem}}};

\draw[->] (4.7,5.0) -- (4.7,4.5);
\draw[->] (4.7,3.9) -- (4.7,3.4);

\draw[->] (2.5,5.5) -- (3.5,5.5);

\node[lab] at (6.6,4.5) {$\cdots$};


\draw[stepbox] (0.1,0.0) rectangle (1.35,2.08);
\node[lab] at (0.75,1.9) {Step 1};
\node[lab] at (0.75,1.2) {\shortstack{Update \\$Q^{(1)}$ \\for one step}};
\node[lab] at (0.75,0.34) { Policy \\ $\pi_{(1)}$};
\draw[->] (0.75,0.9) -- (0.75,0.55);

\draw[stepbox] (1.75,0.0) rectangle (3.0,2.08);
\node[lab] at (2.45,1.9) {Step $n_1$};
\node[lab] at (2.45,1.2) {\shortstack{Update \\$Q^{(n_1)}$\\for one step}};
\node[lab] at (2.45,0.34) { Policy \\ $\pi_{(n_1)}$};
\draw[->] (2.45,0.9) -- (2.45,0.55);

\draw[stepbox] (3.15,0.0) rectangle (4.50,2.08);
\node[lab] at (3.85,1.9) {Step $n_1+1$};
\node[lab] at (3.85,1.2) {\shortstack{Update \\$Q^{(n_1+1)}$\\for one step}};
\node[lab] at (3.85,0.34) { Policy \\ $\pi_{(n_1+1)}$};
\draw[->] (3.85,0.9) -- (3.85,0.55);

\draw[stepbox] (4.85,0.0) rectangle (6.20,2.08);
\node[lab] at (5.55,1.9) {Step $n_1+n_2$};
\node[lab] at (5.55,1.2) {\shortstack{Update \\$Q^{(n_1+n_2)}$\\for one step}};
\node[lab] at (5.55,0.34) { Policy \\ $\pi_{(n_1+n_2)}$};
\draw[->] (5.55,0.9) -- (5.55,0.55);

\draw[->]
  (1.0,1.08) -- (1.4,1.08)
  (1.8,1.08) -- (2.1,1.08);
\node at ($ (1.4,1.08)!0.5!(1.8,1.08) $) {\bfseries$\cdots$};

\draw[->] (2.8,1.08) -- (3.4,1.08);

\draw[->]
  (4.1,1.08) -- (4.5,1.08)
  (4.9,1.08) -- (5.1,1.08);
\node at ($ (4.5,1.08)!0.5!(4.9,1.08) $) {\bfseries$\cdots$};

\node[lab] at (6.6,1.0) {${\cdots}$};

\end{tikzpicture}
}
\caption{Comparison between stage-wise RL and Q-learning frameworks}
\label{framework_com}
\end{figure}
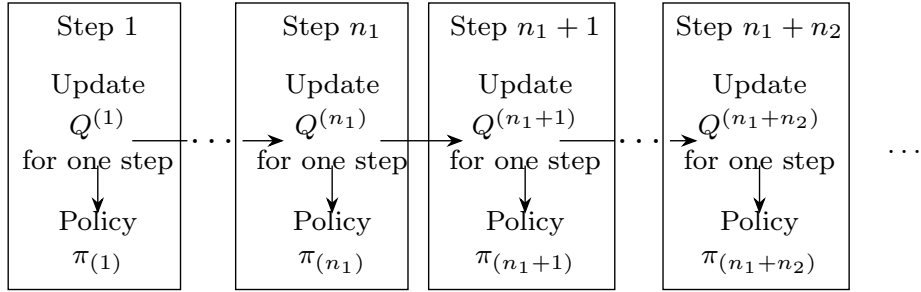

\section{Theoretical Results for RSRMDP}
\subsection{Bellman Equation}
As we have defined the value function above, we further define the Bellman operators on the value function for a fixed posterior $\chi$ as 
\begin{align}
\mathcal{J}_{\chi,\pi}V(s)
&= \beta_{p\sim\chi}\bigg( \sum_{a\in\mathcal{A}} \pi(a|s) \cdot \sigma(c(s,a,\cdot)+\gamma V(\cdot), p(\cdot|s,a)) \bigg),\notag \\
\mathcal{J}_{\chi}V(s)
&=\min_{a\in\mathcal{A}} \beta_{p\sim\chi}\bigg(  \sigma(c(s,a,\cdot)+\gamma V(\cdot), p(\cdot|s,a)) \bigg).\notag
\end{align}
Lemma \ref{lemma_1} demonstrates that both $\mathcal{J}_{\chi,\pi}$ and $\mathcal{J}_{\chi}$ are contraction mappings, a property instrumental in proving Theorem \ref{theorem_1} (the Bellman equation). The Bellman equation serves as a theoretical foundation for the MDP-based RL framework. While \cite{ruszczynski2010risk} establishes the Bellman equation for MDPs under recursive risk measures, our work extends this result by incorporating the posterior $\chi$ of the transition probability. 

\begin{lemma}
\label{lemma_1}
The Bellman operator $\mathcal{J}_{\chi,\pi}$ is uniformly contractive, i.e. \begin{equation}\begin{aligned}
      \left\Vert\mathcal{J}_{\chi,\pi} V_1-\mathcal{J}_{\chi,\pi} V_2\right\Vert_{\infty} \leq \gamma\Vert V_1-V_2\Vert_{\infty},\ \notag \left\Vert\mathcal{J}_{\boldsymbol{\chi}} V_1-\mathcal{J}_{\boldsymbol{\chi}} V_2\right\Vert_{\infty} \leq \gamma\Vert V_1-V_2\Vert_{\infty},\notag
        \end{aligned}\end{equation}
for any prior $\chi$ and policy $\pi,$ where $\left\Vert V\right\Vert_{\infty} = \max_{s\in\mathcal{S}}V(s).$
\end{lemma}

\begin{theorem}[Bellman equation]
\label{theorem_1}
There exists an optimal value $V_{\chi}^*(\cdot)$ such that for any policy $\pi$ and state $s\in\mathcal{S},$
$
  V_{\chi}^*(s)  \leq V_{\chi,\pi}(s),\notag
$
and $V_{\chi}^*$ satisfies $
  V =   \mathcal{J}_{\boldsymbol{\chi}} V.\notag
$
Moreover, an optimal policy $\pi_{\chi}^*$ exists.
\end{theorem}

\begin{corollary}
\label{coro1}
For any initial state distribution $\mu_0,$ $
        \pi_{\chi}^* = \arg\min_{\pi\in\Pi} \left\{\text{Risk}(\chi,\pi,\mu_0,\pi)\right\}$ a.s.
\end{corollary}

\subsection{Posterior Error Analysis }
The Bellman equation we derived is based on a particular posterior $\chi,$ which is an approach to solve Problem \eqref{RLproblem}. However, the optimal policy in the training environment corresponds to Problem \eqref{true_RLproblem}. This naturally leads us to examine how closely the objectives and optimal values of Problems \eqref{true_RLproblem} and \eqref{RLproblem} align.
In other words, we need to analyze the error in the risk quantification caused by model uncertainty.
Intuitively, the approximation quality between these problems should improve as the posterior $\chi$ becomes more accurate. To formally validate this intuitive approximation property, we first introduce Assumption \ref{ass_sigma} on the continuity of the inner risk measure $\rho$ w.r.t. transition probability, which equivalently requires the continuity of the mapping $\sigma$ w.r.t. its second argument.

\begin{assumption}
\label{ass_sigma}
 $\mathscr{\rho}_{q,\pi,t}$ is cross-section continuous w.r.t. $q,$ which is expressed by
\begin{equation}
    \left|\sigma(v,m_1)-\sigma(v,m_2)\right|\leq  B_{\sigma} \cdot \sum_{s'\in\mathcal{S}}\left|m_1(s')-m_2(s') \right|,\notag
\end{equation}
for any $m_1,m_2\in\mathscr{M},s\in\mathcal{S},v \in \mathscr{V}_0$ with $\max_{s\in\mathcal{S}}|v(s)|\leq \frac{\bar{C}}{1-\gamma},$ and some $B_{\sigma}>0.$
\end{assumption}

\begin{example}[Expectation]
\label{assum1_e}
    If $\sigma(v,m)=\sum_{s'\in\mathcal{S}}v(s')m(s'),$ we have 
    \begin{equation}
         \begin{aligned}
        \left|\sigma(v,m_1)-\sigma(v,m_2)\right| \leq \sum_{s'\in\mathcal{S}} v(s')\left|m_1(s')-m_2(s')\right|\leq \frac{\bar{C}}{1-\gamma} \sum_{s'\in\mathcal{S}}\left|m_1(s')-m_2(s') \right|.
    \end{aligned} \notag
     \end{equation}
\end{example}

\begin{example}[$\mathrm{CVaR}_{\alpha}$]
\label{assum1_cvar}
     If $\sigma(v,m)=\min_{y\in\mathbb{R}}\left\{y+\frac{1}{\alpha}\sum_{s'\in\mathcal{S}}(v(s')-y)^+m(s')\right\},$ we have
     \begin{equation}
         \begin{aligned}
        \left|\sigma(v,m_1)-\sigma(v,m_2)\right| &\leq \frac{2}{\alpha} \sup_{y\in\mathbb{R}}  \sum_{s'\in\mathcal{S}}(v(s')-y)^+|m_1(y)-m_2(y)|\\
        &\leq \frac{2\bar{C}}{\alpha(1-\gamma)} \sum_{s'\in\mathcal{S}}\left|m_1(s')-m_2(s') \right|.
    \end{aligned}\notag
     \end{equation}
\end{example}
It should be noted that in \cite{ruszczynski2010risk}'s original definition of Markov risk measures, no continuity assumption on the mapping $\sigma$ analogous to ours is imposed (even though the assumption of coherence implies $\sigma$ is continuous w.r.t. its first argument). This stems from the fact that in their framework (as with most other risk-sensitive RL literature), the transition probability is treated as fixed vectors. Assumption \ref{ass_sigma} is mild and holds for common risk measures including expectation (Example \ref{assum1_e}) and CVaR (Example \ref{assum1_cvar}). Under Assumption \ref{ass_sigma}, Theorem \ref{theorem_2} demonstrates that the bound of the value difference between the true transition probability and the posterior $\chi$ is dominated by the accuracy of $\chi,$ which is evaluated by $\max_{s\in\mathcal{S},a\in\mathcal{A}} \beta_{p\sim\chi} \left(\sum_{s\in\mathcal{S}}\left|p(s'|s,a)-\bar{q}(s'|s,a)\right|\right).$ 
\begin{theorem}
\label{theorem_2}
For any posterior $\chi$ and $\pi\in\Pi,$ 
\begin{equation}
\label{error_est}\begin{aligned}
       \left\Vert V_{\chi,\pi} - V_{\delta_{\bar{q}},\pi}\right\Vert_{\infty} \leq \frac{B_{\sigma}}{1-\gamma} \max_{s\in\mathcal{S},a\in\mathcal{A}} \beta_{p\sim\chi} \left(\sum_{s'\in\mathcal{S}}\left|p(s'|s,a)-\bar{q}(s'|s,a)\right|\right).
\end{aligned}\end{equation}
Furthermore, the conclusion still holds even if the left hand side of \eqref{error_est} is replaced by $ \Vert V^*_{\chi} - V^*_{\delta_{\bar{q}}}\Vert_{\infty}.$
\end{theorem}

\begin{corollary}
	\label{corol2} For any posterior $\chi$ and $\pi\in\Pi,$ initial state distribution $\mu_0$ and initial action distribution $\tau_0,$  $|\text{Risk}(\delta_{\bar{q}},\pi^*_{\bar{q}},\mu_0,\pi^*_{\bar{q}}) - \text{Risk}(\delta_{\bar{q}},\pi^*_{\chi},\mu_0,\pi^*_{\chi})| \leq \frac{2B_{\sigma}}{1-\gamma} \max_{s\in\mathcal{S},a\in\mathcal{A}} \beta_{p\sim\chi} \left(\sum_{s'\in\mathcal{S}}\left|p(s'|s,a)-\bar{q}(s'|s,a)\right|\right),$ almost surely.
\end{corollary}

\section{Bayesian Dynamic Programming}
In the last section, the posterior $\chi$ is treated as fixed and we derive the Bellman equation under a given posterior  $\chi$ along with its error control. 
In this section, we consider a model-free algorithm combining Bellman-equation-based iterative learning with adaptive posterior updates. The method adopts a Bayesian learning framework, where the knowledge of the transition dynamics is updated as training progresses. At each stage, the algorithm optimizes a double-layered risk objective based on the current uncertainty in the transition probabilities, thus providing robustness within limited interactions. Over a long term, the posterior distribution converges to the true transition kernel in the training environment, ultimately achieving near-optimality in the training environment.

We start by showing how we update the posterior using Bayes' theorem. From now on, we assume that the prior $\chi$ has a positive p.d.f. $f_{\chi}$ w.r.t. Lebesgue measure $\nu$ on $\mathscr{P}.$ For a prior $\chi,$ when we have observed $x_{t:T}=(s_t,a_t,s_{t+1},\ldots,s_{T-1},a_{T-1},s_{T})$ from environment $X^{\bar{q},\pi,{\mu}_0,\pi}$ for some $\pi$ and $\mu_0,$ we can calculate the p.d.f. of the posterior $\chi|x_{t:T}$ using Bayes' theorem as
\begin{equation}\begin{aligned}
 f_{\chi|x_{t:T}}(p) 
     &\propto f_{\chi}(p)\prod_{\tau=t}^{T-1}\left.p\left(s_{{\tau}+1}\right|s_{\tau},a_{\tau}\right).
 \end{aligned}\notag\end{equation}

We propose a Bayesian dynamic programming (Bayesian DP) process. In this proposed process, we first randomly initialize $\chi_{(0)},$ $\hat{\pi}^{*}_{{(0)}}$ and $s_{(0)}.$ At the beginning of stage $u\ (u\geq 1),$ there are posterior $\chi_{(u-1)},$ policy $\hat{\pi}^{*}_{{(u-1)}}$ and the last state $s_{(u-1)}.$ We get observations containing $\Delta_{(u)}$ actions under $\mathbb{P}^{\bar{q},\hat{\pi}^{*}_{{(u-1)}},\delta_{s_{(u-1)}},\hat{\pi}^{*}_{{(u-1)}}},$ i.e., executing policy $\hat{\pi}^{*}_{{(u-1)}}$ for $\Delta_{(u)}$ steps starting from state $s_{(u-1)},$ and we denote the observations by $x_{(u)}$. Then we update the posterior $\chi_{(u)} = \chi_{(u-1)}|x_{(u)}$ and then iteratively execute
\begin{equation}\begin{aligned}
    &\hat{V}_{(u)}^0(s) = \hat{V}^{*}_{{(u-1)}}(s),\\
    &\hat{Q}_{(u)}^k(s,a) = \hat{\mathcal{J}}_{{\chi}_{(u)},\delta_a} \hat{V}^{k-1}_{{(u)}}(s),k\geq 1,\\
    &\hat{V}^{k}_{{(u)}}(s) =   \min_{a\in\mathcal{A}}\hat{Q}_{(u)}^k(s,a),k\geq 1,
\end{aligned}\notag\end{equation}
until $\left\Vert\hat{V}^{k_u-1}_{{(u)}}-\hat{V}^{k_u}_{{(u)}}\right\Vert_{\infty}<\theta$ for some $k_u\geq 1,$ where the tolerance $\theta>0$ is fixed. Then let $\hat{V}^{*}_{{(u)}}=\hat{V}^{k_u}_{{(u)}},$ and $\hat{\pi}^{*}_{{(u)}} (a|s) = 1-(1-\frac{1}{|\mathcal{A}|})\epsilon_{(u)}$ if $a=\arg\min_{a\in\mathcal{A}}\hat{Q}_{(u)}^{{k}_{u}}(s,a)$ and $\frac{\epsilon_{(u)}}{|\mathcal{A}|}$ otherwise. This Bayesian DP algorithm is summarized in Algorithm \ref{alg:bayes-dp}.

\begin{algorithm}[htb] \caption{Bayesian Dynamic Programming (Bayesian DP)} \label{alg:bayes-dp} \begin{algorithmic}[1] \STATE \textbf{Input:} initial posterior $\chi_{(0)}$, initial policy $\hat{\pi}^{*}_{(0)}$, initial state $s_{(0)}$; tolerance $\theta>0$; stage lengths $\{\Delta_{(u)}\}$; exploration schedule $\{\epsilon_{(u)}\};$ number of stages $L.$ \STATE \textbf{Initialize:} set $u \gets 1$, set $\hat{V}^{*}_{(0)}$ arbitrarily (or from prior). \FOR{stage $u=1,2,\ldots,L$} \STATE \textbf{Rollout/observe:} starting from $s_{(u-1)}$, execute $\hat{\pi}^{*}_{(u-1)}$ for $\Delta_{(u)}$ steps under $\mathbb{P}^{\bar{q},\,\hat{\pi}^{*}_{(u-1)},\,\delta_{s_{(u-1)}},\,\hat{\pi}^{*}_{(u-1)}}$; collect observations $x_{(u)}$ and last state $s_{(u)}$. \STATE \textbf{Posterior update:} $\chi_{(u)} \gets \chi_{(u-1)} \mid x_{(u)}$. \STATE \textbf{Value iteration initialize:} $\hat{V}^{0}_{(u)}(s) \gets \hat{V}^{*}_{(u-1)}(s)$; $k \gets 1$. \REPEAT \STATE \textbf{Q-update:} $\displaystyle \hat{Q}^{k}_{(u)}(s,a) \gets \hat{\mathcal{J}}_{\chi_{(u)},\,\delta_a}\big(\hat{V}^{\,k-1}_{(u)}\big)(s)$ \ \ for all $s,a$ using Algorithm \ref{alg:estimate-J}. \STATE \textbf{V-update:} $\displaystyle \hat{V}^{k}_{(u)}(s) \gets \min_{a\in\mathcal{A}} \hat{Q}^{k}_{(u)}(s,a)$ \ \ for all $s$. \STATE $k \gets k+1$. \UNTIL{$\left\Vert \hat{V}^{k-1}_{(u)} - \hat{V}^{k-2}_{(u)} \right\Vert_{\infty} < \theta$} \STATE Set $\hat{V}^{*}_{(u)} \gets \hat{V}^{k-1}_{(u)}$ \ (let $k_u \gets k-1$). \STATE \textbf{Policy update ($\epsilon$-greedy):} \STATE For each $s$, let $a^{*}(s) \in \arg\min_{a\in\mathcal{A}} \hat{Q}^{{k}_{u}}_{(u)}(s,a)$. \STATE Define \[ \hat{\pi}^{*}_{(u)}(a\mid s) = \begin{cases} 1-\bigl(1-\tfrac{1}{|\mathcal{A}|}\bigr)\epsilon_{(u)}, & \text{if } a=a^{*}(s),\\[4pt] \tfrac{\epsilon_{(u)}}{|\mathcal{A}|}, & \text{otherwise.} \end{cases} \] \STATE \textbf{Proceed:} $u \gets u+1$ and repeat. \ENDFOR \STATE \textbf{Output:} updated value $\hat{V}^{*}_{(u-1)}$ and policy $\hat{\pi}^{*}_{(u-1)}$ at termination. \end{algorithmic} \end{algorithm}

A key distinction between the framework proposed in this paper and \cite{wang2023bayesian} lies in our adoption of an $\epsilon$-greedy policy scheme. This is because our work focuses on on-policy RL tasks, where exploratory behaviors must be incorporated to prevent the agent from converging to suboptimal local policies and to discover potentially superior policies. More discussion about this can be found in \cite{sutton1998reinforcement}.

\subsection{Estimator for Bellman Operator}

A key step in the Bayesian DP process is to estimate Bellman operator $\mathcal{J}_{\chi,\delta_a}$ properly. In most existing works, the construction of the estimator is based on unbiasedness, i.e.,  $\mathbb{E} \hat{\mathcal{J}}_{\chi,\delta_a}V(s) = \mathcal{J}_{\chi,\delta_a}V(s).$ In conventional RL algorithms without risk-sensitivity, an unbiased estimator is easy to obtain since the Bellman operator is actually an expectation operator which is linear. Some examples are Q-learning (\cite{watkins1989learning}) and traditional Bayesian RL (\cite{rieder1975bayesian} and \cite{strens2000bayesian}). For the case where the optimization objective is a non-linear risk measure, \cite{wang2023bayesian} provide estimators for VaR and CVaR. These estimators are based on empirical distribution quantiles. However, their method does not have applicability for other coherent risk measures. In this paper, we propose estimators for the Bellman operators addressing a wide range of coherent risk measures (both $\rho$ and $\beta$). Our approach is based on the the concept of risk envelope using the Monte Carlo simulation and the convex optimization. 
Our estimator achieves strong consistency, surpassing conventional requirements of unbiasedness or asymptotic unbiasedness. 

To begin with, we introduce the risk envelope of a risk measure following \cite{delbaen2002coherent}. Given a set of random variable $\mathcal{X}$ w.r.t. probability measure $\mathbb{P},$ for any coherent risk measure $f:\mathcal{X}\to \mathbb{R},$ it holds that
$
    f(X) = \sup_{\mu\in\mathcal{C}_f}\mathbb{E}^{\mu} X=\sup_{\mu\in\mathcal{C}_f}\int X\mu(X)\mathrm{d}\mathbb{P},\ \forall X\in\mathcal{X}.
$
Here, $\mathcal{C}_f$ is called the risk envelope of the risk measure $f.$
 
\begin{assumption}[Risk envelope of $\rho$]
\label{AssRho}
The risk envelope of $\rho_{p,\pi,t}$ can be written as 
\begin{equation}
\mathcal{U}(m) = \left\{ 
\xi \in \mathscr{Y}_0 \,\middle|\, 
\begin{array}{l}
\sum_{s'} \xi(s') m(s') = 1,\quad \xi(s') \geq 0,\ \forall s'\in\mathcal{S}, \\
\xi(s') + f_{s'}(h,m) = 0,\ \forall s'\in\mathcal{S}, \\
g_i(h,m) \leq 0,\ \forall i \in \mathcal{I}
\end{array}
\right\}\notag
\end{equation}
and thus $
    \rho_{p,\pi,t}(v) = \max_{\xi \in\mathcal{U}(p(\cdot|S_t,A_t))} \sum_{s'\in\mathcal{S}}\xi(s')p(s'|S_t,A_t)v(s'),
    $ for any $p\in\mathscr{P},t\in\mathcal{T},\pi\in\Pi,$ and $v\in\mathscr{V}_0$ measurable w.r.t. $\mathscr{F}_{t+1}.$ 
Here $f_{\omega}$ are affine w.r.t. $h$ and $g_i$ are convex w.r.t. $h,$ and $\mathcal{I}$ is finite. 
\end{assumption}

\begin{assumption}[Risk envelope of $\beta$]
\label{AssBeta}
 The risk envelope of $\beta_{p\sim\chi}$ can be written as \begin{equation}
\mathcal{V}(\chi) = \left\{
\mu \in \mathscr{Y}_1 \,\middle|\,
\begin{array}{l}
\int_{\mathscr{P}} \mu(p) \, \mathrm{d}F_\chi(p) = 1,\quad \mu(p) \geq 0,\ \forall p \in \mathscr{P}, \\
w_k(\mu(p)) \leq 0,\ \forall p \in \mathscr{P},\ \forall k \in \mathcal{K}, \\
\int_{\mathscr{P}} g_e(\mu(p)) \, \mathrm{d}F_\chi(p) \leq 0,\ \forall e \in \mathcal{E}
\end{array}
\right\}\notag
\end{equation}
and thus for any $v\in \mathscr{V}_1,$ $
    \beta_{p\sim\chi}(v) = \max_{\mu\in\mathcal{V}(\chi)} \int_{\mathscr{P}} \mu(p) v(p) \mathrm{d}F_{\chi}(p).
$ Here, $w_k$ and $g_e$ are convex w.r.t. $\mu,$ and $\mathcal{K},$ $\mathcal{E}$ are finite. 

\end{assumption}

There are several remarks on Assumptions \ref{AssRho} and \ref{AssBeta}. First, it is worth noting that $\rho$ is based on a discrete probability distribution while $\beta$ is based on a continuous probability measure. Second, 
Assumption \ref{AssRho} is a more generalized form of the assumption used by \cite{tamar2016sequential} and \cite{coache2024reinforcement}. As noted by \cite{tamar2016sequential},``\textit{all coherent risk measures we are aware of in the literature are already captured by [that] risk envelope}." Our assumption is strictly broader, and in particular accommodates additional examples such as the semi-deviation risk measure. Third, while the assumption used by \cite{tamar2016sequential} and \cite{coache2024reinforcement} is limited to the discrete setting, Assumption \ref{AssBeta} can be regarded as its natural extension to continuous probability spaces. Such an extension is necessary because, although the MDP is finite and discrete, the posterior distribution of the transition probability is typically continuous, which also poses challenges for estimating the Bellman operator. Our assumption provides a reasonable extension, and it also includes the widely used risk measure $\mathrm{CVaR}$, which is also an important example in their work.

 Now we show how we estimate $\mathcal{J}_{\chi,\delta_a}V(s)$ for all $(s,a)\in\mathcal{S}\times\mathcal{A}$ and the algorithm is summarized in Algorithm \ref{alg:estimate-J}. We sample $p_1,p_2,\ldots,p_N$ independently from distribution $F_{\chi},$ and denote $(\mathbb{P}^{\chi})^{\infty}$ by  $\mathbb{P}^{\text{sample}}.$ For each $p_i$ and $(s,a)\in\mathcal{S}\times\mathcal{A},$ we first calculate $\sigma(c(s,a,\cdot)+\gamma V(\cdot),p_i(\cdot|s,a))$ by solving the convex optimization problem: \begin{equation}
\label{eq:sigma_opt}
\min_{h} \sum_{s'\in\mathcal{S}} f_{s'}(h,p) p(s') \big(c(s,a,s') + \gamma V(s')\big)
\quad \text{s.t.} 
\begin{cases}
\sum_{s'\in\mathcal{S}} f_{s'}(h,p)p(s') + 1 = 0, \\
f_{s'}(h,p) \leq 0, \quad \forall s'\in \mathcal{S}, \\
g_i(h,p) \leq 0, \quad \forall i \in \mathcal{I}.
\end{cases}\notag
\end{equation} Then we estimate $\mathcal{J}_{\chi,\delta_a}V(s)$ by solving the convex optimization problem: 
\begin{equation}
\label{opt_goal}
\min_{\hat{\mu}} -\frac{1}{N} \sum_{i=1}^N \hat{\mu}(p_i) \cdot \sigma\big(c(s,a,\cdot) + \gamma V(\cdot),\, p_i(\cdot|s,a)\big)
\quad \text{s.t.}
\begin{cases}
\frac{1}{N} \sum_{i=1}^N \hat{\mu}(p_i) = 1, \\
\hat{\mu}(p_i) \geq 0, \quad \forall\, i, \\
w_k(\hat{\mu}(p_i)) \leq 0, \quad \forall\, i,\, k \in \mathcal{K}, \\
\frac{1}{N} \sum_{i=1}^N g_e(\hat{\mu}(p_i)) \leq 0, \quad \forall\, e \in \mathcal{E}.
\end{cases}
\end{equation}
\begin{algorithm}[htb]
  \caption{Estimating $\mathcal{J}_{\chi,\delta_a}V(s)$ via Sampling and Convex Programs}
  \label{alg:estimate-J}
  \begin{algorithmic}[1]
    \STATE \textbf{Input:} posterior distribution function $F_{\chi}$, value function $V$, sample size $N$, constraint functions $f_{s'}(\cdot,\cdot)$, $g_i(\cdot,\cdot)$ for $i\in\mathcal{I}$, $w_k(\cdot)$ for $k\in\mathcal{K}$, $g_e(\cdot)$ for $e\in\mathcal{E}$.
    \STATE \textbf{Output:} estimates $\widehat{\mathcal{J}}_{\chi,\delta_a}V(s)$ for all $(s,a)\in\mathcal{S}\times\mathcal{A}$.

    \STATE \textbf{Sampling:} draw $p_1,\ldots,p_N \overset{\text{i.i.d.}}{\sim} F_{\chi}.$

    \FOR{each $(s,a)\in\mathcal{S}\times\mathcal{A}$} \label{line:loop-sa}
      \STATE \textbf{Inner convex programs:}
      \FOR{$i=1$ to $N$}
        \STATE Solve the convex optimization
        \[
          \sigma_i(s,a) \;:=\;
          \min_{h}\ \sum_{s'\in\mathcal{S}} f_{s'}(h,p_i)\,p_i(s'|s,a)\,\big(c(s,a,s')+\gamma V(s')\big)
        \]
        \[
        \text{s.t.}\quad
        \sum_{s'\in\mathcal{S}} f_{s'}(h,p_i)\,p_i(s'|s,a) + 1 = 0,\quad
        f_{s'}(h,p_i)\le 0\ \forall s'\in\mathcal{S},\quad
        g_j(h,p_i)\le 0,\quad \forall j\in\mathcal{I}.
        \]
        \STATE Store $\sigma_i(s,a)$.
      \ENDFOR

      \STATE \textbf{Outer convex program:}
      \STATE Solve for weights $\hat{\mu}^{\,s,a}(p_i)\ge 0$:
      \[
        \widehat{\mathcal{J}}_{\chi,\delta_a}V(s)
        \ :=\ \min_{\{\hat{\mu}(p_i)\}_{i=1}^N}\;
        -\frac{1}{N}\sum_{i=1}^N \hat{\mu}(p_i)\cdot \sigma_i(s,a)
      \]
      \[
        \text{s.t.}\quad
        \frac{1}{N}\sum_{i=1}^N \hat{\mu}(p_i)=1,\quad
        w_k(\hat{\mu}(p_i))\le 0\ \ \forall i,\ k\in\mathcal{K},\quad
        \frac{1}{N}\sum_{i=1}^N g_e(\hat{\mu}(p_i))\le 0\ \ \forall e\in\mathcal{E}.
      \]
      \STATE Let the optimal objective value be the estimate $\widehat{\mathcal{J}}_{\chi,\delta_a}V(s)$ for this $(s,a)$; record the optimizer $\hat{\mu}^{\,s,a}$ if needed.
    \ENDFOR

  \end{algorithmic}
\end{algorithm}

To establish the strong consistency of the estimator, we introduce the tool of equal-measure partition, whose existence is guaranteed by results in semi-discrete optimal transport. Based on this tool, we derive Lemma \ref{important_lem} and Proposition \ref{important_prop}, which together are used to prove the strong consistency of the estimator (Theorem \ref{theorem_convergenceBE}).  
We consider an $N$-partition of $\mathscr{P}$ based on the first $N$ samples. The $i$-th region is defined as  
$
D_i=\Bigl\{p\in\mathscr{P}:\ \Vert p-p_i\Vert-w_i \leq \Vert p-p_j\Vert-w_j,\ \forall j \Bigr\},
$
where $(w_1,w_2,\ldots,w_N)\in\mathbb{R}^N$ satisfies $\mathbb{P}^{\chi}(D_i)=\tfrac{1}{N}.$ The existence of such $(w_1,\ldots,w_N)$ is guaranteed by results on the Laguerre tessellation (see Theorem 1 in \cite{geiss2013optimally} and Theorem 2.31 in \cite{dieci2024solving}). This construction corresponds precisely to the dual problem of a semi-discrete optimal transport, i.e., transporting $\mathbb{P}^{\chi}$ to a discrete uniform distribution; see \cite{merigot2021optimal}.  
Furthermore, since we choose $\Vert\cdot\Vert$ instead of $\Vert\cdot\Vert_2$ as the cost function in our Laguerre tessellation, we have $p_i\in D_i$ (see Lemma 2.10 in \cite{dieci2024solving}). Given $\hat{\mu}_N(p_i)\in\mathbb{R}$ for $1\leq i \leq N$, we define  
$
\tilde{\mu}_N(p) = \sum_{i=1}^{N} \hat{\mu}_N(p_i)\,\mathds{1}_{D_i}(p).
$ 
It then follows that $\tilde{\mu}_N\in\mathscr{Y}_1^{\chi}.$

\begin{lemma}
\label{important_lem}
As $N$ goes infinity, $\max_{1\leq i\leq N}\text{diam}(D_i)\to0,$  $\mathbb{P}^{\text{sample}}$-almost surely.
\end{lemma}

\begin{proposition}
\label{important_prop}
 Define  \begin{equation}\begin{aligned}
     \tilde{\mathcal{V}}_N (\chi)=  \Bigg\{\mu\in\mathscr{Y}_1:&{\mu}(p) = \sum_{i=1}^{N} \hat{\mu}(p_i)\mathds 1_{D_i},\\
     &\frac{1}{N}\sum_{i=1}^N\hat{\mu}(p_i)=1,\hat{\mu}(p_i)\geq 0,\forall 1\leq i\leq N,\\
    &w_k(\hat{\mu}(p_i))\leq 0,\forall 1\leq i\leq N,k\in\mathcal{K},\\
    &\frac{1}{N}\sum_{i=1}^N g_{e}(\hat{\mu}(p_i)) \leq 0,\forall e\in\mathcal{E}
    \Bigg\},
 \end{aligned}\end{equation}
 and the following conclusions hold $\mathbb{P}^{\text{sample}}$-almost surely:
\begin{itemize}

    \item[(1)] $\tilde{\mathcal{V}}_N (\chi)\subset{\mathcal{V}} (\chi),$ for any $N\geq 1.$ Therefore, $\tilde{\mathcal{V}}_N$ are uniformly $\Vert\cdot\Vert_{m_1}$-bounded, i.e.,
    \begin{equation}
        \sup_{N\geq 1}\sup_{\mu\in\tilde{\mathcal{V}}_N (\chi)} \Vert{\mu}\Vert_{m_1}< \infty;
    \end{equation}
     \item[(2)] Denote by $W_N$ the optimal value of  \eqref{opt_goal}. Then for any $(s,a)\in\mathcal{S}\times\mathcal{A},$ \begin{equation}\begin{aligned}
    \lim_{N\to\infty} \Bigg|W_N -\max_{\tilde{\mu}_N\in\tilde{\mathcal{V}}_N(\chi)}\int_{\mathscr{P}}\tilde{\mu}_N(p)\cdot\sigma(c(s,a,\cdot)+\gamma V(\cdot),p(\cdot|s,a))\d F_{\chi}(p)\Bigg| = 0;
    \end{aligned}\end{equation} 
     \item[(3)] $\tilde{\mathcal{V}}_N (\chi)\xrightarrow{\Gamma}\mathcal{V}(\chi)$ in $\mathscr{L}_{m_1} $ weak topology, i.e., for any $\mu\in\mathcal{V}(\chi),$ there exists a sequence $\left\{\tilde{\mu}_{N}\right\}^{\infty}_{N=1}$ with $\tilde{\mu}_N\in\tilde{\mathcal{V}}_N(\chi)$ such that $\tilde{\mu}_N \rightharpoonup\mu$ in $\mathscr{L}_{m_1} $ weak topology, and for any sequence $\left\{\tilde{\mu}_{N}\right\}^{\infty}_{N=1}$ with $\tilde{\mu}_N\in\tilde{\mathcal{V}}_N(\chi),$ there exists a subsequence $\left\{\tilde{\mu}_{N_k}\right\}^{\infty}_{k=1}$ such that $\tilde{\mu}_{N_k} \rightharpoonup\mu$ in $\mathscr{L}_{m_1} $ weak topology for some $\mu\in\mathcal{V}(\chi).$

    \end{itemize}
\end{proposition}

\begin{theorem}[Strong Consistency of Bellman estimator]

\label{theorem_convergenceBE}
With probability $1,$ we have
\begin{equation}
        \lim_{N\to\infty} \max_{(s,a)\in\mathcal{S}\times\mathcal{A}}\left|\hat{\mathcal{J}}_{\chi,\delta_a}V(s) - \mathcal{J}_{\chi,\delta_a}V(s)\right| = 0,\notag
    \end{equation}
holds uniformly for any value function $V$ with $\Vert V\Vert_{\infty}\leq \frac{\bar{C}}{1-\gamma}.$ 
\end{theorem}

\subsection{Convergence Analysis}
Our theoretical results are established within an on-policy learning framework with $\epsilon$-greedy exploration, which presents a distinct contrast to the offline learning in \cite{wang2023bayesian}. In their work, the data is generated by an externally specified policy that cannot be improved by the agent. Consequently, their method only guarantees convergence for state-action pairs visited infinitely often, while global convergence depends on the exploratory properties of the predetermined policy. The only requirement in this paper is the irreducibility of the state space $\mathcal{S}$ (i.e., any state $s\in\mathcal{S}$ is reachable from any other state with positive probability).
\begin{assumption}
\label{ass_4}
The state space $\mathcal{S}$ is irreducible, i.e. for any $s,s'\in\mathcal{S},$ there exist $s_1,s_2,\ldots,s_n\in\mathcal{S}$ and $a_0,a_1,a_2,\ldots,a_n\in\mathcal{A}$ such that
$
   \bar{q}(s_1|s,a_0)\bar{q}(s_2|s_1,a_1)\cdots \bar{q}(s_n|s_{n-1},a_{n-1})\bar{q}(s'|s_n,a_n)>0.\notag
$
\end{assumption} Building upon this assumption, we subsequently establish the convergence of the posterior distribution in Lemma \ref{theorem_posterior} and the convergence of the whole algorithm in Theorem \ref{theorem_convergenceALL}.
\begin{lemma}[Convergence of posterior]
\label{theorem_posterior}
  If $\inf_{u\geq 1}\epsilon_{(u)} >0 ,$  $\chi_{(u)}\to {\delta_{\bar{q}}}$ almost surely as $u\to\infty.$
\end{lemma}

\begin{theorem}[Convergence of Bayesian DP]
\label{theorem_convergenceALL} 
$
 \limsup_{u\to\infty,N\to\infty} \Vert \hat{V}^{*}_{{(u)}}-V^*_{\delta_{\bar{q}}}\Vert_{\infty}\leq \frac{\theta}{1-\gamma}\notag
 $ almost surely.
\end{theorem}

\subsection{Complexity Analysis for a Dirichlet Posterior and CVaR}
\label{Complexity}
Although we consider general prior and posterior distributions, the Dirichlet distribution remains a pivotal example, as widely adopted in most existing Bayesian RL works (see  \cite{strens2000bayesian}, \cite{poupart2006analytic}, \cite{asmuth2012bayesian}, \cite{osband2013more}, and \cite{wang2023bayesian}). 
The Dirichlet distribution is a continuous probability distribution defined on the $(K-1)$-dimensional simplex, parameterized by a $K$-dimensional vector $\alpha = (\alpha_1, \alpha_2, \dots, \alpha_K)$, where $\alpha_i > 0$. If a random vector $y = (y_1, y_2, \dots, y_K)$ follows a Dirichlet distribution with parameter $\alpha,$ the probability density function is defined as:
$
f(y) = \frac{\Gamma\left(\sum_{i=1}^K \alpha_i\right)} {\prod_{i=1}^K \Gamma(\alpha_i)} \prod_{i=1}^K y_i^{\alpha_i - 1},
$
where $y_i \geq 0,$  $\sum_{i=1}^K y_i = 1$ and $\Gamma(\cdot)$ denotes the Gamma function. 

For any $n_0,m_0 \in (1,\infty)$ with $\frac1{n_0}+\frac1{m_0}=1,$ we have
$
\mathcal{L}_{n_0}(\mathcal{S},2^{\mathcal{S}},\mathbb{P}_0)= \mathcal{L}_{m_0}(\mathcal{S},2^{\mathcal{S}},\mathbb{P}_0)=  \mathcal{L}_{\infty}(\mathcal{S},2^{\mathcal{S}},\mathbb{P}_0).
$
Thus, $\mathscr{M}$ is exactly the $(K-1)$-dimensional simplex. For fixed $s\in\mathcal{S}$ and $a \in \mathcal{A},$ we consider the transition probability $p\left(s,a\right)=(p\left.\left(s^1\right|s,a\right),p\left.\left(s^2\right|s,a\right),\ldots,p\left.\left(s^K\right|s,a\right))$ as a random vector supported on  the $K$-dimensional simplex and place a Dirichlet prior with parameter $\alpha\left(s,a\right)=\left(\alpha\left(s^1|s,a\right),\alpha\left(s^2|s,a\right),\ldots,\alpha\left(s^K|s,a\right)\right)$ on it. Moreover, we assume the prior for each $(s,a)$ is independent; we stack all $\alpha\left(s,a\right)$ into vector $\alpha.$ Then the prior $\chi$ can be represented as 
$
    p\sim\chi\sim \text{D}({\alpha}\left(s^1,a^1\right))\otimes \text{D}\left(\alpha(s^1,a^2\right))\otimes\cdots\otimes \text{D}({\alpha}\left(s^K,a^B\right)),
$
which we denote as $p\sim  \text{D}(\otimes \alpha)$ for short.  

When we observed $x_{t:T}=(s_t,a_t,s_{t+1},\ldots,s_{T-1},a_{T-1},s_T),$
based on prior $\chi = D(\otimes {\alpha}),$ we calculate the posterior by Bayes' formula as 
\begin{equation}\begin{aligned}
      f_{\chi|x_{t:T}}({p})
      &\propto   f_{\chi_t}({p}|{x}_{t:T})f({x}_{t:T}|{p})\\
     & \propto \prod_{s,a}\prod_{i=1}^K p\left.\left(s^i\right|s,a\right)^{\alpha(s^i|s,a)+m_{t:T}(s,a,s^i) - 1},
 \end{aligned}\end{equation}
which implies that the posterior for $p\left(s,a\right)$ follows a Dirichlet distribution with parameter $
{\alpha}\left(s,a\right)+{m}_{t:T}\left(s,a\right),
 $ where ${m}_{t:T}\left(s,a\right)=\left(m_{t:T}(s,a,s^1),\ldots,m_{t:T}(s,a,s^K)\right)$ and $m_{t:T}(s,a,s')$ denotes the number of occurrences of the state–action–next-state tuple $(s,a,s').$ We stack all ${m}_{t:T}\left(s,a\right)$ into the vector ${m}_{t:T}$ and we have $
     \chi|{x_{t:T}}\sim  D(\otimes ({\alpha}+{m}_{t:T})).
$ In Bayes DP RL process, given a prior Dirichlet parameter $\alpha_{(0)},$ we have $\chi_{(u)}\sim\text{D}(\alpha_{(u-1)}),$ where $\alpha_{(u)} = \alpha_{(u-1)}+m_{(u)}.$

From now on, we assume both the prior and the posterior on the transition probability follow a Dirichlet distribution. The prior Dirichlet parameters satisfy that $\bar{A}_0=\max_{s,a} \sum_{s'}\alpha(s'|s,a)$ remains bounded as $|\mathcal{S}|$ and $|\mathcal{A}|$ grow. The outer risk measures $\beta_{p \sim \chi}$ are chosen as $\mathrm{CVaR}_{\alpha_2}$, and the inner risk measures $\rho_{p,\pi,t}$ are chosen as $\mathrm{CVaR}_{\alpha_1}$. We also make an additional assumption that the costs are positive, i.e., $c(s,a,s')>0$ for any $(s,a,s')\in\mathcal{S}\times\mathcal{A}\times\mathcal{S}.$ Furthermore, we assume that the training environment provides a sufficient exploration coverage, as stated in Assumption \ref{policy_assum}, which guarantees that every state–action pair is visited with at least a polynomially small probability.
Building upon the preceding analysis of convergence and the additional assumptions, we now turn to the complexity analysis of the proposed algorithm. Specifically, we examine both its sample complexity and computational complexity.
\begin{assumption}
\label{policy_assum}
We assume that $\bar{q}$ satisfies the following coverage property: there exists a constant $T_0>0$ such that, for any stage-wise $\epsilon$-greedy policy $\pi$, any initial state distribution $\mu_0$, and all $(s,a)\in\mathcal{S}\times\mathcal{A}$, we have $
        \mathbb{P}^{\bar{q},\pi,\mu_0,\pi}(S_t=s,A_t = a) \geq \mu_{\min}>0,\ \forall t\geq T_0, \notag
    $ where $ 
        \mu_{\min}^{-1} = \mathcal{O}\left(|\mathcal{A}|^{\xi}|\mathcal{S}|^{\eta}\right), $ for some $\xi,\eta>0.$
\end{assumption}

\subsubsection{Sample Complexity}

In this subsection, we establish explicit bounds on the number of samples required
for the proposed learning procedure to achieve a prescribed tolerance. 
We quantify how large the total number of collected transitions $T$ must be in order to guarantee that the posterior optimal policy $\pi^*_{\chi_T}$ 
achieves achieves a risk value close to that of the oracle policy $\pi^*_{\bar{q}}$ in the training environment (Theorem \ref{thm:sample_complexity}). 

\begin{theorem}[Sample Complexity Bound]
\label{thm:sample_complexity}
It is sufficient that the total number of samples $T$ satisfies
\begin{equation}
T \;\ge\; T_0 +
\max\Bigg\{
\frac{16\bar{A}_0\bar{C}}{\mu_{\min}(1-\gamma)^2\alpha_1\alpha_2\theta}
,\,
\frac{128\,|\mathcal{S}|\,\bar{C}^2}
{\mu_{\min}(1-\gamma)^4\alpha_1^2\alpha_2^2\theta^2},\,
\frac{8}{\mu_{\min}}\ln\!\left(\frac{2|\mathcal{S}||\mathcal{A}|}{\delta}\right)
\Bigg\},
\notag
\end{equation} to guarantee that
$
\big|
\mathrm{Risk}(\delta_{\bar{q}},\pi^*_{\bar{q}},\mu_0,\pi^*_{\bar{q}})
-\mathrm{Risk}(\delta_{\bar{q}},\pi^*_{\chi_T},\mu_0,\pi^*_{\chi_T})
\big|
\leq \theta
\notag
$ with probability $1-\delta,$
\end{theorem}

\begin{corollary}[Asymptotic Sample Complexity]
\label{cor:sample_complexity_order}
The sufficient number of samples scales as
\begin{equation}
T = \mathcal{O}\!\left(|\mathcal{S}|^{\xi}\,|\mathcal{A}|^{\eta} \cdot
\frac{|\mathcal{S}|+\ln\left(\frac{|\mathcal{A}||\mathcal{S}|}{\delta}\right)}
{(1-\gamma)^4\,\theta^2} 
\right).
\notag
\end{equation}
\end{corollary}

Corollary \ref{cor:sample_complexity_order} follows by substituting $\mu_{\min}^{-1} = \mathcal{O}(|\mathcal{A}|^{\eta}|\mathcal{S}|^{\xi})$ from Assumption~\ref{policy_assum}.
The result shows that the sample complexity of the proposed algorithm scales polynomially with all relevant problem parameters, including the number of states, actions, error tolerance and the discount factor, while depending only logarithmically on the confidence level~$1/\delta$.
Consequently, this provides a provable performance guarantee of the proposed algorithm.

\subsubsection{Computational Complexity}

We analyze the computational complexity of Algorithm \ref{alg:bayes-dp} under the below implementment of $\Delta_{(u)}$.
We implement the algorithm in a sweep-based manner: every newly observed state–action pair 
$(s,a)$ defines a new stage and is subsequently treated as ``known'' until the sweep completes.
When all pairs are ``known'', the current sweep is completed, and the next sweep begins following the same procedure.
We denote by $U_L$ the initial stage of the $L$-th sweep. The algorithm terminates when, within a single sweep, all stages (corresponding to all state–action pairs) converge within one step of value iteration.
This mechanism has been widely adopted in Bayesian RL algorithms, such as in \cite{brafman2002r} and \cite{asmuth2012bayesian}. In this case, we first provide a perturbation bound on the optimal value function induced by updates of Dirichlet parameters (Proposition \ref{lemma_DPcom}), then derive the number of value iterations needed to ensure convergence in each sweep (Proposition \ref{sweep_bound}), and further establish a global iteration bound (Theorem \ref{c1}). Next, we analyze the per-iteration computational cost, particularly under the CVaR risk measures (both inner and outer), where we present a closed-form solution for the optimization problems under CVaR (Proposition \ref{cvar_closed}) together with the corresponding computational complexity (Corollary \ref{c2}). All in all, these results characterize our algorithm’s overall computational complexity.

\begin{proposition}
\label{lemma_DPcom}
If ${m}$ is the sum of $\Delta $ one-hot vectors, then 
\begin{equation}
   \left\Vert V^*_{\text D(\otimes({\alpha}+{m}))} - V^*_{\text D(\otimes{\alpha})}\right\Vert_{\infty} \leq \frac{4 \bar{C} }{\alpha_1\alpha_2}\cdot\frac{|\mathcal{S}|^2|\mathcal{A}|}{(1-\gamma)^2}\ln\left(1+\frac{\Delta }{|\mathcal{S}||\mathcal{A}|O_{\alpha}}\right),
\end{equation}
with the confidence 
$
    O_{\alpha} = \min_{s,a} \sum_{s'\in\mathcal{S}} \alpha(s'|s,a).
$
\end{proposition}

\begin{corollary}[Iteration Bound per Stage]
\label{corol_DPcom}
In the $u$-th value iteration, the convergence is guaranteed once the number of iterations satisfies 
\[
k^{(u)} \;=\; \left\lceil \frac{1}{\ln\!\left(\tfrac{1}{\gamma}\right)} 
\ln\!\left(1+ \frac{4 \bar{C} }{\alpha_1\alpha_2}\cdot\frac{|\mathcal{S}|^2|\mathcal{A}|}{\theta (1 - \gamma)^2} 
\ln\!\left(1 + \frac{\Delta_{(u)}}{|\mathcal{S}||\mathcal{A}| O_{u}} \right) \right)\right\rceil,
\]
where $O_u$ denotes the confidence associated with the Dirichlet parameters at stage $u.$
\end{corollary}

Proposition \ref{lemma_DPcom} provides a bound scaling with $(1-\gamma)^{-2}$, which seems less favorable compared to the more commonly used and straightforward bound of $\frac{2\bar{C}}{1-\gamma}$ typically scaling with $(1-\gamma)^{-1}$. However, we emphasize that our results provide an iteration count that decays with respect to the stage number $u$. This decay is necessary for determining the number of active stages and is used to estimate the global iteration count of the algorithm over multiple stages (Theorem \ref{c1}). This additional $(1-\gamma)^{-1}$ factor also arises from handling the double-layered CVaR structure, and it affects the order of the global iteration bound as well. An interesting but open question is whether this additional factor can be eliminated without increasing the order of the active stages, that is, whether it is truly necessary. We leave this for future research.

\begin{proposition}[Iteration bound in a full sweep]
\label{sweep_bound}
In the $L$-th full sweep — that is, after all state–action pairs $(s,a)$ have been traversed once in stage range $[U_L, U_{L+1}]$ — the cumulative number of iterations required for convergence satisfies
\begin{equation}
    \sum_{u=U_L}^{U_{L+1}} k^{(u)} \leq  {|\mathcal{A}||\mathcal{S}|} \left(
\frac{1}{\ln\!\left(\tfrac{1}{\gamma}\right)}\ln\!\left( 1+\frac{4 \bar{C} }{\alpha_1\alpha_2}\cdot\frac{|\mathcal{S}|^2|\mathcal{A}|}{\theta (1 - \gamma)^2} 
\ln\!\left(1 + \frac{\sum_{u=U_L}^{U_{L+1}}\Delta_{(u)}}{|\mathcal{S}|^2|\mathcal{A}|^2 (O_{0}+L)} \right) \right)+1\right).
\end{equation}
\end{proposition}

\begin{theorem}[Global iteration bound]
\label{c1}
The number of active stages is at rate $\mathcal{O}_{p}\Big(|\mathcal{S}|^{\xi+1}|\mathcal{A}|^{\eta}\cdot\frac{1}{\theta(1-\gamma)^3}\cdot\ln(|\mathcal{S}||\mathcal{A}|)\Big).$ The total number of value iterations is at rate 
\begin{equation}
{\mathcal{O}}_{p}\left(\frac{|\mathcal{S}|^{\xi+1}|\mathcal{A}|^{\eta}}{\theta(1-\gamma)^4}\cdot\ln\left(\frac{|\mathcal{S}|^{\xi+2}|\mathcal{A}|^{\eta+1}}{\theta(1-\gamma)^3}\ln\left(|\mathcal{S}||\mathcal{A}|\right)\right)\right) = \widetilde{\mathcal{O}}_{p}\left(\frac{|\mathcal{S}|^{\xi+1}|\mathcal{A}|^{\eta}}{\theta(1-\gamma)^4}\right).\notag
\end{equation}
Here, $\mathcal{O}_p$ denotes the order in probability, i.e., a probabilistic bound.
\end{theorem}

Next, we turn to the computational complexity of estimating the Bellman operator within each value iteration. In this part, we focus on the setting where both the inner and outer risk measures are chosen as $\mathrm{CVaR}$, consistent with the example presented in the introduction. The advantage of $\mathrm{CVaR}$ lies not only in its interpretability and rationality, but also in the fact that the corresponding optimization problem admits a closed-form solution, as established in Proposition \ref{cvar_closed}.

\begin{proposition}
    \label{cvar_closed}
The solution of the following problem
\begin{equation}
        \min_{h} \sum_{s'\in\mathcal{S}}  h(s')p(s') \big(c(s,a,s') + \gamma V(s')\big)
\quad \text{s.t.} 
\begin{cases}
\sum_{s'\in\mathcal{S}} h(s')p(s') + 1=0, \\
-\frac{1}{\alpha}-h(s') \leq 0,
\end{cases}\notag
    \end{equation}
is $h(s')=\frac{1}{\alpha} $ if  $c(s,a,s') + \gamma V(s')>\lambda$ and $h(s')=0$ if  $c(s,a,s') + \gamma V(s')\leq\lambda,$ where $\lambda$ is the $1-\alpha$ quantile of  $c(s,a,\cdot) + \gamma V(\cdot).$
\end{proposition}
\begin{corollary}
\label{c2}
Each value iteration incurs a computational complexity of
$$
\mathcal{O}\!\left(|\mathcal{A}|\cdot N\cdot|\mathcal{S}|\cdot\big(|\mathcal{S}|\log|\mathcal{S}|+\log N\big)\right).
$$
\end{corollary}

According to Proposition \ref{cvar_closed},
when solving the optimization problem corresponding to CVaR, the main computational cost arises from computing the quantile 
$\lambda$, which is equivalent to a sorting operation. Therefore, the computational complexities of the inner and outer optimizations are 
$\mathcal{O}(|\mathcal{S}|\log|\mathcal{S}|)$ and $\mathcal{O}(N\log N)$, respectively. Consequently, Corollary \ref{c2} follows directly, establishing the total computational complexity of each value iteration.

\section{Synthetic Experiments}

 \subsection{Problem Descriptions}
\textbf{1. Coin Toss:} 
We consider a benchmark coin-toss game that has been also adopted in the robust RL literature (see, e.g., \cite{neufeld2024robust,wang2023bayesian}). 
At each time step, the agent observes the outcomes of 10 independent coin tosses, each resulting in either heads (encoded as 1) or tails (encoded as 0). 
The state variable $S_t \in \{0,1,\ldots,10\}$ represents the total number of heads observed at time $t$. 
The agent can choose one of three possible actions $A := \{-1,0,1\}$, corresponding to betting that the next sum of heads will be smaller ($A_t=-1$), abstaining from betting ($A_t=0$), or betting that it will be larger ($A_t=1$) than the current sum. 
The cost is defined as
$
c(x,a,x') = -a \mathbf{1}_{\{x < x'\}} + a \mathbf{1}_{\{x > x'\}} + |a| \mathbf{1}_{\{x = x'\}},\notag
$
so that the agent earns one dollar for a correct prediction, loses one dollar for an incorrect prediction, and receives zero payoff when abstaining. 
The transition distribution in the training environment is modeled as a binomial law with parameters $n=10$ and $p=0.6$, that is,
\begin{equation}
\bar{q}(s,a,s') = \binom{10}{s'} 0.6^{s'} 0.4^{10-s'}.\notag
\end{equation}
Additionally, the discount factor $\gamma$ is chosen as $0.9.$

\textbf{2. Inventory Management:} 
We further consider an inventory management problem that has been widely studied in the robust RL literature (see, e.g., \cite{liu2022distributionally,neufeld2024robust,wang2023bayesian}). Our specific setup is more closely related to that in \cite{liu2022distributionally}, but here the action represents the \emph{target inventory position} the agent aims to reach. At the beginning of period $t$, the agent observes the previous period's excess-demand quantity 
$S_t \in \{-n, \ldots, 0, \ldots, n\}$, where the current on-hand 
inventory is given by $(S_t)^+$. The agent then selects a target inventory level $A_t \in \{0,1,\ldots,n\}$, 
and the actual order quantity is given by $(A_t - (S_t)^+)^+$. Each unit ordered incurs an ordering cost of $k$. At the end of the period, a random demand $D_t\in\{0,1,\ldots,n\}$ is realized and results in an excess-demand quantity $S_{t+1} = (S_t)^++(A_t - (S_t)^+)^+-D_t.$ 
Each unit of positive (long) excess-demand incurs a holding cost of $h$.
 Each unit of the negative (short) excess-demand incurs an unmet-demand penalty of $p$. Formally, the cost function is defined as
$
c(s,a,s') = \big(k \mathds{1} \left\{(a-s^+)^+>0\right\} + h (s')^+ + p (s')^- \big).
$
In the training environment, the demand $D_t$ follows a uniform distribution on $\left\{0,1,\ldots,n\right\},$ i.e.,
\begin{equation}
\bar{q}(s'|s,a)
= \mathbb{P}\big( s^+ + (a-s^+)^{+} - D_t = s' \big), \ D_t\sim U(\left\{0,1,\ldots,n\right\}),\notag
\end{equation}
which is consistent with \cite{liu2022distributionally}.  
We use $n=10$, $k=3$, $h=1$, $p=2$, $\gamma=0.9$ and an initial inventory $S_0=0$.

\subsection{Results}

\subsubsection{Risk-sensitivity analysis and oracle optimal policies}

To illustrate the connection between inner risk measures and risk sensitivity, we present the oracle optimal policies induced by different inner-risk specifications under the training transition probabilities.
 In particular, we demonstrate and analyze the oracle optimal policies $\pi^*_{\delta_{\bar{q}}}$ associated with various choices of the inner risk measure $\rho$. Tables \ref{table2} and \ref{table1} respectively present the optimal policies under different inner risk measures for both problems when using the true transition probabilities in the training environment. A key insight emerges: more conservative inner risk measures yield increasingly cautious policies.

\begin{table}[h!]
\centering
\caption{Oracle Optimal Policies in Coin Toss}
\label{table2}
\begin{tabular}{
    >{\centering\arraybackslash}m{2.5cm}
    >{\centering\arraybackslash}m{1.8cm}
    >{\centering\arraybackslash}m{1.8cm}
    >{\centering\arraybackslash}m{1.5cm}
    >{\centering\arraybackslash}m{1.5cm}
    >{\centering\arraybackslash}m{1.5cm}
    >{\centering\arraybackslash}m{1.8cm}
}
\toprule
    & $0\leq s\leq 1$ & $2\leq s\leq 4$& $s=5$& $s=6$&$s=7$&$8\leq s\leq 10$\\
\midrule
Expectation& $1$& $1$& $1$& $0$& $-1$& $-1$\\
$\mathrm{CVaR}_{0.5}$& $1$& $1$& $0$& $0$&$0$ & $-1$\\
$\mathrm{CVaR}_{0.2}$& $1$& $0$& $0$& $0$& $0$& $-1$\\
\bottomrule
\end{tabular}
\end{table}

\begin{table}[h!]
\centering
\caption{Oracle Optimal Policies in Inventory Management}
\label{table1}
\begin{tabular}{
    >{\centering\arraybackslash}m{2.5cm}
    >{\centering\arraybackslash}m{2.5cm}
    >{\centering\arraybackslash}m{1.5cm}
    >{\centering\arraybackslash}m{1.5cm}
    >{\centering\arraybackslash}m{1.5cm}
    >{\centering\arraybackslash}m{1.5cm}
    >{\centering\arraybackslash}m{1.5cm}
}
\toprule
    & $-10\leq s\leq 1$&$s=2$& $s=3$& $s=4$&$s=5$ &$s\geq5$\\
\midrule
Expectation& $8$& $8$& $s$& $s$& $s$&$s$\\
$\mathrm{CVaR}_{0.5}$& $8$& $8$& $8$& $8$& $s$&$s$\\
$\mathrm{CVaR}_{0.2}$& $8$& $8$& $8$& $8$& $8$&$s$\\
KL-DRRL& $7$& $7$& $7$& $7$& $s$&$s$\\
Wass-DRRL& $8$& $7$& $8$& $s$& $s$&$s$\\
\bottomrule
\end{tabular}
\end{table}

In the Coin Toss problem, the optimal risk-neutral policy in the training environment uses $s = 6$ as the threshold: when the number of heads is greater than $6$, the agent guesses fewer (i.e., chooses action $-1$); when it is less than $6$, it guesses more (i.e., chooses action $1$); and when it is exactly $6$, it chooses not to guess (i.e., chooses action $0$). As the inner risk measure becomes more conservative, the frequency of abstaining from guessing increases. Under $\mathrm{CVaR}_{0.2},$ the agent will only make a guess when there is less than one head or more than eight heads. This is because in certain intermediate cases, although the expected return of making a guess is positive, it also introduces a probability of incurring a loss, which increases the tail risk. Under a CVaR-based preference, such losses tend to be avoided. In the Inventory Management problem, 
the optimal risk-neutral policy in the training environment is to replenish up to $8$ units when the inventory level is less than or equal to $2$, and to keep the inventory unchanged when it is above $2$. This is consistent with the results reported in \cite{liu2022distributionally} and \cite{neufeld2024robust}. As the inner risk measure becomes more conservative, the replenishment threshold increases in order to avoid the tail risk caused by unmet demand. Under $\mathrm{CVaR}_{0.2},$ replenishment begins once the inventory drops to $5$ or below. 
We also compare our results in Table \ref{table1} with the oracle optimal policies of the two DRRL methods, which are reported in \cite{liu2022distributionally} and \cite{neufeld2024robust}. KL-DRRL refers to the DRRL method based on KL-divergence ambiguity sets (\cite{liu2022distributionally}), while Wass-DRRL refers to the DRRL method based on Wasserstein ambiguity sets (\cite{neufeld2024robust}).
 An important distinction is that changes in the inner risk measure only affect the replenishment threshold but not the replenishment quantity, whereas the two DRRL methods behave differently. This is because in the training environment, once replenishment occurs, the target inventory of $8$ is optimal under both risk-neutral and tail-risk perspectives, and only the fixed cost needs to be considered. In contrast, DRRL accounts for model mis-specification rather than tail risk, and when the transition probability change, the optimal inventory level changes accordingly. This also illustrates the difference between risk sensitivity and robustness.

\subsubsection{Algorithm Convergence}

Next, we evaluate the convergence of the proposed Bayesian DP algorithm in the two problems. We conduct experimental studies for both cases of the inner risk measure, namely 
$
\rho=\mathrm{Mean}$ (Mean preference) and 
$
\rho=\mathrm{CVaR}_{0.5}$ (CVaR preference).
Under both preference settings, we conduct experiments by selecting the outer risk measure as either $\mathrm{Mean}$ or $\mathrm{CVaR}_{0.6}$
. Under the Mean preference, we compare our method with KL-DRRL, Wass-DRRL, and traditional Q-learning; under the CVaR preference, we compare it with iterated CVaR RL (see \cite{du2022provably}).
We evaluate performance using the average of the value function—computed under the stationary distribution—corresponding to the updated policy at each stage. All the results reflect the performance of each model within a single episode, but to mitigate randomness, the results for each model are averaged over 50 independent runs. Each model interacts with the training environment for 2000 steps, with every 100 steps forming a stage. Our model performs multiple iterations at the beginning or end of each stage, while the other models perform one iteration at every step. The results are shown in Figure \ref{fig:Converge}.

\begin{figure}[htb]
    \centering

    \begin{subfigure}{0.45\textwidth}
        \centering
        \includegraphics[width=\textwidth]{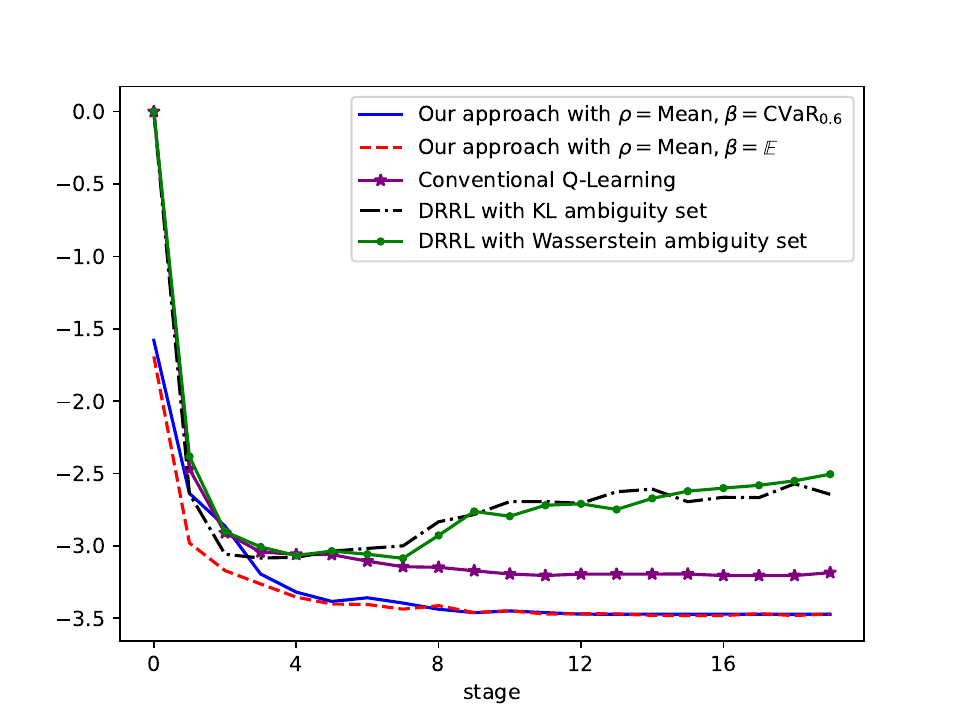}
        \caption{Coin Toss\\
        with Mean preference}
    \end{subfigure}
    \hspace{0.05\textwidth}
    \begin{subfigure}{0.45\textwidth}
        \centering
        \includegraphics[width=\textwidth]{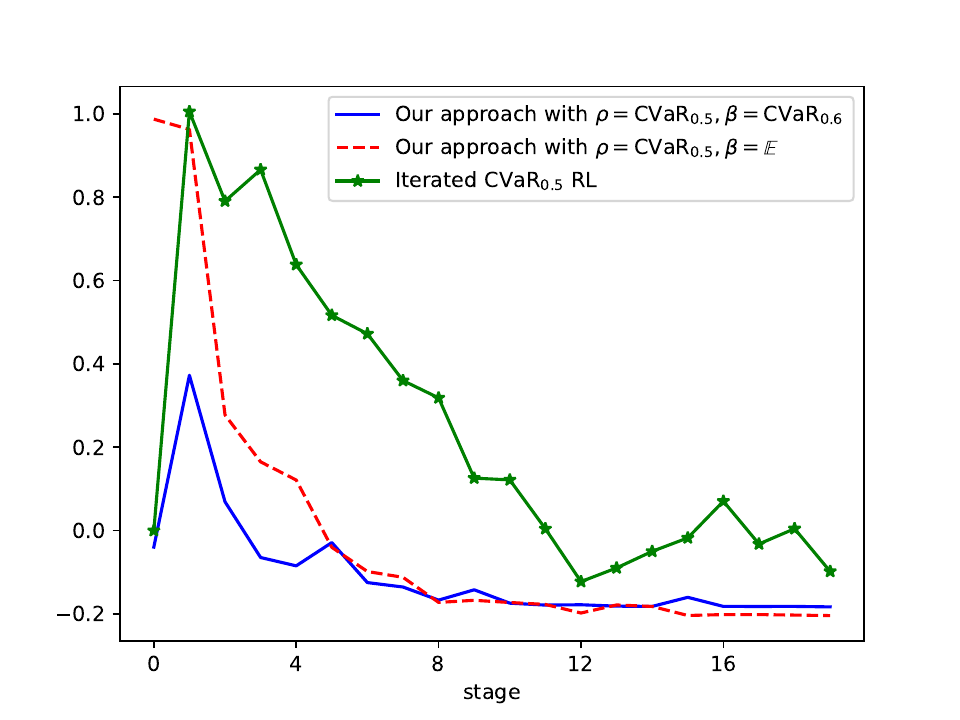}
        \caption{Coin Toss\\
        with CVaR preference}
    \end{subfigure}

    \vspace{0.05\textwidth}

    \begin{subfigure}{0.45\textwidth}
        \centering
        \includegraphics[width=\textwidth]{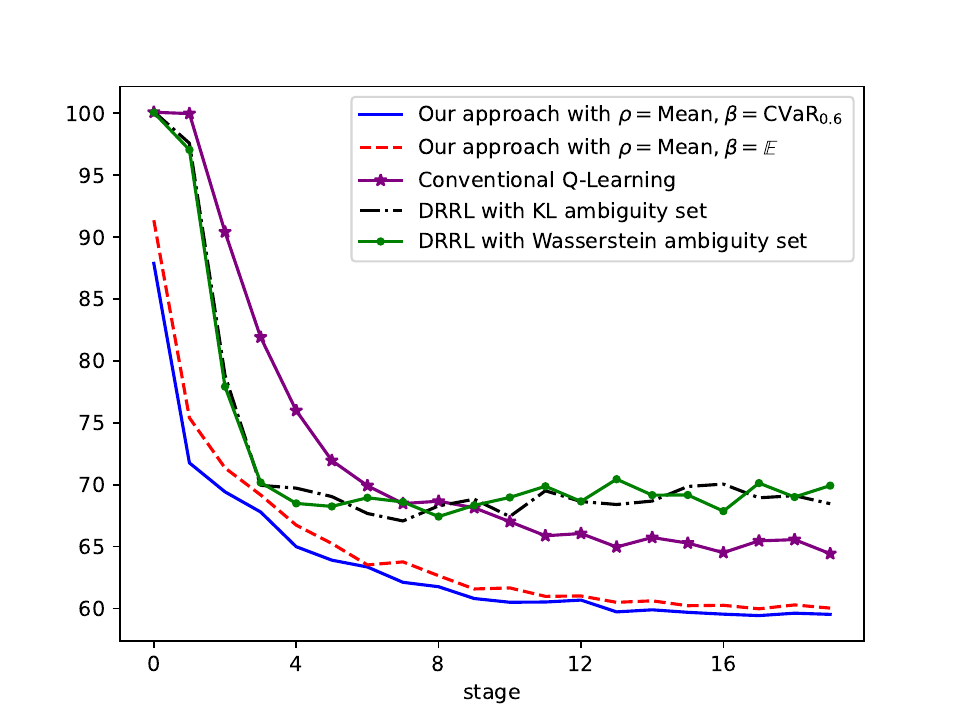}
        \caption{Inventory Management\\
        with Mean preference}
    \end{subfigure}
    \hspace{0.05\textwidth}
    \begin{subfigure}{0.45\textwidth}
        \centering
        \includegraphics[width=\textwidth]{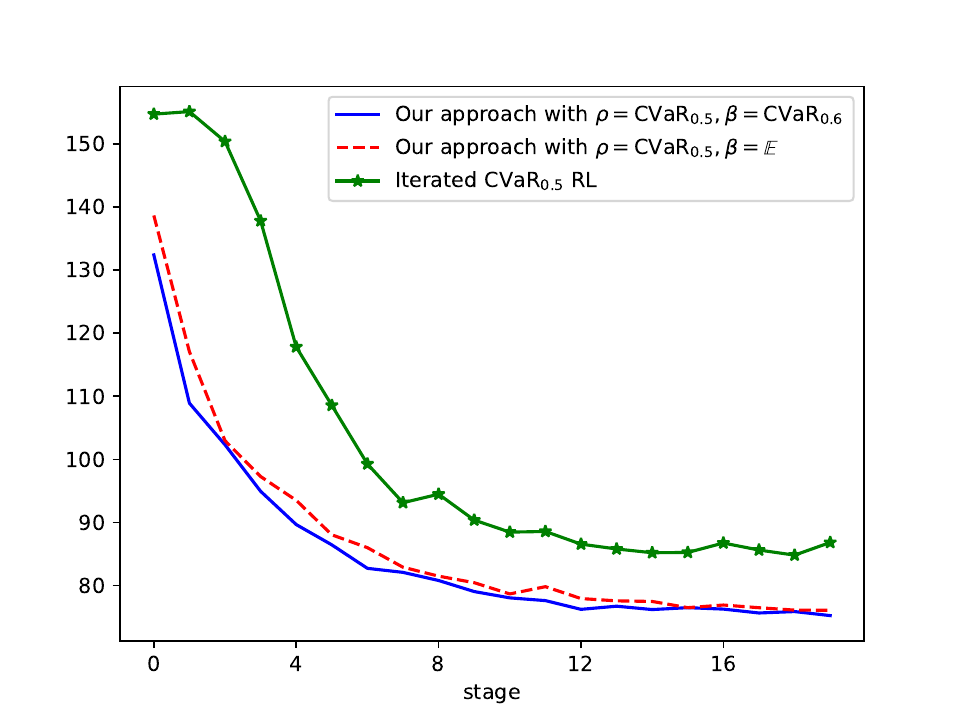}
        \caption{Inventory Management\\
        with CVaR preference}
    \end{subfigure}

    \caption{Oracle value across training stages}
    \label{fig:Converge}
\end{figure}

First, we observe that our method converges to the optimal policy in the training environment across both problems and under both preference settings (i.e., the two types of inner risk measures). This is consistent with Theorem \ref{theorem_convergenceALL}. In sharp contrast, the two DRRL models fail to converge to the optimal policy in the training environment because their objective is to minimize the worst-case cost, which prevents them from achieving optimality under the nominal model. This highlights the trade-off between optimality and robustness.
Second, the choice of the outer risk measure does not affect the final convergence outcome, nor does it have a noticeable impact on the convergence speed. Our analysis in subsection \ref{Complexity} shows that replacing the expectation with CVaR mainly increases the sample complexity by a factor associated with $(1-\gamma)^{-1},$ and the experimental results are consistent with this observation. Third, Bayesian DP exhibits a clear advantage in convergence compared with traditional Q-learning or iterative CVaR RL. Furthermore, its performance does not depend on the choice of a learning rate; the only parameter that needs to be specified is the prior over transition probabilities, which is typically interpretable. In our experiments, we adopt the simplest Dirichlet prior, where for any $(s,a,s')$, the initial parameters are given by
$\alpha(s'|s,a) = \frac{1}{|\mathcal{S}|}$
representing complete ignorance of the model. If additional knowledge were incorporated into the prior, the convergence behavior could be further improved. We illustrate this point with an example. Figure \ref{prior} demonstrates the advantage of using a more informative prior in the Coin Toss problem, where the Prior 1 is a Dirichlet prior satisfying $\alpha(s'|s,a) = \frac{1}{|\mathcal{S}|}$ and Prior 2 is a Dirichlet prior satisfying $\alpha(s' \mid s,a) \propto\binom{10}{s'} \left(\frac{2}{3}\right)^{s'} (\frac{1}{3})^{10-s'}$.

 \begin{figure}
     \centering
     \label{prior}
     \includegraphics[width=0.6\linewidth]{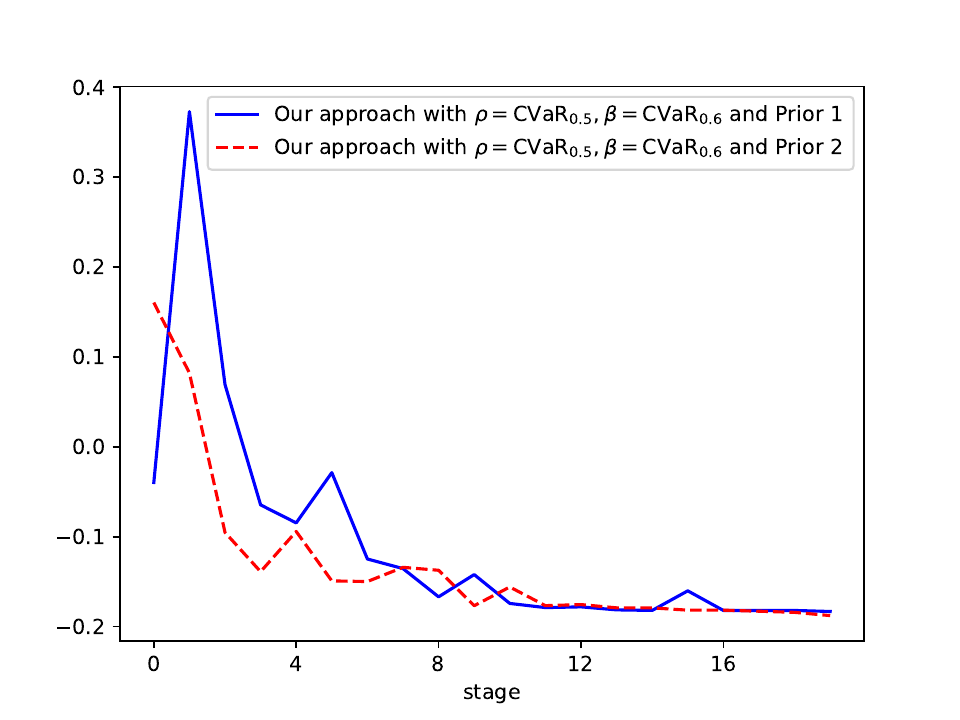}
     \caption{Different prior in Coin Toss with CVaR preference}
     \label{prior}
 \end{figure}

\subsubsection{Robustness comparison}

Below we visually demonstrate the robustness advantages brought by the introduction of outer risk measures. To compare the performance of robustness, we evaluate the trained models in a series of deployment environments whose transition probabilities are perturbed versions of those in the training environment. For the Coin Toss problem, we perturb the probability of each coin obtaining heads, changing it from the training environment value of $0.6$ to values ranging from $0.3$ to $0.9$ with increments of $0.1$. For the Inventory Management problem, we apply exponential tilting to the distribution of the demand $D$, that is,
$
\mathbb{P}(D = i) \propto \exp\!\left(\theta\left(i - \frac{n}{2}\right)\right)\ (0\leq i\leq n).
$
We vary $\theta$ from $-5$ to $5$ in increments of $1$. 
We evaluate performance using the worst value across all deployment environments, where the value is still computed as the stationary-distribution–weighted average over the states. We test the policy obtained at each training stage in the deployment environments in order to examine how model robustness evolves throughout the training process. The results are also averaged across $50$ independent runs.

\begin{figure}[htb]

    \centering
    \begin{subfigure}{0.32\textwidth}
        \centering
        \includegraphics[width=\textwidth]{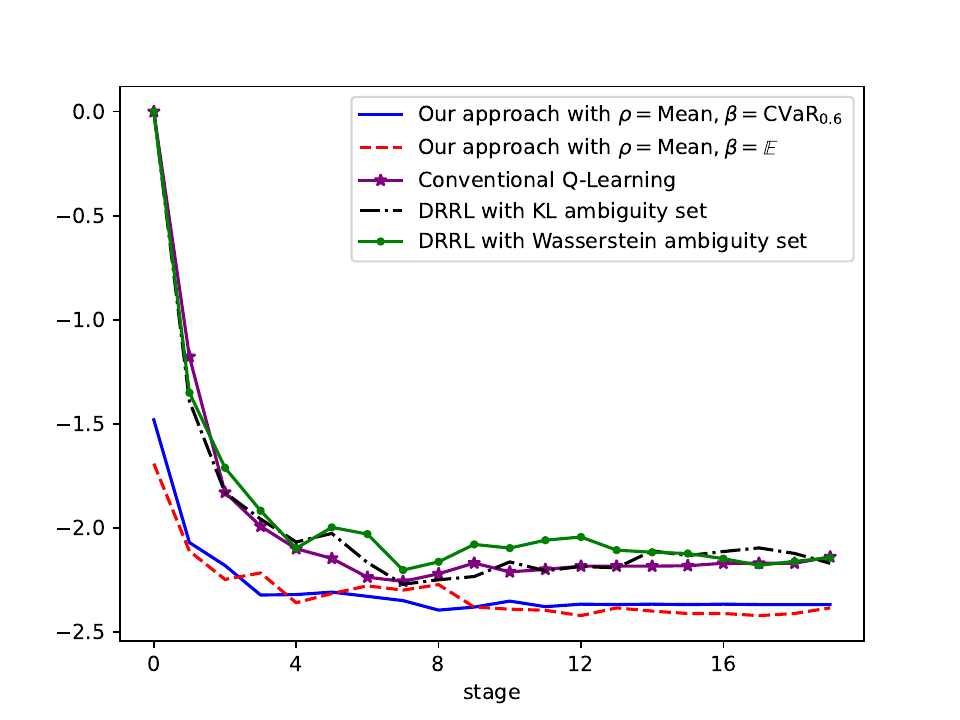}
        \caption{Mean preference\\ with $p\in\left\{0.5,0,6,0.7\right\}$}
    \end{subfigure}
    \hfill
    \begin{subfigure}{0.32\textwidth}
        \centering
        \includegraphics[width=\textwidth]{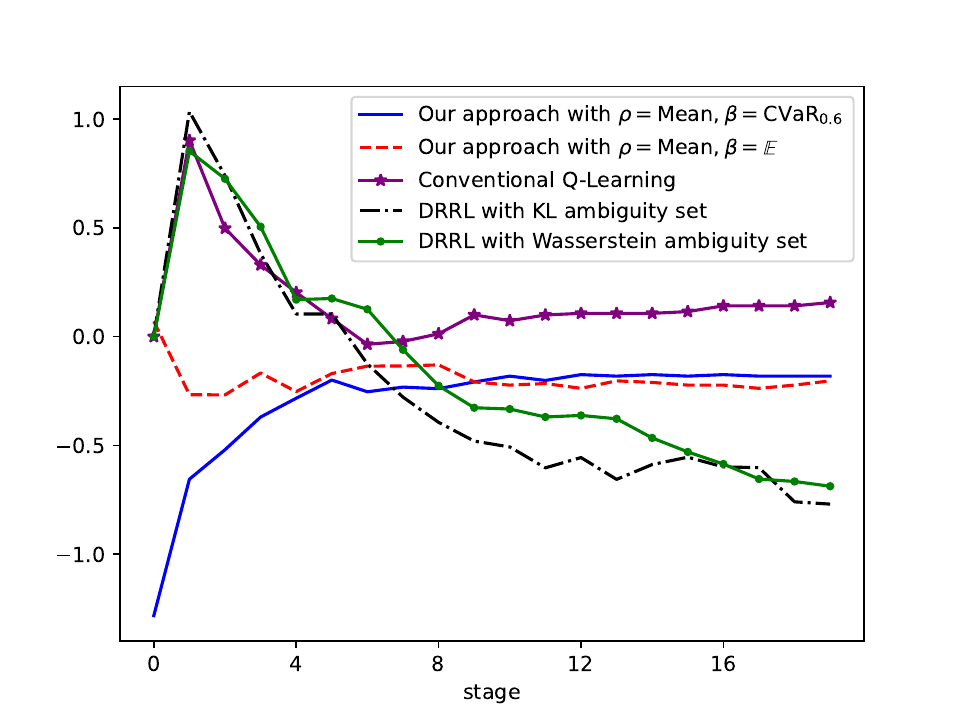}
        \caption{Mean preference\\ with $p\in\left\{0.4,\ldots,0.8\right\}$}
    \end{subfigure}
    \hfill
    \begin{subfigure}{0.32\textwidth}
        \centering
        \includegraphics[width=\textwidth]{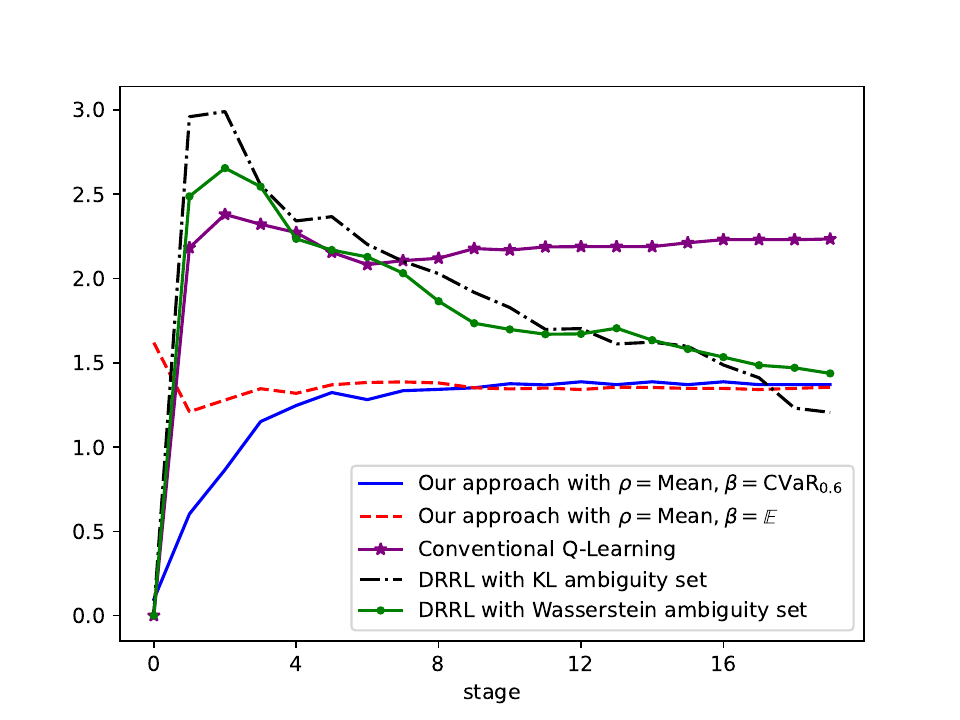}
        \caption{Mean preference\\ with $p\in\left\{0.3,\ldots,0.9\right\}$}
    \end{subfigure}

    \vspace{0.04\textwidth} 

    \begin{subfigure}{0.32\textwidth}
        \centering
        \includegraphics[width=\textwidth]{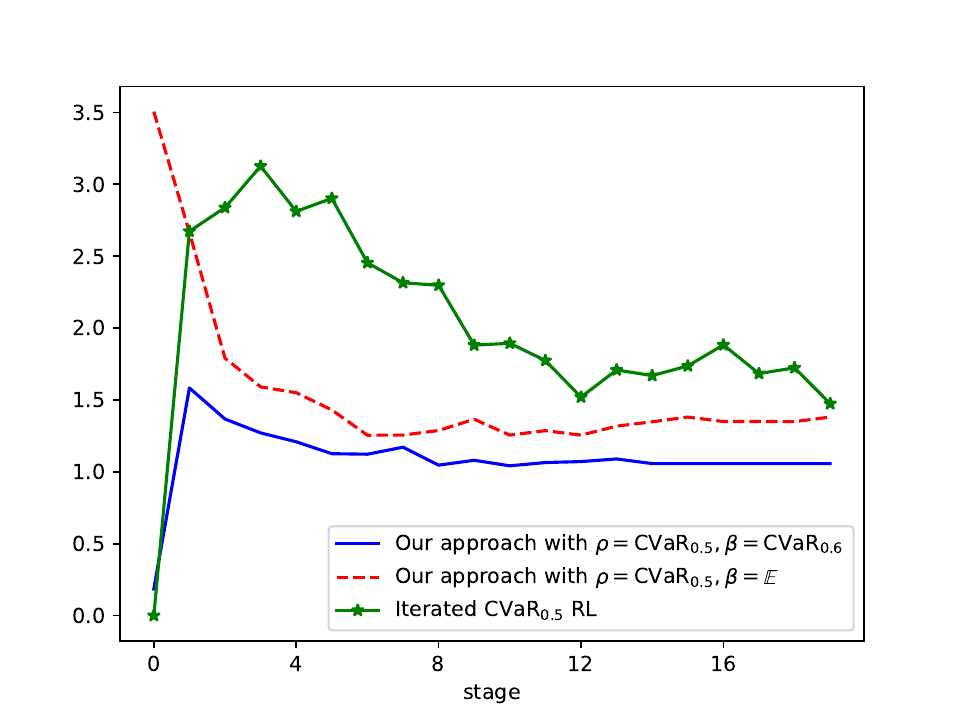}
        \caption{CVaR preference\\ with $p\in\left\{0.5,0,6,0.7\right\}$}
    \end{subfigure}
    \hfill
    \begin{subfigure}{0.32\textwidth}
        \centering
        \includegraphics[width=\textwidth]{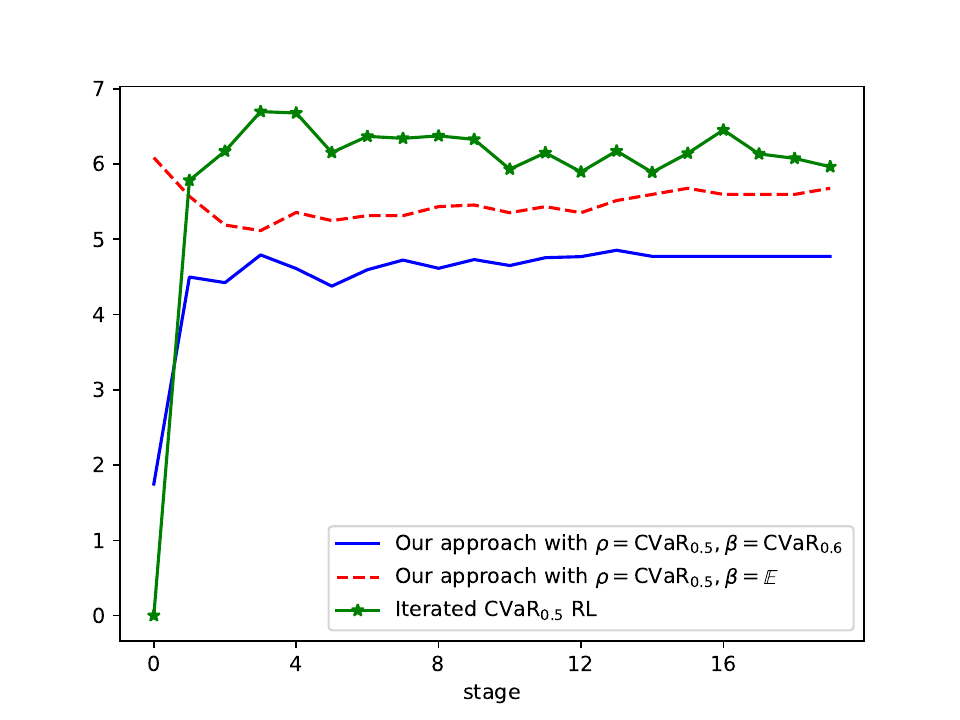}
        \caption{CVaR preference\\ with $p\in\left\{0.4,\ldots,0.8\right\}$}
    \end{subfigure}
    \hfill
    \begin{subfigure}{0.32\textwidth}
        \centering
        \includegraphics[width=\textwidth]{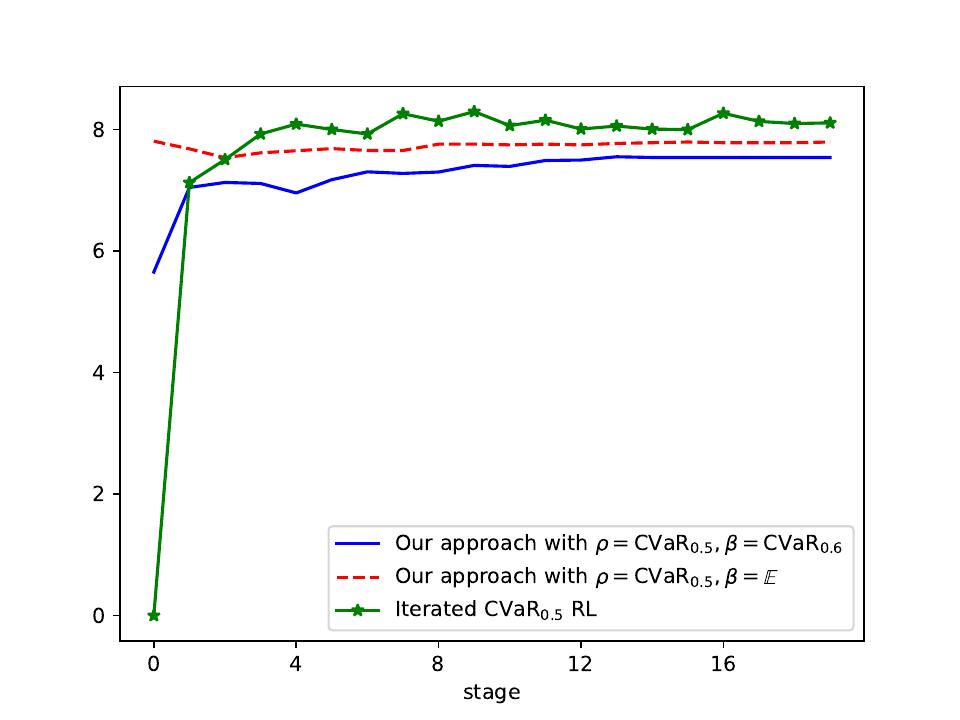}
        \caption{CVaR preference\\ with $p\in\left\{0.3,\ldots,0.9\right\}$}
    \end{subfigure}

    \caption{Robustness comparison for Coin Toss}
    \label{fig:3x2}
\end{figure}

The results of the Coin Toss problem is shown in Figure \ref{fig:3x2}. We evaluate the worst value over deployment environments for each perturbation range: the probability of heads varying within $\left\{0.5,0.6,0.7\right\}$, $\left\{0.4,\ldots,0.8\right\}$, and $\left\{0.3,\ldots,0.9\right\}$, respectively, for both the Mean preference and the CVaR preference. 
First, across all experiment groups, we observe that choosing CVaR as the outer risk measure leads to lower costs in the deployment environments compared with using the mean, highlighting the impact of the outer risk measure on robustness. 
Second, our model exhibits significantly better robustness than conventional Q-learning and iterated CVaR RL, maintaining a lower worst-case deployment cost over most of the $20$ stages. Even when the outer risk measure is chosen as expectation, our model still demonstrates a relative advantage. 
Third, in comparison with DRRL, we identify a new perspective: although there exists a trade-off between optimality and robustness, in certain settings, optimality may actually enhance robustness—particularly when the degree of transition uncertainty is unknown. When the perturbation from the training environment is small, we find that the performance in the deployment environments does not differ substantially from that in the training environment; hence, the model that is optimal in the training environment retains its advantage (e.g., subfigure (a)). However, when the perturbation becomes large, the performances of all models tend to be similar (e.g., subfigure (c)). This occurs because DRRL training depends on the radius of the ambiguity set (which can be interpreted as the degree of transition uncertainty), and its advantage emerges only when the actual perturbation matches the prescribed ambiguity radius. Additionally, even in the scenario where DRRL holds an advantage, its superiority emerges only in the stages after convergence. In contrast, our model exhibits better robustness during the initial stages, which benefits from the Bayesian framework underlying the proposed RL method. This observation also suggests that incorporating an early-stopping mechanism may further enhance robustness.

\begin{figure}[htb]
    \centering

    \begin{subfigure}{0.45\textwidth}
        \centering
        \includegraphics[width=\textwidth]{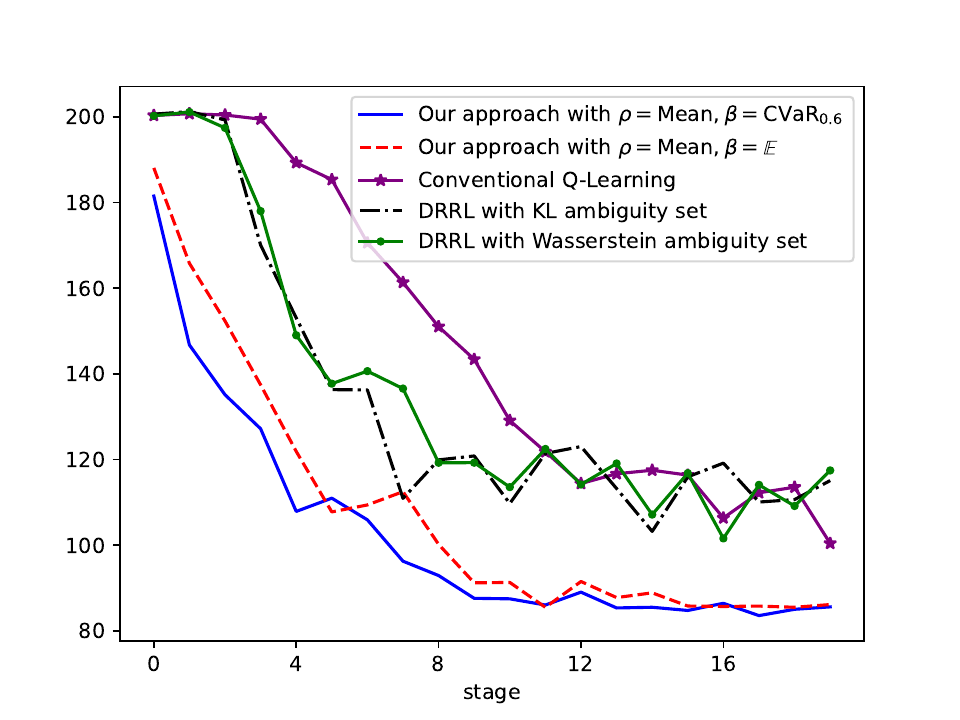}
        \caption{Mean preference with $\theta\in\left\{-5,\ldots,5\right\}$}
    \end{subfigure}
    \hfill
    \begin{subfigure}{0.45\textwidth}
        \centering
        \includegraphics[width=\textwidth]{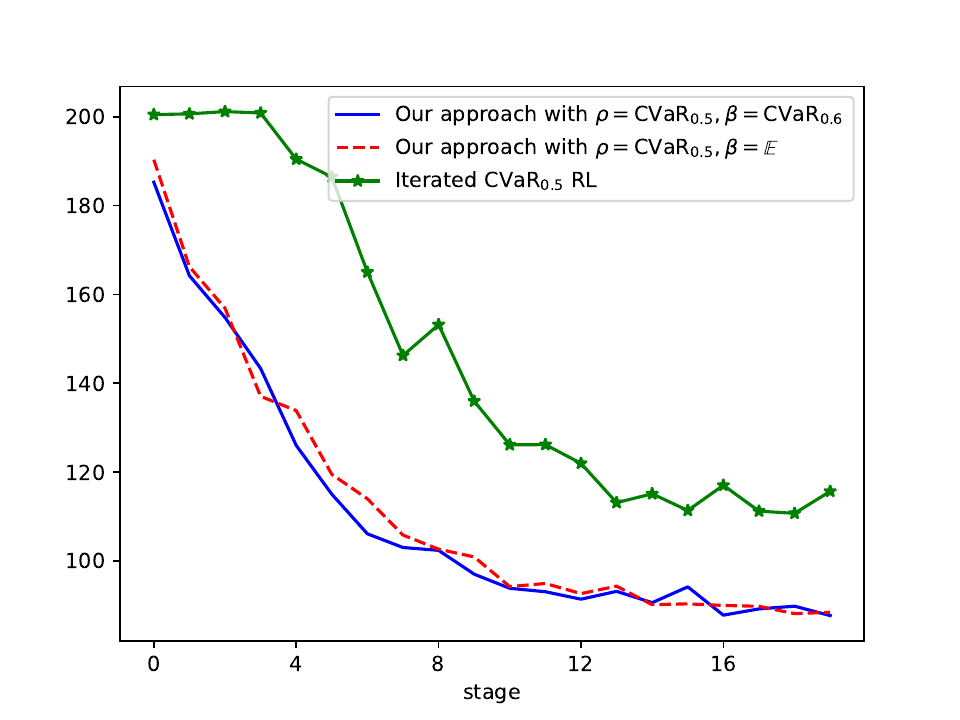}
        \caption{CVaR preference with $\theta\in\left\{-5,\ldots,5\right\}$}
    \end{subfigure}

    \caption{Robustness comparison for Inventory Management}
    \label{fig:two-in-row}
\end{figure}

The results of the Inventory Management problem is shown in Figure \ref{fig:two-in-row}. The conclusion is similar to that of the previous problem: our model consistently demonstrates a robustness advantage over $20$ stages, and selecting CVaR as the outer risk measure further strengthens this advantage. Under perturbations of this level, our model achieves a better deployment performance than DRRL. This indicates that our approach achieves a favorable balance between optimality and robustness and satisfactory performance in both the training and deployment environments.

\section{Empirical Study}

In this section, we employ an option hedging example as an empirical application of our proposed algorithm. 
We formulate the option hedging problem as a MDP model and the agent is established as the option writer, whose objective is to hedge a short position of one unit of a European call option. 
We consider a discrete hedging framework, in which the position of the underlying asset is rebalanced at predetermined, equally spaced time points.
The state space $\mathcal{S}$ is designed to include the underlying asset price $P$ and the remaining time to expiration $\tau,$ i.e., $S_t=(P_t,\tau_t).$ The action space $\mathcal{A}$ is a discrete set containing 11 equally spaced points covering the range from 0 to 1, representing the target position of the underlying asset, that is, $\mathcal{A} = \{0, 0.1, 0.2, \dots, 1.0\}$. Prior to the expiry date ($\tau > 0$), the cost is defined as the loss induced by the underlying asset; at the expiry date ($\tau = 0$), an additional cost associated with fulfilling the option obligation is incurred. Formally, $c(p,\tau,a,p',\tau-\Delta\tau)=a(p-p')-(p'-K)^+ \mathds{1} _{\{\tau=1\}}.$ Our study examines one-year options, rebalanced monthly (with a hedging interval of 
$\Delta\tau=\frac{1}{12}$). The underlying asset price is assumed to follow a geometric Brownian motion, characterized by parameters 
$(\mu,\sigma).$ Prior to the option's inception, a prior distribution is imposed on these parameters based on historical data; subsequently, the posterior distribution is updated via Bayes' theorem before each hedging operation.

Two backtesting configurations are employed. The first focuses on the SSE 50 Index (000016.SH) using the European call HO2406-C-2300 (\$K=2300\$) over the period June 19, 2023, to June 21, 2024. The second centers on the CSI 300 Index (000300.SH) using the call IO2406-C-3400 (\$K=3400\$) from June 26, 2023, to the shared expiry of June 21, 2024. The results of our approach are benchmarked against the classical Black-Scholes Delta hedging strategy.

\begin{table}[h]
\centering
\caption{Comparison of Total Hedging Losses under Different Risk Configurations}
\label{hedging_performance}
\begin{tabular}{lcc}
\toprule
Model Configuration & Experiment 1: CSI 300 Hedging & Experiment 2: SSE 50 Hedging \\
\midrule
(CVaR, Mean) & -401.23 & -189.86 \\
(Mean, Mean) & -95.62 & -98.95 \\
BS Delta Hedging & -454.86 & -253.25 \\
\bottomrule
\end{tabular}

\smallskip
\footnotesize
Note: The table presents the total hedging losses (in RMB) for three different risk-configured models across two empirical hedging experiments. Model 1: outer risk measure is CVaR, inner risk measure is Mean; Model 2: outer risk measure is Mean, inner risk measure is Mean.\\

\end{table}

Table \ref{hedging_performance} summarizes the key performance metrics. The results show that under real-world market constraints, our reinforcement learning-based approach substantially outperforms the traditional Black-Scholes Delta Hedging benchmark, particularly by achieving lower total hedging losses (negative values denote seller costs).

In Experiment 1 (CSI 300 Hedging), the BS Delta Hedging benchmark incurs the largest loss at 454.86. In contrast, our models recorded significantly lower losses: the (CVaR, Mean) configuration results in 401.23, and the (Mean, Mean) configuration achieves the best result at 95.62. Similarly, in Experiment 2 (SSE 50 Hedging), the benchmark performs comparably poorly with a loss of 253.25, while both learned policies reduce the loss substantially, with the (CVaR, Mean) model reaching 189.86 and the (Mean, Mean) model achieving 98.95. Empirically, all learned policies consistently outperform the BS Delta benchmark (with losses closer to zero), validating the proposed approach for cost-effective and robust hedging in the real market.

{\small
\bibliographystyle{apalike}
\bibliography{ref}
}

\newpage
\appendix
\section{Proofs}
\begin{proof}[Proof of Lemma \ref{lemma_1}]
First, we show that for a given set of random variables $\mathcal{X},$ a coherent risk measure $f:\mathcal{X}\to\mathbb{R}$ satisfies that $|f(X_1)-f(X_2)|\leq f(|X_1-X_2|).$ 
Let $Y = X_2-X_1.$ We have
\begin{equation}
    f(X_2) = f(X_1+Y)\leq f(X_1)+f(Y),
\notag\end{equation}
i.e.,
\begin{equation}
    f(X_2)-f(X_1)\leq f(X_1-X_2)\leq f(|X_1-X_2|).
\notag\end{equation}
Similarly, we have $f(X_1)-f(X_2)\leq f(|X_1-X_2|)$ and the conclusion holds.

For any $q\in\mathscr{P},s\in\mathcal{S},a\in\mathcal{A},$ we let $\tau_0=\pi=\delta_a,\mu_0=\delta_s$ and we have
\begin{equation}
\begin{aligned}
    |\sigma(v_1,q(\cdot|s,a))-\sigma(v_2,q(\cdot|s,a))|&= |\rho_{q,\pi,0}(v_1)-\rho_{q,\pi,0}(v_2)|\\
    &\leq \rho_{q,\pi,0}(|v_1-v_2|)\\
    &=\sigma(|v_1-v_2|,q(\cdot|s,a)),
    \end{aligned}
\notag\end{equation} and \begin{equation}
\begin{aligned}
    \sigma(|v_1-v_2|,q(\cdot|s,a))&= \rho_{q,\pi,0}(|v_1-v_2|)\\
    &\leq \max_{s'\in\mathcal{S}} |v_1(s')-v_2(s')|.
    \end{aligned}
\notag\end{equation} Therefore, 

\begin{align}
    \left|\mathcal{J}_{\chi,\pi} V_1(s)-\mathcal{J}_{\chi,\pi} V_2(s)\right| 
    &= \Bigg|\beta_{p\sim\chi}\left(\sum_{a\in\mathcal{A}}\pi(a|s) \cdot\sigma(c(s,a,\cdot)+\gamma V_1(\cdot),p(\cdot|s,a))\right)\notag\\
    &\qquad\qquad-\beta_{p\sim\chi}\left(\sum_{a\in\mathcal{A}}\pi(a|s) \cdot\sigma(c(s,a,\cdot)+\gamma V_2(\cdot),p(\cdot|s,a))\right)\Bigg|\notag\\
    &\leq \beta_{p\sim\chi}\Bigg(\Bigg|\sum_{a\in\mathcal{A}}\pi(a|s) \cdot\sigma(c(s,a,\cdot)+\gamma V_1(\cdot),p(\cdot|s,a))\notag\\
    &\qquad\qquad-\sum_{a\in\mathcal{A}}\pi(a|s) \cdot\sigma(c(s,a,\cdot)+\gamma V_2(\cdot),p(\cdot|s,a))\Bigg|\Bigg)\notag\\
    &\leq \beta_{p\sim\chi}\Bigg(\sum_{a\in\mathcal{A}}\pi(a|s) \cdot\Bigg|\sigma(c(s,a,\cdot)+\gamma V_1(\cdot),p(\cdot|s,a))\notag\\
    &\qquad\qquad\qquad\qquad\qquad- \sigma(c(s,a,\cdot)+\gamma V_2(\cdot),p(\cdot|s,a))\Bigg|\Bigg)\notag\\
    &\leq \beta_{p\sim\chi}\Bigg(\sum_{a\in\mathcal{A}}\pi(a|s) \cdot\sigma(\gamma |V_1(\cdot)-V_2(\cdot)|,p(\cdot|s,a))\Bigg)\notag\\
    &\leq \beta_{p\sim\chi}\Bigg(\sum_{a\in\mathcal{A}}\pi(a|s)  \gamma\max_{s'\in\mathcal{S}}|V_1(s')-V_2(s')|\Bigg)\notag\\
    &=\beta_{p\sim\chi}\Bigg(  \gamma\max_{s'\in\mathcal{S}}|V_1(s')-V_2(s')|\Bigg)\notag\\
    &=\gamma\max_{s'\in\mathcal{S}}|V_1(s')-V_2(s')|,\notag
    \end{align}
  which implies $\max_{s\in\mathcal{S}} \left|\mathcal{J}_{\chi,\pi} V_1(s)-\mathcal{J}_{\chi,\pi} V_2(s)\right|\leq \gamma\max_{s\in\mathcal{S}}|V_1(s)-V_2(s)|,$ i.e., $\left\Vert\mathcal{J}_{\chi,\pi} V_1-\mathcal{J}_{\chi,\pi} V_2\right\Vert_{\infty} \leq \gamma\Vert V_1-V_2\Vert_{\infty}.$ Similarly, we have
\begin{align}
    \left|\mathcal{J}_{\chi} V_1(s)-\mathcal{J}_{\chi} V_2(s)\right| 
    &= \Bigg|\min_{a\in\mathcal{A}}\beta_{p\sim\chi}\left(\sigma(c(s,a,\cdot)+\gamma V_1(\cdot),p(\cdot|s,a))\right)\notag\\
    &\qquad\qquad-\min_{a\in\mathcal{A}}\beta_{p\sim\chi}\left(\sigma(c(s,a,\cdot)+\gamma V_2(\cdot),p(\cdot|s,a))\right)\Bigg|\notag\\
    &\leq \max_{a\in\mathcal{A}}\beta_{p\sim\chi}\Bigg(\Bigg|\sigma(c(s,a,\cdot)+\gamma V_1(\cdot),p(\cdot|s,a))\notag\\
    &\qquad\qquad-\sigma(c(s,a,\cdot)+\gamma V_2(\cdot),p(\cdot|s,a))\Bigg|\Bigg)\notag\\
    &\leq \max_{a\in\mathcal{A}}\beta_{p\sim\chi}\Bigg(\sigma(\gamma |V_1(\cdot)-V_2(\cdot)|,p(\cdot|s,a))\Bigg)\notag\\
    &\leq\beta_{p\sim\chi}\Bigg(  \gamma\max_{s'\in\mathcal{S}}|V_1(s')-V_2(s')|\Bigg)\notag\\
    &=\gamma\max_{s'\in\mathcal{S}}|V_1(s')-V_2(s')|,\notag
    \end{align} which implies $\max_{s\in\mathcal{S}} \left|\mathcal{J}_{\chi} V_1(s)-\mathcal{J}_{\chi} V_2(s)\right|\leq \gamma\max_{s\in\mathcal{S}}|V_1(s)-V_2(s)|,$ i.e., $\left\Vert\mathcal{J}_{\chi} V_1-\mathcal{J}_{\chi} V_2\right\Vert_{\infty} \leq \gamma\Vert V_1-V_2\Vert_{\infty}.$ 
\end{proof}

\begin{proof}[Proof of Theorem \ref{theorem_1}]
By Lemma \ref{lemma_1}, we have that $\mathcal{J}_{\boldsymbol{\chi}}$ is a $\gamma$-contraction operator w.r.t. $\Vert\cdot\Vert_{\infty}$ norm. According to Banach’s contraction mapping principle (\cite{banach1922operations}), there exists a unique value function $V_{\chi}^*$ such that $ V =   \mathcal{J}_{\boldsymbol{\chi}} V,$ i.e., $ V(s) =   \mathcal{J}_{\boldsymbol{\chi}} V(s),$ for any $s\in\mathcal{S}.$ For any policy $\pi,$ and any state $s\in\mathcal{S},$ 
\begin{align}
V_{\chi,\pi}(s)
= &\beta_{p\sim\chi}({\rho}_{p,\pi,0}(c(s,A_0,S_1)+\notag\\
&\quad\gamma\beta_{p\sim\chi}(\mathscr{\rho}_{p,\pi,1}( c(S_1,A_1,S_2) +\notag\\
&\quad\quad\gamma\beta_{p\sim\chi}(\mathscr{\rho}_{p,\pi,2}\left( c(S_2,A_2,S_3) 
+\cdots\right)))))),\notag\\
= &\beta_{p\sim\chi}\Big(\sum_{a\in\mathcal{A}}\pi(a|s)\sigma((c(s,a,S_1)+\notag\\
&\quad\gamma\beta_{p\sim\chi}(\mathscr{\rho}_{p,\pi,1}( c(S_1,A_1,S_2) +\notag\\
&\quad\quad\gamma\beta_{p\sim\chi}(\mathscr{\rho}_{p,\pi,2}\left( c(S_2,A_2,S_3) 
+\cdots\right)))),p(\cdot|,s,a))\Big),\notag\\
=& \beta_{p\sim\chi}\Big(\sum_{a\in\mathcal{A}}\pi(a|s)\sigma(c(s,a,\cdot) + \gamma V_{\chi,\pi}(\cdot),p(\cdot|,s,a))\Big)\notag\\
\geq& \beta_{p\sim\chi}\Big(\min_{a\in\mathcal{A}}\sigma(c(s,a,\cdot) + \gamma V_{\chi,\pi}(\cdot),p(\cdot|,s,a))\Big)\notag\\
=& \mathcal{J}_{\chi}V_{\chi,\pi}(s).\notag
\end{align}
Therefore, through recursive iteration we get $V_{\chi,\pi}(s) \geq(\mathcal{J}_{\chi})^nV_{\chi,\pi}(s)$ for any $n\geq1.$ According to Banach’s contraction mapping principle, $\lim_{n\to \infty} \Vert(\mathcal{J}_{\chi})^nV_{\chi,\pi}-V_{\chi}^*\Vert\to 0,$ which implies $\lim_{n\to \infty} (\mathcal{J}_{\chi})^nV_{\chi,\pi}(s)=V_{\chi}^*(s).$ Thus, we have  
\begin{equation}
    V_{\chi,\pi}(s) \geq\lim_{n\to\infty}(\mathcal{J}_{\chi})^nV_{\chi,\pi}(s) = V_{\chi}^*(s).
\notag\end{equation}
Let $\pi_{\chi}^*$ be a Markov policy satisfying that $\pi_{\chi}^*(a|s)$ is a point mass on
\begin{equation}
    \arg \min_{a\in\mathcal{A}}  \left\{\beta_{p\sim\chi}\left(\sigma(c(s,a,\cdot)+\gamma\cdot V_{\chi}^*(s),p(\cdot|s,a))\right)\right\},
\notag\end{equation} for any $(s,a) \in \mathcal{A}\times\mathcal{S}.$ Then we have for any $s\in\mathcal{S},$
\begin{align}
V_{\chi,\pi^*_{\chi}}(s)
= &\beta_{p\sim\chi}({\rho}_{p,\pi^*_{\chi},0}(c(s,A_0,S_1)+\notag\\
&\quad\gamma\beta_{p\sim\chi}(\mathscr{\rho}_{p,\pi^*_{\chi},1}( c(S_1,A_1,S_2) +\notag\\
&\quad\quad\gamma\beta_{p\sim\chi}(\mathscr{\rho}_{p,\pi^*_{\chi},2}\left( c(S_2,A_2,S_3) 
+\cdots\right)))))),\notag\\
= &\beta_{p\sim\chi}\Big(\sum_{a\in\mathcal{A}}\pi^*_{\chi}(a|s)\sigma((c(s,a,S_1)+\notag\\
&\quad\gamma\beta_{p\sim\chi}(\mathscr{\rho}_{p,\pi^*_{\chi},1}( c(S_1,A_1,S_2) +\notag\\
&\quad\quad\gamma\beta_{p\sim\chi}(\mathscr{\rho}_{p,\pi^*_{\chi},2}\left( c(S_2,A_2,S_3) 
+\cdots\right)))),p(\cdot|,s,a))\Big),\notag\\
=& \beta_{p\sim\chi}\Big(\sum_{a\in\mathcal{A}}\pi^*_{\chi}(a|s)\sigma(c(s,a,\cdot) + \gamma V_{\chi,\pi^*_{\chi}}(\cdot),p(\cdot|,s,a))\Big)\notag\\
=& \min_{a\in\mathcal{A}}\beta_{p\sim\chi}\Big(\sigma(c(s,a,\cdot) + \gamma V_{\chi,\pi^*_{\chi}}(\cdot),p(\cdot|,s,a))\Big)\notag\\
=& \mathcal{J}_{\chi}V_{\chi,\pi^*_{\chi}}(s).\notag
\end{align}
Therefore, through recursive iteration we get $V_{\chi,\pi}(s) =(\mathcal{J}_{\chi})^nV_{\chi,\pi}(s)$ for any $n\geq1.$ According to Banach’s contraction mapping principle, $
    V_{\chi,\pi^*_{\chi}}(s) =\lim_{n\to\infty}(\mathcal{J}_{\chi})^nV_{\chi,\pi^*_{\chi}}(s) = V_{\chi}^*(s).
$
\end{proof}

\begin{proof}[Proof of Corollary \ref{coro1}]
For any $s\in\mathcal{S,}$ we have $\text{Risk}(\chi,\pi,\mu_0,\pi)=V_{\chi,\pi}(s)$ on $\{S_0=s\}.$ By Theorem \ref{theorem_1}, we have $V_{\chi,\pi^*_{\chi}}(s)\leq V_{\chi,\pi}(s)$ for any Markov policy $\pi,$ which implies 
\begin{equation}
    \text{Risk}(\chi, \pi_{\chi}^*,\mu_0, \pi_{\chi}^*)\leq \text{Risk}(\chi,\pi,\mu_0,\pi)\ \text{for any $\pi$ on}\ \{S_0=s\},
\notag\end{equation}
i.e.,
\begin{equation}
    \pi_{\chi}^* = \arg\min_{\pi\in\Pi} \left\{\text{Risk}(\chi,\pi,\mu_0,\pi)\right\}\ \text{on}\ \{S_0=s\}.
\notag\end{equation}
Since $\mathbb{P}^{q,\pi,\mu_0,\tau_0}(\cup_{s\in\mathcal{S}}\{ S_0=s\})=1,$ it holds that
\begin{equation}
    \pi_{\chi}^* = \arg\min_{\pi\in\Pi} \left\{\text{Risk}(\chi,\pi,\mu_0,\pi)\right\}\ \text{almost surely.}
\notag\end{equation}
\end{proof}

\begin{proof}[Proof of Theorem \ref{theorem_2}]
	(1) For the first conclusion, from the proof of Theorem \ref{theorem_1}, we have $V_{\chi,\pi}  = \mathcal{J}_{\chi,\pi}V_{\chi,\pi}.$ Therefore, we have \begin{align}
			|V_{\chi,\pi}(s)-V_{\delta_{\bar{q}},\pi}(s)| & = |\mathcal{J}_{\chi,\pi}V_{\chi,\pi}(s)-\mathcal{J}_{\delta_{\bar{q}},\pi}V_{\delta_{\bar{q}},\pi}(s)|\notag\\
			&=\Bigg|\beta_{p\sim\chi}\left(\sum_{a\in\mathcal{A}}\pi(a|s)\sigma(c(s,a,\cdot)+\gamma V_{\chi,\pi}(\cdot),p(\cdot|s,a))\right)\notag\\
			&\qquad\qquad-\beta_{p\sim\delta_{\bar{q}}}\left(\sum_{a\in\mathcal{A}}\pi(a|s)\sigma(c(s,a,\cdot)+\gamma V_{\delta_{\bar{q}},\pi}(\cdot),p(\cdot|s,a))\right)\Bigg|\notag\\
			&=\Bigg|\beta_{p\sim\chi}\left(\sum_{a\in\mathcal{A}}\pi(a|s)\sigma(c(s,a,\cdot)+\gamma V_{\chi,\pi}(\cdot),p(\cdot|s,a))\right)\notag\\
			&\qquad\qquad-\sum_{a\in\mathcal{A}}\pi(a|s)\sigma(c(s,a,\cdot)+\gamma V_{\delta_{\bar{q}},\pi}(\cdot),\bar{q}(\cdot|s,a))\Bigg|\notag\\
			&=\Bigg|\beta_{p\sim\chi}\Bigg(\sum_{a\in\mathcal{A}}\pi(a|s)(\sigma(c(s,a,\cdot)+\gamma V_{\chi,\pi}(\cdot),p(\cdot|s,a))\notag\\
			&\qquad\qquad\qquad\qquad-\sigma(c(s,a,\cdot)+\gamma V_{\delta_{\bar{q}},\pi}(\cdot),\bar{q}(\cdot|s,a)))\Bigg)\Bigg|\notag\\
			&=\beta_{p\sim\chi}\Bigg(\sum_{a\in\mathcal{A}}\pi(a|s)\Bigg(\Bigg|\sigma(c(s,a,\cdot)+\gamma V_{\chi,\pi}(\cdot),p(\cdot|s,a))-\sigma(c(s,a,\cdot)+\gamma V_{\delta_{\bar{q}},\pi}(\cdot),p(\cdot|s,a))\Bigg|\notag\\
            &\qquad\qquad\qquad\qquad+\Bigg|\sigma(c(s,a,\cdot)+\gamma V_{\delta_{\bar{q}},\pi}(\cdot),p(\cdot|s,a))-\sigma(c(s,a,\cdot)+\gamma V_{\delta_{\bar{q}},\pi}(\cdot),\bar{q}(\cdot|s,a))\Bigg|\Bigg)\Bigg).\notag
		\end{align}
	By the proof of Lemma \ref{lemma_1} and Assumption \ref{ass_sigma}, we have 
		\begin{align}
			|V_{\chi,\pi}(s)-V_{\delta_{\bar{q}},\pi}(s)| & \leq\beta_{p\sim\chi}\Bigg(\sum_{a\in\mathcal{A}}\pi(a|s)\Bigg(\gamma\Vert V-V \Vert_{\infty}+ B_{\sigma} \sum_{s'\in\mathcal{S}}|p(s'|s,a)-\bar{q}(s'|s,a)|\Bigg)\Bigg)\notag\\
			& \leq\gamma\Vert V_{\chi,\pi}-V_{\delta_{\bar{q}},\pi} \Vert_{\infty}+\beta_{p\sim\chi}\Bigg(\sum_{a\in\mathcal{A}}\pi(a|s)\Bigg(B_{\sigma} \sum_{s'\in\mathcal{S}}|p(s'|s,a)-\bar{q}(s'|s,a)|\Bigg)\Bigg)\notag\\
            & \leq\gamma\Vert V_{\chi,\pi}-V_{\delta_{\bar{q}},\pi} \Vert_{\infty}+B_{\sigma}\sum_{a\in\mathcal{A}}\pi(a|s)\beta_{p\sim\chi}\Bigg( \sum_{s'\in\mathcal{S}}|p(s'|s,a)-\bar{q}(s'|s,a)|\Bigg)\notag\\
            & \leq\gamma\Vert V_{\chi,\pi}-V_{\delta_{\bar{q}},\pi} \Vert_{\infty}+B_{\sigma}\max_{a\in\mathcal{A}}\beta_{p\sim\chi}\Bigg( \sum_{s'\in\mathcal{S}}|p(s'|s,a)-\bar{q}(s'|s,a)|\Bigg)\notag
		\end{align}
	which implies 
	\begin{align}
		\Vert V_{\chi,\pi}-V_{\delta_{\bar{q}},\pi}\Vert_{\infty}&\leq\gamma\Vert V_{\chi,\pi}-V_{\delta_{\bar{q}},\pi} \Vert_{\infty}+B_{\sigma}\max_{a\in\mathcal{A},s\in\mathcal{S}}\beta_{p\sim\chi}\Bigg( \sum_{s'\in\mathcal{S}}|p(s'|s,a)-\bar{q}(s'|s,a)|\Bigg),\notag\\
		\Vert V_{\chi,\pi}-V_{\delta_{\bar{q}},\pi}\Vert_{\infty}&\leq\frac{B_{\sigma}}{1-\gamma} \max_{a\in\mathcal{A},s\in\mathcal{S}}\beta_{p\sim\chi}\Bigg( \sum_{s'\in\mathcal{S}}|p(s'|s,a)-\bar{q}(s'|s,a)|\Bigg).\notag
	\end{align}
	
	(2) For the second conclusion, from the proof of Theorem \ref{theorem_1}, we have $V_{\chi}^*  = \mathcal{J}_{\chi}V_{\chi}^*.$ Therefore, we have 
		\begin{align}
			|V_{\chi}^*(s)-V_{\delta_{\bar{q}}}^*| & = |\mathcal{J}_{\chi}V_{\chi}^*(s)-\mathcal{J}_{\delta_{\bar{q}}}V_{\delta_{\bar{q}}}^*(s)|\notag\\
			&=\Bigg|\min_{a\in\mathcal{A}}\beta_{p\sim\chi}\left(\sigma(c(s,a,\cdot)+\gamma V_{\chi}^*(\cdot),p(\cdot|s,a))\right)\notag\\
			&\qquad\qquad-\min_{a\in\mathcal{A}}\beta_{p\sim\delta_{\bar{q}}}\left(\sigma(c(s,a,\cdot)+\gamma V_{\delta_{\bar{q}}}^*(\cdot),p(\cdot|s,a))\right)\Bigg|\notag\\
			&=\Bigg|\min_{a\in\mathcal{A}}\beta_{p\sim\chi}\left(\sigma(c(s,a,\cdot)+\gamma V_{\chi}^*(\cdot),p(\cdot|s,a))\right)\notag\\
			&\qquad\qquad-\min_{a\in\mathcal{A}}\sigma(c(s,a,\cdot)+\gamma V_{\delta_{\bar{q}}}^*(\cdot),\bar{q}(\cdot|s,a))\Bigg|\notag\\
			&\leq \max_{a\in\mathcal{A}}\beta_{p\sim\chi}\Bigg(\Bigg|\sigma(c(s,a,\cdot)+\gamma V_{\chi}^*(\cdot),p(\cdot|s,a))\notag\\
			&\qquad\qquad\qquad\qquad-\sigma(c(s,a,\cdot)+\gamma V_{\delta_{\bar{q}}}^*(\cdot),\bar{q}(\cdot|s,a))\Bigg|\Bigg)\notag\\
			&=\max_{a\in\mathcal{A}}\beta_{p\sim\chi}\Bigg(\Bigg|\sigma(c(s,a,\cdot)+\gamma V_{\chi}^*(\cdot),p(\cdot|s,a))-\sigma(c(s,a,\cdot)+\gamma V_{\delta_{\bar{q}}}^*(\cdot),p(\cdot|s,a))\Bigg|\notag\\
			&\qquad\qquad\qquad\qquad+\Bigg|\sigma(c(s,a,\cdot)+\gamma V_{\delta_{\bar{q}}}^*(\cdot),p(\cdot|s,a))-\sigma(c(s,a,\cdot)+\gamma V_{\delta_{\bar{q}}}^*(\cdot),\bar{q}(\cdot|s,a))\Bigg|\Bigg).\notag
		\end{align}
	By the proof of Lemma \ref{lemma_1} and Assumption \ref{ass_sigma}, we have 
	\begin{align}
			|V_{\chi,\pi}(s)-V_{\delta_{\bar{q}},\pi}(s)| & \leq\max_{a\in\mathcal{A}}\beta_{p\sim\chi}\Bigg(\gamma\Vert V-V \Vert_{\infty}+ B_{\sigma} \sum_{s'\in\mathcal{S}}|p(s'|s,a)-\bar{q}(s'|s,a)|\Bigg)\notag\\
			& \leq\gamma\Vert V_{\chi,\pi}-V_{\delta_{\bar{q}},\pi} \Vert_{\infty}+\max_{a\in\mathcal{A}}\beta_{p\sim\chi}\Bigg(B_{\sigma} \sum_{s'\in\mathcal{S}}|p(s'|s,a)-\bar{q}(s'|s,a)|\Bigg)\notag\\
            & \leq\gamma\Vert V_{\chi,\pi}-V_{\delta_{\bar{q}},\pi} \Vert_{\infty}+B_{\sigma}\max_{a\in\mathcal{A}}\beta_{p\sim\chi}\Bigg( \sum_{s'\in\mathcal{S}}|p(s'|s,a)-\bar{q}(s'|s,a)|\Bigg)\notag
		\end{align}
	which implies 
	\begin{align}
		\Vert V_{\chi,\pi}-V_{\delta_{\bar{q}},\pi}\Vert_{\infty}&\leq\gamma\Vert V_{\chi,\pi}-V_{\delta_{\bar{q}},\pi} \Vert_{\infty}+B_{\sigma}\max_{a\in\mathcal{A},s\in\mathcal{S}}\beta_{p\sim\chi}\Bigg( \sum_{s'\in\mathcal{S}}|p(s'|s,a)-\bar{q}(s'|s,a)|\Bigg),\notag\\
		\Vert V_{\chi,\pi}-V_{\delta_{\bar{q}},\pi}\Vert_{\infty}&\leq\frac{B_{\sigma}}{1-\gamma} \max_{a\in\mathcal{A},s\in\mathcal{S}}\beta_{p\sim\chi}\Bigg( \sum_{s'\in\mathcal{S}}|p(s'|s,a)-\bar{q}(s'|s,a)|\Bigg).\notag
	\end{align}
	
\end{proof}

\begin{proof}[Proof of Corollary \ref{corol2}]
	Similar to the proof of Corollary \ref{coro1}, we only need to show the conclusion holds on $\left\{S_0=s\right\}$ for any $s\in\mathcal{S},$ on which $\text{Risk}(\chi,\pi,\mu_0,\pi) = V_{\chi,\pi}(s).$ On  $\left\{S_0=s\right\},$ we have \begin{equation}
		\begin{aligned}
			|\text{Risk}(\delta_{\bar{q}},\pi^*_{\bar{q}},\mu_0,\pi^*_{\bar{q}}) - \text{Risk}(\delta_{\bar{q}},\pi^*_{\chi},\mu_0,\pi^*_{\chi})| &= |V_{\delta_{\bar{q}},\pi^*_{\bar{q}}}(s)-V_{\delta_{\bar{q}},\pi^*_{\chi}}(s)|.
		\end{aligned}
	\notag\end{equation}
	Furthermore, 
	\begin{equation}
		\begin{aligned}
			|V_{\delta_{\bar{q}},\pi^*_{\bar{q}}}(s)-V_{\delta_{\bar{q}},\pi^*_{\chi}}(s)| & =|V_{\delta_{\bar{q}}}^*(s)-V_{\delta_{\bar{q}},\pi^*_{\chi}}(s)| \\
			& =|V_{\delta_{\bar{q}}}^*(s)-V_{\chi}^*(s)+V_{\chi}^*(s)-V_{\delta_{\bar{q}},\pi^*_{\chi}}(s)| \\
			& =|V_{\delta_{\bar{q}}}^*(s)-V_{\chi}^*(s)+V_{\chi,\pi^*_{\chi}}(s)-V_{\delta_{\bar{q}},\pi^*_{\chi}}(s)| \\
			&\leq |V_{\delta_{\bar{q}}}^*(s)-V_{\chi}^*(s)|+|V_{\chi,\pi^*_{\chi}}(s)-V_{\delta_{\bar{q}},\pi^*_{\chi}}(s)|.
		\end{aligned}
	\notag\end{equation}
	By Theorem \ref{theorem_2}, we have
	\begin{align}
		|V_{\delta_{\bar{q}}}^*(s)-V_{\chi}^*(s)| \leq \frac{B_{\sigma}}{1-\gamma} \max_{a\in\mathcal{A},s\in\mathcal{S}}\beta_{p\sim\chi}\Bigg( \sum_{s'\in\mathcal{S}}|p(s'|s,a)-\bar{q}(s'|s,a)|\Bigg),\notag\\
		|V_{\chi,\pi^*_{\chi}}(s)-V_{\delta_{\bar{q}},\pi^*_{\chi}}(s)| \leq\frac{B_{\sigma}}{1-\gamma} \max_{a\in\mathcal{A},s\in\mathcal{S}}\beta_{p\sim\chi}\Bigg( \sum_{s'\in\mathcal{S}}|p(s'|s,a)-\bar{q}(s'|s,a)|\Bigg),\notag
	\end{align}
	Thus the conclusion follows.
	
\end{proof}

	\begin{proof}[Proof of Proposition  \ref{important_prop}]
		\begin{itemize}
			\item[(1)] For any $N\geq 1 $ and $\mu\in  \tilde{\mathcal{V}}_N (\chi),$ we have that with probability $1,$
			\begin{equation}\begin{aligned}
					&\int_{\mathscr{P}} \mu(p) \mathrm{d}  F_{\chi}(p) =  \sum_{i=1}^N\int_{D_i}\hat{\mu}(p_i)\mathrm{d} F_{\chi}(p) =\frac{1}{N}\sum_{i=1}^N \hat{\mu}(p_i)=1,\\
					&\mu(p)=\sum_{i=1}^N\hat{\mu}(p_i)\mathds 1_{D_i}\geq 0 ,\ \forall p\in\mathscr{P},\\
					& w_k(\mu(p))=\sum_{i=1}^N w_k\left(\hat{\mu}(p_i)\right) \mathds 1_{D_i} \leq 0,\ k\in\mathcal{K},\\
					&  \int_{\mathscr{P}} g_e(\mu(p)) \mathrm{d} F_{\chi}(p)=\sum_{i=1}^N\int_{D_i}  g_e(\hat{\mu}(p_i)) \mathrm{d} F_{\chi}(p)=\frac{1}{N}\sum_{i=1}^N g_e(\hat{\mu}(p_i))  \leq 0,\ \forall e \in \mathcal{E},
			\end{aligned}\notag\end{equation} which implies $\mu\in\mathcal{V}(\chi).$ Therefore we have $\tilde{\mathcal{V}}_N (\chi)\subset \mathcal{V}(\chi)$ for any $N\geq 1 $ and 
			\begin{equation}
				\sup_{N\geq 1}\sup_{\mu\in\tilde{\mathcal{V}}_N (\chi)} \Vert{\mu}\Vert_{m_1} \leq \sup_{\mu\in{\mathcal{V}} (\chi)} \Vert{\mu}\Vert_{m_1}<\infty.
			\notag\end{equation}
			\item[(2)] First, we note that 
			\begin{equation}
				W_N = \max_{\tilde{\mu}_N\in  \tilde{\mathcal{V}}_N (\chi)} \sum_{i=1}^N\int_{D_i}\tilde{\mu}_N(p)\cdot\sigma(c(s,a,\cdot)+\gamma V(\cdot),p_i(\cdot|s,a))\mathrm{d} F_{\chi}(p).
			\notag\end{equation}
			Moreover, 
			\begin{equation}\begin{aligned} &\Bigg|\int_{\mathscr{P}}\tilde{\mu}_N(p)\cdot\sigma(c(s,a,\cdot)+\gamma V(\cdot),p(\cdot|s,a))\mathrm{d} F_{\chi}(p)\notag\\
					&\qquad-\sum_{i=1}^N\int_{D_i}\tilde{\mu}_N(p)\cdot\sigma(c(s,a,\cdot)+\gamma V(\cdot),p_i(\cdot|s,a))\mathrm{d} F_{\chi}(p)\Bigg|\notag\\
					\leq & \sum_{i=1}^N \int_{D_i}\tilde{\mu}_N(p)|\sigma(c(s,a,\cdot)+\gamma V(\cdot),p(\cdot|s,a))\notag\\
					&\qquad-\sigma(c(s,a,\cdot)+\gamma V(\cdot),p_i(\cdot|s,a))| \mathrm{d} F_{\chi}(p)\notag\\
					\leq& \sum_{i=1}^N \int_{D_i}\tilde{\mu}_N(p) \max_{1\leq i \leq N}\sup_{p\in D_i}|\sigma(c(s,a,\cdot)+\gamma V(\cdot),p(\cdot|s,a))\notag\\
					&\qquad-\sigma(c(s,a,\cdot)+\gamma V(\cdot),p_i(\cdot|s,a))|\mathrm{d} F_{\chi}(p)\notag\\
					\leq&\max_{1\leq i \leq N}\sup_{p\in D_i}B_{\sigma} \cdot \left(\sum_{s'\in\mathcal{S}}\left|p(s'|s,a)-p_i(s'|s,a) \right|\right)\notag\\
					\leq&\max_{1\leq i \leq N}\sup_{p\in D_i}B_{\sigma} \cdot\Vert p-p_i \Vert \notag\\
					\leq&B_{\sigma} \cdot  \left(\max_{1\leq i\leq N}\text{diam}(D_i)+ \max_{1\leq i\leq N}\text{dist}(p_i,D_i)\right).\notag
			\end{aligned}\end{equation}
			Therefore, 
			\begin{equation}\begin{aligned}
					&\Bigg|W_N -\max_{\tilde{\mu}_N\in\tilde{\mathcal{V}}_N(\chi)}\int_{\mathscr{P}}\tilde{\mu}_N(p)\cdot\sigma(c(s,a,\cdot)+\gamma V(\cdot),p(\cdot|s,a))\mathrm{d} F_{\chi}(p)\Bigg|\\
					\leq & \max_{\tilde{\mu}_N\in\tilde{\mathcal{V}}_N(\chi)} \Bigg|\int_{\mathscr{P}}\tilde{\mu}_N(p)\cdot\sigma(c(s,a,\cdot)+\gamma V(\cdot),p(\cdot|s,a))\mathrm{d} F_{\chi}(p)\notag\\
					&\qquad-\sum_{i=1}^N\int_{D_i}\tilde{\mu}_N(p)\cdot\sigma(c(s,a,\cdot)+\gamma V(\cdot),p_i(\cdot|s,a))\mathrm{d} F_{\chi}(p)\Bigg|\notag\\
					\to& 0\ (N\to \infty),
			\end{aligned}\end{equation}    uniformly for any $V$ with $\Vert V\Vert_{\infty}\leq \frac{\bar{C}}{1-\gamma},$  $\mathbb{P}^{\text{sample}}$-almost surely.
			\item[(3)] For any $N\geq 1$ and $\mu\in\mathcal{V}(\chi),$ let $\hat{\mu}(p) = N\int_{D_i} \mu(p) \mathrm{d} F_{\chi}(p) $ and $\tilde{\mu}_N(p) = \sum_{i=1}^N \hat{\mu}(p)\mathds 1_{D_i}.$ Then we have that with probability $1,$
			\begin{equation}\begin{aligned}
					&\frac{1}{N}\sum_{i=1}^N \hat{\mu}(p_i)=\sum_{i=1}^N\int_{D_i} \mu(p) \mathrm{d} F_{\chi}(p)=1 ,\\
					&\hat{\mu}(p_i)=N\int_{D_i} \mu(p) \mathrm{d} F_{\chi}(p)\geq 0 ,\ \forall 1\leq i\leq N,\\
					& w_k(\hat{\mu}(p_i))=w_k\left(N\int_{D_i} \mu(p) \mathrm{d} F_{\chi}(p)\right)\\
					&\qquad\qquad\ \leq N\int_{D_i}w_k\left({\mu}(p)\right)\mathrm{d} F_{\chi}(p)  \leq 0,\ \forall 1\leq i\leq N,\ k\in\mathcal{K},\\
					&  \frac 1N \sum_{i=1}^N g_e(\hat{\mu}(p_i)) =\frac 1N \sum_{i=1}^N g_e\left(N\int_{D_i} \mu(p) \mathrm{d} F_{\chi}(p)\right)\\
					&\qquad\qquad\ \leq\sum_{i=1}^N\int_{D_i} g_e(\mu(p))\mathrm{d} F_{\chi}(p)= \int_{\mathscr{P}} g_e(\mu(p))\mathrm{d} F_{\chi}(p)\leq 0,\ \forall e \in \mathcal{E},
			\end{aligned}\notag\end{equation}
			which implies $\tilde{\mu}_N \in \tilde{\mathcal{V}}_N(\chi).$ Since $\text{diam}(D_i)\to 0$ almost surely, by Lebesgue Differentiation Theorem we have $\tilde{\mu}_N \rightharpoonup\mu$ in $\mathscr{L}_{m_1} $ weak topology. For any sequence $\left\{\tilde{\mu}_{N}\right\}^{\infty}_{N=1}$ with $\tilde{\mu}_N\in\tilde{\mathcal{V}}_N(\chi),$ by conclusion (1) we have $\left\{\tilde{\mu}_{N}\right\}^{\infty}_{N=1}\subset\mathcal{V}(\chi).$ As $\mathcal{V}(\chi)$ is convex and bounded, it is weakly sequentially compact. Therefore, there exists a subsequence $\left\{\tilde{\mu}_{N_k}\right\}^{\infty}_{k=1}$ such that $\tilde{\mu}_{N_k} \rightharpoonup\mu$ in $\mathscr{L}_{m_1} $ weak topology for some $\mu\in\mathcal{V}(\chi).$

		\end{itemize}
	
Furthermore, due to the continuity of $\sigma,$ we have 
\begin{equation}
    \sup_{V:\Vert V\Vert\leq \frac{\bar{C}}{1-\gamma}} \int_{\mathscr{P}} (\sigma(c(s,a,\cdot)+\gamma V(\cdot),p(\cdot|s,a)))^{m_1} \mathrm d F_{\chi}(p) <\infty.\notag
\end{equation}
Then by Banach–Steinhaus Theorem, we have 
\begin{equation}
\begin{aligned}
    \sup_{V:\Vert V\Vert\leq \frac{\bar{C}}{1-\gamma}}\left|\langle\tilde{\mu}_N-\mu,\sigma(c(s,a,\cdot)+\gamma V(\cdot),p(\cdot|s,a))\rangle\right|\to 0 \  (N\to \infty),\\
    \sup_{V:\Vert V\Vert\leq \frac{\bar{C}}{1-\gamma}}\left|\langle\tilde{\mu}_{N_k}-\mu,\sigma(c(s,a,\cdot)+\gamma V(\cdot),p(\cdot|s,a))\rangle\right|\to 0 \  (N\to \infty),
    \end{aligned}\notag
\end{equation}
    which implies that both convergence results are uniform.
    \end{proof}
	
\begin{proof}[Proof of Theorem \ref{theorem_convergenceBE}]

	By (2) in Proposition \ref{important_prop}, it is sufficient to show that for any $(s,a)\in\mathcal{S}\times\mathcal{A},$ \begin{equation}\begin{aligned}
			&\lim_{N\to\infty} \Bigg|\max_{\mu\in{\mathcal{V}}(\chi)}\int_{\mathscr{P}}\mu(p)\cdot\sigma(c(s,a,\cdot)+\gamma V(\cdot),p(\cdot|s,a))\mathrm{d} F_{\chi}(p) \\
			&\qquad\qquad\ -\max_{\mu\in\tilde{\mathcal{V}}_N(\chi)}\int_{\mathscr{P}}\mu(p)\cdot\sigma(c(s,a,\cdot)+\gamma V(\cdot),p(\cdot|s,a))\mathrm{d} F_{\chi}(p)\Bigg| = 0.
	\end{aligned}\notag\end{equation} 
	We define
	\begin{equation}\begin{aligned}
			&\mu^* = \arg\max_{\mu\in{\mathcal{V}}(\chi)}\int_{\mathscr{P}}\mu(p)\cdot\sigma(c(s,a,\cdot)+\gamma V(\cdot),p(\cdot|s,a))\mathrm{d} F_{\chi}(p)\\
			&\mu^*_N =   \arg\max_{\mu\in\tilde{\mathcal{V}}_N(\chi)}\int_{\mathscr{P}}\mu(p)\cdot\sigma(c(s,a,\cdot)+\gamma V(\cdot),p(\cdot|s,a))\mathrm{d} F_{\chi}(p),N\geq 1
	\end{aligned}\notag\end{equation}
	On one hand, Proposition \ref{important_prop} (3), with probability $1,$ there exists a sequence $\left\{{\mu}_{N}\right\}^{\infty}_{N=1}$ with ${\mu}_N\in\tilde{\mathcal{V}}_N(\chi)$ such that ${\mu}_N \rightharpoonup\mu^*.$ Thus we have 
	\begin{equation}\begin{aligned}
			&\liminf_{N\to\infty} \int_{\mathscr{P}}\mu^*_{N}(p)\cdot\sigma(c(s,a,\cdot)+\gamma V(\cdot),p(\cdot|s,a))\mathrm{d} F_{\chi}(p)\\
			\geq &\liminf_{N\to\infty} \int_{\mathscr{P}}\mu_{N}(p)\cdot\sigma(c(s,a,\cdot)+\gamma V(\cdot),p(\cdot|s,a))\mathrm{d} F_{\chi}(p)\\
			=&\int_{\mathscr{P}}\mu^*(p)\cdot\sigma(c(s,a,\cdot)+\gamma V(\cdot),p(\cdot|s,a))\mathrm{d} F_{\chi}(p)\\
			=&\max_{\mu\in{\mathcal{V}}(\chi)}\int_{\mathscr{P}}\mu(p)\cdot\sigma(c(s,a,\cdot)+\gamma V(\cdot),p(\cdot|s,a))\mathrm{d} F_{\chi}(p),
	\end{aligned}\notag\end{equation}  and this holds uniformly for $V$ with bound $\frac{\bar{C}}{1-\gamma}.$ 
	On the other hand, by Proposition \ref{important_prop} (1), with probability $1,$ $\left\{\mu^*_N\right\}_{N=1}^{\infty}\subset\mathcal{V}(\chi).$ Thus, we have
	\begin{equation}\begin{aligned}
			&\limsup_{N\to\infty} \int_{\mathscr{P}}\mu^*_{N}(p)\cdot\sigma(c(s,a,\cdot)+\gamma V(\cdot),p(\cdot|s,a))\mathrm{d} F_{\chi}(p)\\
			\leq&\int_{\mathscr{P}}\mu^*(p)\cdot\sigma(c(s,a,\cdot)+\gamma V(\cdot),p(\cdot|s,a))\mathrm{d} F_{\chi}(p)\\
			=&\max_{\mu\in{\mathcal{V}}(\chi)}\int_{\mathscr{P}}\mu(p)\cdot\sigma(c(s,a,\cdot)+\gamma V(\cdot),p(\cdot|s,a))\mathrm{d} F_{\chi}(p),
	\end{aligned}\notag\end{equation} and this holds uniformly for $V$ with bound $\frac{\bar{C}}{1-\gamma}.$
	Then, we have for any value function $V,$
    \begin{equation}\label{convegence}
    \begin{aligned}
        &\max_{\mu\in\tilde{\mathcal{V}}_N(\chi)}\int_{\mathscr{P}}\mu(p)\cdot\sigma(c(s,a,\cdot)+\gamma V(\cdot),p(\cdot|s,a))\mathrm{d} F_{\chi}(p)\\
        \to &\max_{\mu\in{\mathcal{V}}(\chi)}\int_{\mathscr{P}}\mu(p)\cdot\sigma(c(s,a,\cdot)+\gamma V(\cdot),p(\cdot|s,a))\mathrm{d} F_{\chi}(p).
        \end{aligned}
    \end{equation}

 Furthermore, the convergence in \eqref{convegence} is uniform for $V$ with $\Vert V\Vert_{\infty}\leq \frac{\bar{C}}{1-\gamma}.$ Combining this with Proposition \ref{important_prop} (2), we get the conclusion. 
\end{proof}

\begin{proof}[Proof of Theorem \ref{theorem_posterior}]
   Denote $\inf_{u\geq 1}\epsilon_{(u)} >0$ by $\underline{\epsilon}\ (\underline{\epsilon}>0).$ We have \begin{equation}
       \mathbb{P}(A_t=a|S_t=s)\geq \frac{\underline{\epsilon}}{|\mathcal{A}|},
   \notag\end{equation}
for any $s\in\mathcal{S},a\in\mathcal{A}$ and $t\geq0.$ Define
\begin{equation}
\begin{aligned}
    d(s,s')= \inf \{n\geq1:&\exists s_1,s_2,\ldots,s_n\in\mathcal{S},a_0,a_1,a_2,\ldots,a_n\in\mathcal{A},\ \\
    & \text{such that}\ \bar{q}(s_1|s,a_0)\bar{q}(s_2|s_1,a_1)\cdots \bar{q}(s_n|s_{n-1},a_{n-1})\bar{q}(s'|s_n,a_n)>0\} ,
\end{aligned}
\notag\end{equation}
 and $d = \max_{s,s'}d(s,a).$ Furthermore, we define 
\begin{equation}
\begin{aligned}
    \delta(s,s')= \max \{x\in\mathbb{R}:&\exists s_1,s_2,\ldots,s_n\in\mathcal{S},a_0,a_1,a_2,\ldots,a_n\in\mathcal{A},\ \text{such that}\ n\leq d(s,s')\\
    &\text{and}\  x=\bar{q}(s_1|s,a_0)\bar{q}(s_2|s_1,a_1)\cdots \bar{q}(s_n|s_{n-1},a_{n-1})\bar{q}(s'|s_n,a_n)\},
\end{aligned}
\notag\end{equation} and $\delta = \inf_{s,s'}\delta(s,s').$ Therefore, for any $s,s'\in\mathcal{S},$ there exist $n(s,s')\leq d$ and $s_1,\ldots,s_{n(s,s')},$ $a_0,a_1,\ldots,a_{n(s,s')},$
\begin{equation}
\begin{aligned}
     &\mathbb{P}(S_{t+n(s,s')+1}=s'|S_t=s)\\
     \geq& \frac{\underline{\epsilon}}{|\mathcal{A}|}\bar{q}(s_1|s,a_0)\frac{\underline{\epsilon}}{|\mathcal{A}|}\bar{q}(s_2|s_1,a_1)\cdots \frac{\underline{\epsilon}}{|\mathcal{A}|}\bar{q}(s'|s_{n(s,s')},a_{n(s,s')})\\
     \geq&\left (\frac{\underline{\epsilon}}{|\mathcal{A}|}\right)^{n(s,s')+1} \delta\\
     \geq&\left (\frac{\underline{\epsilon}}{|\mathcal{A}|}\right)^{d+1} \delta.\\
     \end{aligned}
\notag\end{equation}
Thus, we have for any $s,s'\in\mathcal{S}$ and $t\geq 1,$
\begin{equation}
    \begin{aligned}
        \mathbb{P}\left(\exists k \geq 1,S_{t+k}=s'|S_t=s\right) & = 1-\mathbb{P}\left(\forall k \geq 1,S_{t+k}\neq s'|S_t=s\right)\\
        & \geq 1-\mathbb{P}\left(\cap^{\infty}_ {l=0}\{S_{t+k}\neq s',ld\leq k\leq(l+1)d\}|S_t=s\right)\\
        & \geq 1-\prod_{l=0}^{\infty}\max_{s\in\mathcal{S}}\mathbb{P}\left(S_{t+k}\neq s',ld\leq k\leq(l+1)d|S_{t+ld}=s\right)\\
        & \geq 1-\prod_{l=0}^{\infty}\left(1-\min_{s\in\mathcal{S}}\mathbb{P}\left(\exists ld\leq k\leq(l+1)d,S_{t+k}=s'|S_{t+ld}=s\right)\right)\\
        & \geq 1-\prod_{l=0}^{\infty}\left(1-\min_{s\in\mathcal{S}}\mathbb{P}\left(S_{t+n(s,s')}=s'|S_{t+ld}=s\right)\right)\\
        &\geq 1- \left(1-\left (\frac{\underline{\epsilon}}{|\mathcal{A}|}\right)^{d+1} \delta\right)^{\infty}\\
        &=1.
    \end{aligned}
\notag\end{equation}
Therefore, we have
\begin{equation}
    \begin{aligned}
        \mathbb{P}\left(S_{t+k}=s',\ \text{i.o.}|S_t=s\right) &\geq \mathbb{P}\left(\cap_{m=0}^{\infty}\{\exists k\geq u ,S_{t+k}=s'\}|S_t=s\right)\\
    &\geq \mathbb{P}\left(\lim_{u\to\infty}\{\exists k\geq u ,S_{t+k}=s'\}|S_t=s\right)\\
     &\geq \lim_{u\to\infty}\sum_{s_1\in\mathcal{S}}\mathbb{P}\left(\exists k\geq u ,S_{t+k}=s'|S_{t+u}=s_1\right)\mathbb{P}(S_{t+u}=s_1|S_t=s)\\
     &= \lim_{u\to\infty}\sum_{s_1\in\mathcal{S}}1\cdot\mathbb{P}(S_{t+u}=s_1|S_t=s)\\
     & = 1,
    \end{aligned}
\notag\end{equation}
which implies 
\begin{equation}
    \begin{aligned}
        \mathbb{P}\left(S_{t}=s,\ \text{i.o.}\right) = \sum_{s_0\in\mathcal{S}}\mathbb{P}(S_0=s_0)\mathbb{P}\left(S_{t}=s,\ \text{i.o.}|S_0=s_0\right) =1.
    \end{aligned}\notag
\end{equation}
Moreover, we have 
\begin{equation}
    \begin{aligned}
        &\mathbb{P}\left(S_{t}=s,\ \text{i.o.},\forall s\in\mathcal{S}\right) =\mathbb{P}\left(\cap_{s\in\mathcal{S}}\{S_{t}=s,\ \text{i.o.}\}\right)= 1,\\
     &   \mathbb{P}(S_t=s,A_t=a,\ \text{i.o.},\forall s\in\mathcal{S},a\in\mathcal{A}) = 1.
    \end{aligned}\notag
\end{equation}
 Considering the posterior, we have
\begin{equation}\begin{aligned}
 f_{\chi|x_{0:T}}(p) 
     &\propto f_{\chi}(p)\prod_{s,a\in\mathcal{S}\times\mathcal{A}}\prod_{s'\in\mathcal{S}}\left(\left.p\left(s'\right|s,a\right)\right)^{m_T(s,a,s')}\\
     &= f_{\chi}(p)\prod_{s,a\in\mathcal{S}\times\mathcal{A}}\exp\left(\sum_{s'\in\mathcal{S}}m_T(s,a,s')\log(p(s'|s,a))\right)\\
     &= f_{\chi}(p)\prod_{s,a\in\mathcal{S}\times\mathcal{A}}\exp\left(T\sum_{s'\in\mathcal{S}}\frac{m_T(s,a,s')}{T}\log(p(s'|s,a))\right),
 \end{aligned}\notag\end{equation}
where $m_T(s,a,s') = \sum_{t=1}^{T}\mathds 1_{\{S_t=s,A_t =a,S_{t+1}=s'\}}.$
On $\{S_t=s,A_t=a,\ \text{i.o.},\forall s\in\mathcal{S},a\in\mathcal{A}\},$ we have
\begin{equation}
    \frac{m_T(s,a,s')}{T} \sim \bar{q}(s'|s,a)\ (T\to\infty).
\notag\end{equation}
Thus, it holds that
\begin{equation}\begin{aligned}
& f_{\chi}(p)\prod_{s,a\in\mathcal{S}\times\mathcal{A}}\exp\left(T\sum_{s'\in\mathcal{S}}\frac{m_T(s,a,s')}{T}\log(p(s'|s,a))\right)\\
 =& f_{\chi}(p)\prod_{s,a\in\mathcal{S}\times\mathcal{A}}\exp\left(T\sum_{s'\in\mathcal{S}}\bar{q}(s'|s,a)\log(p(s'|s,a))+o(T)\right)\\
 =&  f_{\chi}(p)\prod_{s,a\in\mathcal{S}\times\mathcal{A}}\prod_{s'\in\mathcal{S}}\left(\left.p\left(s'\right|s,a\right)\right)^{T\bar{q}(s'|s,a)+o(T)},
 \end{aligned}\notag\end{equation}
 which implies 
 \begin{equation}
     f_{\chi|x_{0:T}}(p) = \frac{f_{\chi}(p)\prod_{s,a\in\mathcal{S}\times\mathcal{A}}\prod_{s'\in\mathcal{S}}\left(\left.p\left(s'\right|s,a\right)\right)^{T\bar{q}(s'|s,a)+o(T)}}{\int_{\mathscr{P}}f_{\chi}(p)\prod_{s,a\in\mathcal{S}\times\mathcal{A}}\prod_{s'\in\mathcal{S}}\left(\left.p\left(s'\right|s,a\right)\right)^{T\bar{q}(s'|s,a)+o(T)}\mathrm{d} p}
\label{frac}\end{equation} Denote $\sum_{s,a,s'}\bar{q}(s'|s,a)\log(p(s'|s,a))$ by $l(x).$ We have $l(p)<l(\bar{q})$ for any $p\in\mathscr{P}$ and $p\neq \bar{q}.$ Then the numerator of \eqref{frac}  satisfies 
\begin{equation}
    f_{\chi}(p)\prod_{s,a\in\mathcal{S}\times\mathcal{A}}\prod_{s'\in\mathcal{S}}\left(\left.p\left(s'\right|s,a\right)\right)^{T\cdot\bar{q}(s'|s,a)+o(T)}\sim C_1 \exp(T\cdot l(p))
\notag\end{equation}
for some constant $C_1>0$ and by Laplace's method for the asymptotic approximation of integrals (\cite{olver1997asymptotics}) the denominator satisfies 
\begin{equation}
    \int_{\mathscr{P}}f_{\chi}(p)\prod_{s,a\in\mathcal{S}\times\mathcal{A}}\prod_{s'\in\mathcal{S}}\left(\left.p\left(s'\right|s,a\right)\right)^{T\bar{q}(s'|s,a)+o(T)}\mathrm{d} p\sim C_2 \exp(T\cdot l(\bar{q}))T^{-C_3}
\notag\end{equation}
for some constants $C_2,C_3>0.$ Therefore, we have
\begin{equation}
    f_{\chi|x_{0:T}}(p) \sim \frac{C_1}{C_2}\exp(T(l(p)-l(\bar{q})))T^{C_3}.
\notag\end{equation}
Thus $f_{\chi|x_{0:T}}(p)\to0$ if $p\neq \bar{q},$ and $f_{\chi|x_{0:T}}(p)\to\infty$ if $p=\bar{q},$ and the conclusion follows. 
\end{proof}

\begin{proof}[Proof of Theorem \ref{theorem_convergenceALL}]
  
    Combining Theorem \ref{theorem_posterior} with Theorem \ref{theorem_2}, we have \begin{equation}
        \begin{aligned}
            \Vert V^*_{\chi_{(u)}}-V^*_{\delta_{\bar{q}}}\Vert_{\infty} &\leq\frac{B_{\sigma}}{1-\gamma} \max_{a\in\mathcal{A},s\in\mathcal{S}}\beta_{p\sim\chi_{(u)}} \left(\sum_{s,a,s'}\left|p(s'|s,a)-\bar{q}(s'|s,a)\right|\right)\\
            &\leq\frac{B_{\sigma}}{1-\gamma} \beta_{p\sim\chi_{(u)}} \left(\max_{s,a,s'}\left|p(s'|s,a)-\bar{q}(s'|s,a)\right|\right)\\
            &\to\frac{B_{\sigma}}{1-\gamma} \beta_{p\sim\delta_{\bar{q}}} \left(\max_{s,a,s'}\left|p(s'|s,a)-\bar{q}(s'|s,a)\right|\right)\\
            &=0,
        \end{aligned}
    \notag\end{equation}
    almost surely as $u\to\infty.$ By Theorem \ref{theorem_convergenceBE}, with probability $1,$ when $N$ is sufficiently large, there holds \begin{equation}
        \left\Vert\hat{\mathcal{J}}_{\chi_{(u)}}V_1-\hat{\mathcal{J}}_{\chi_{(u)}}V_2\right\Vert_{\infty} \leq \gamma \Vert V_1-V_2\Vert_{\infty}
    \notag\end{equation} for any $V_1,V_2$ with bound $\frac{\bar{C}}{1-\gamma}.$ By Banach’s contraction mapping principle, there exists $\tilde{V}_{(u)}$ such that $\hat{\mathcal{J}}_{\chi_{(u)}}\tilde{V}_{(u)}=\tilde{V}_{(u)}.$ Therefore, we have
    \begin{equation}
        \begin{aligned}
            \left\Vert \hat{V}^{*}_{{(u)}}-V^*_{\chi_{(u)}}\right\Vert_{\infty} &\leq \left\Vert \hat{V}^{*}_{{(u)}}-\tilde{V}_{(u)}\right\Vert+\left\Vert\tilde{V}_{(u)}-V^*_{\chi_{(u)}}\right\Vert_{\infty},\\
            \left\Vert \hat{V}^{*}_{{(u)}}-V^*_{\chi_{(u)}}\right\Vert_{\infty}&\leq \left\Vert \hat{V}^{*}_{{(u)}}-\hat{\mathcal{J}}_{\chi_{(u)}}\hat{V}^{*}_{{(u)}}\right\Vert+ \left\Vert\hat{\mathcal{J}}_{\chi_{(u)}} \hat{V}^{*}_{{(u)}}-\tilde{V}_{(u)}\right\Vert+\left\Vert\tilde{V}_{(u)}-V^*_{\chi_{(u)}}\right\Vert_{\infty},\\
            \left\Vert \hat{V}^{*}_{{(u)}}-V^*_{\chi_{(u)}}\right\Vert_{\infty}&\leq \left\Vert \hat{V}^{*}_{{(u)}}-\hat{\mathcal{J}}_{\chi_{(u)}}\hat{V}^{*}_{{(u)}}\right\Vert+ \left\Vert\hat{\mathcal{J}}_{\chi_{(u)}} \hat{V}^{*}_{{(u)}}-\hat{\mathcal{J}}_{\chi_{(u)}}\tilde{V}_{(u)}\right\Vert+\left\Vert\tilde{V}_{(u)}-V^*_{\chi_{(u)}}\right\Vert_{\infty},\\
             \left\Vert \hat{V}^{*}_{{(u)}}-V^*_{\chi_{(u)}}\right\Vert_{\infty}&\leq \left\Vert \hat{V}^{*}_{{(u)}}-\hat{\mathcal{J}}_{\chi_{(u)}}\hat{V}^{*}_{{(u)}}\right\Vert+\gamma \left\Vert\hat{V}^{*}_{{(u)}}-\tilde{V}_{(u)}\right\Vert+\left\Vert\tilde{V}_{(u)}-V^*_{\chi_{(u)}}\right\Vert_{\infty},\\
             \left\Vert \hat{V}^{*}_{{(u)}}-V^*_{\chi_{(u)}}\right\Vert_{\infty}&\leq \frac{1}{1-\gamma}\left\Vert \hat{V}^{*}_{{(u)}}-\hat{\mathcal{J}}_{\chi_{(u)}}\hat{V}^{*}_{{(u)}}\right\Vert+ \frac{1}{1-\gamma}\left\Vert\tilde{V}_{(u)}-V^*_{\chi_{(u)}}\right\Vert_{\infty},\\
             \left\Vert \hat{V}^{*}_{{(u)}}-V^*_{\chi_{(u)}}\right\Vert_{\infty}&\leq \frac{\theta}{1-\gamma}+ \frac{1}{1-\gamma}\left\Vert\tilde{V}_{(u)}-V^*_{\chi_{(u)}}\right\Vert_{\infty}.
        \end{aligned}
   \notag \end{equation}
Moreover, 
\begin{equation}
    \begin{aligned}
        \left\Vert\tilde{V}_{(u)}-V^*_{\chi_{(u)}}\right\Vert_{\infty}&\leq \left\Vert\tilde{V}_{(u)}-{\mathcal{J}}_{\chi_{(u)}}\tilde{V}_{(u)}\right\Vert_{\infty}+ \left\Vert{\mathcal{J}}_{\chi_{(u)}}\tilde{V}_{(u)}- V^*_{\chi_{(u)}}\right\Vert_{\infty}, \\
        \left\Vert\tilde{V}_{(u)}-V^*_{\chi_{(u)}}\right\Vert_{\infty}&\leq \left\Vert\hat{\mathcal{J}}_{\chi_{(u)}}\tilde{V}_{(u)}-{\mathcal{J}}_{\chi_{(u)}}\tilde{V}_{(u)}\right\Vert_{\infty}+ \left\Vert{\mathcal{J}}_{\chi_{(u)}}\tilde{V}_{(u)}- {\mathcal{J}}_{\chi_{(u)}}V^*_{\chi_{(u)}}\right\Vert_{\infty}, \\
        \left\Vert\tilde{V}_{(u)}-V^*_{\chi_{(u)}}\right\Vert_{\infty}&\leq \left\Vert\hat{\mathcal{J}}_{\chi_{(u)}}\tilde{V}_{(u)}-{\mathcal{J}}_{\chi_{(u)}}\tilde{V}_{(u)}\right\Vert_{\infty}+ \gamma\left\Vert\tilde{V}_{(u)}- V^*_{\chi_{(u)}}\right\Vert_{\infty}, \\
        \left\Vert\tilde{V}_{(u)}-V^*_{\chi_{(u)}}\right\Vert_{\infty}&\leq \frac{1}{1-\gamma}\left\Vert\hat{\mathcal{J}}_{\chi_{(u)}}\tilde{V}_{(u)}-{\mathcal{J}}_{\chi_{(u)}}\tilde{V}_{(u)}\right\Vert_{\infty}, 
    \end{aligned}
\notag\end{equation}
which approaches $0$ almost surely as $N\to\infty$ by Theorem \ref{theorem_convergenceBE}.
Finally, we have \begin{equation}
        \begin{aligned}
            \limsup_{u\to\infty,N\to\infty}\Vert \hat{V}^{*}_{{(u)}}-V^*_{\delta_{\bar{q}}}\Vert_{\infty} &= \limsup_{u\to\infty,N\to\infty}\Vert \hat{V}^{*}_{{(u)}}-V^*_{\chi_{(u)}}+ V^*_{\chi_{(u)}}-V^*_{\delta_{\bar{q}}}\Vert_{\infty}\\
            &\leq \limsup_{u\to\infty,N\to\infty}\Vert \hat{V}^{*}_{{(u)}}-V^*_{\chi_{(u)}}\Vert_{\infty}+ \limsup_{u\to\infty,N\to\infty}\Vert V^*_{\chi_{(u)}}-V^*_{\delta_{\bar{q}}}\Vert_{\infty},\\
            & =  \limsup_{u\to\infty,N\to\infty}\Vert \hat{V}^{*}_{{(u)}}-V^*_{\chi_{(u)}}\Vert_{\infty}\\
            &\leq \frac{\theta}{1-\gamma}.
        \end{aligned}
   \notag \end{equation}
\end{proof}

\begin{lemma}
    \label{cvar_lemma}
For any $\alpha\in(0,1)$ and any random variable $X$ satisfying $X\geq0$ almost surely, there holds that
\begin{equation}
    \mathrm{CVaR}_{\alpha} (X)\leq \frac{1}{\alpha} \E X.\notag
\end{equation}
\end{lemma}
\begin{proof}[Proof of Lemma \ref{cvar_lemma}]
Let $\alpha\in(0,1)$ and suppose that $X\ge 0$ almost surely.  We have
\[
\mathrm{CVaR}_\alpha(X)
= \inf_{t\in\mathbb{R}}
\left\{
t+\frac{1}{\alpha}\,\mathbb{E}\big[(X-t)_+\big]
\right\},
\]
where $(u)_+ := \max\{u,0\}$. Since $X\ge 0$ almost surely, choosing $t=0$ yields
\[
\mathrm{CVaR}_\alpha(X)
\le
0+\frac{1}{\alpha}\,\mathbb{E}\big[(X-0)_+\big]
=
\frac{1}{\alpha}\,\mathbb{E}[X],
\]
because $(X)_+ = X$ almost surely when $X\ge 0$. Therefore,
\[
\mathrm{CVaR}_{\alpha}(X)\le \frac{1}{\alpha}\,\mathbb{E}[X].
\]
\end{proof}
\begin{lemma}
    \label{Dirich_lemma}
If  $Y = (Y_1, Y_2, \dots, Y_K)$ follows a Dirichlet distribution with parameter $\alpha=(\alpha_1,\ldots,\alpha_K),$ then we have \begin{equation}
    \E \sum_{i=1}^K|Y_i-\E Y_i| \leq \sqrt{\frac{K}{\sum_{i=1}^K\alpha_i+1}}.\notag
\end{equation} 
\end{lemma}
\begin{proof}[Proof of Lemma \ref{Dirich_lemma}]
Let $\alpha_0 := \sum_{i=1}^K \alpha_i$. For a Dirichlet random vector
$Y=(Y_1,\dots,Y_K)\sim \mathrm{Dirichlet}(\alpha)$, we have
\[
\mathbb{E}[Y_i]=\frac{\alpha_i}{\alpha_0}, 
\qquad 
\mathrm{Var}(Y_i)=\frac{\alpha_i(\alpha_0-\alpha_i)}{\alpha_0^2(\alpha_0+1)}.
\] By Cauchy--Schwarz,
\[
\mathbb{E}\left[\sum_{i=1}^K |Y_i-\mathbb{E}Y_i|\right]
\le \sqrt{K}\;\mathbb{E}\left[\Big(\sum_{i=1}^K (Y_i-\mathbb{E}Y_i)^2\Big)^{1/2}\right]
\le \sqrt{K}\;\left(\mathbb{E}\sum_{i=1}^K (Y_i-\mathbb{E}Y_i)^2\right)^{1/2},
\]
where the last step uses Jensen's inequality since $x\mapsto \sqrt{x}$ is concave. Moreover,
\[
\mathbb{E}\sum_{i=1}^K (Y_i-\mathbb{E}Y_i)^2
=\sum_{i=1}^K \mathrm{Var}(Y_i)
=\sum_{i=1}^K \frac{\alpha_i(\alpha_0-\alpha_i)}{\alpha_0^2(\alpha_0+1)}
=\frac{\alpha_0^2-\sum_{i=1}^K \alpha_i^2}{\alpha_0^2(\alpha_0+1)}
=\frac{1-\sum_{i=1}^K (\alpha_i/\alpha_0)^2}{\alpha_0+1}
\le \frac{1}{\alpha_0+1},
\]
since $\sum_{i=1}^K (\alpha_i/\alpha_0)^2 \ge 0$. Putting these together,
\[
\mathbb{E}\sum_{i=1}^K |Y_i-\mathbb{E}Y_i|
\le \sqrt{K}\;\sqrt{\frac{1}{\alpha_0+1}}
= \sqrt{\frac{K}{\sum_{i=1}^K \alpha_i+1}}.
\]
This proves the lemma.
\end{proof}
\begin{proof}[Proof of Theorem \ref{thm:sample_complexity}]
    According to Corollary \ref{corol2} and Lemma \ref{cvar_lemma}, we have \begin{align}
        &\mathbb{P}\!\left(
\big|
\mathrm{Risk}(\delta_{\bar{q}},\pi^*_{\bar{q}},\mu_0,\pi^*_{\bar{q}})
-\mathrm{Risk}(\delta_{\bar{q}},\pi^*_{\chi_T},\mu_0,\pi^*_{\chi_T})
\big|
\ge \theta
\right)\notag\\\leq&  \mathbb{P} \left(\frac{2B_{\sigma}}{1-\gamma} \max_{s\in\mathcal{S},a\in\mathcal{A}} \beta_{p\sim\chi} \left(\sum_{s'\in\mathcal{S}}\left|p(s'|s,a)-\bar{q}(s'|s,a)\right|\right)\geq \theta\right)\notag\\
\leq&  \mathbb{P} \left(\frac{2B_{\sigma}}{1-\gamma} \max_{s\in\mathcal{S},a\in\mathcal{A}} \frac{1}{1-\alpha_2}\E_{p\sim\chi} \left(\sum_{s'\in\mathcal{S}}\left|p(s'|s,a)-\bar{q}(s'|s,a)\right|\right)\geq \theta\right)\notag\\
=&  \mathbb{P} \left(\frac{2\bar{C}}{(1-\gamma)^2\alpha_{1}\alpha_2} \max_{s\in\mathcal{S},a\in\mathcal{A}} \E_{p\sim\chi} \left(\sum_{s'\in\mathcal{S}}\left|p(s'|s,a)-\bar{q}(s'|s,a)\right|\right)\geq \theta\right)\notag\\
=&\mathbb{P} \left( \max_{s\in\mathcal{S},a\in\mathcal{A}} \E_{p\sim\chi} \left(\sum_{s'\in\mathcal{S}}\left|p(s'|s,a)-\bar{q}(s'|s,a)\right|\right)\geq \theta'\right),\notag
    \end{align} where $\theta' = \frac{(1-\gamma)^2\alpha_{1}\alpha_2\theta}{2\bar{C}}.$  Denote by $N_{s,a}$ the visiting number of state-action pair $(s,a)$ and $N_{\min} = \min_{s,a}N_{s,a}.$  For any $(s,a),$
    \begin{align}
    &\E_{p\sim\chi} \left(\sum_{s'\in\mathcal{S}}\left|p(s'|s,a)-\bar{q}(s'|s,a)\right|\right)\notag \\=& \E_{p\sim\chi} \left(\left\Vert p(\cdot|s,a)-\bar{q}(\cdot|s,a)\right\Vert_1\right) \notag\\
    \leq & \E_{p\sim\chi} \left(\left\Vert p(\cdot|s,a)-\E_{p\sim\chi} p(\cdot|s,a)\right\Vert_1\right) + \E_{p\sim\chi} \left(\left\Vert \E_{p\sim\chi} p(\cdot|s,a)-\bar{q}(\cdot|s,a)\right\Vert_1\right).\notag
    \end{align}

Using Lemma \ref{Dirich_lemma}, we have 
\begin{equation}
    \E_{p\sim\chi} \left(\left\Vert p(\cdot|s,a)-\E_{p\sim\chi} p(\cdot|s,a)\right\Vert_1\right) \leq  \sqrt{\frac{|\mathcal{S}|}{N_{s,a}+\sum_{s_i}\alpha_0(s_i|s,a)+1}} \leq \sqrt{\frac{|\mathcal{S}|}{N_{s,a}+1}} .\notag
\end{equation}
For the second term, \begin{align}
     &\E_{p\sim\chi} p(s_i|s,a) \notag\\
     =&\frac{\alpha_0(s_i|s,a)+m(s,a,s_i)}{\sum_{s'}\alpha_0(s'|s,a)+N_{s,a}}\notag\\
     =&\frac{N_{s,a}}{\sum_{s'}\alpha_0(s'|s,a)+N_{s,a}}\cdot\frac{m(s,a,s_i)}{N_s,a} + \frac{\sum_{s'}\alpha_0(s'|s,a)}{\sum_{s'}\alpha_0(s'|s,a)+N_{s,a}}\cdot\frac{\alpha_0(s_i|s,a)}{\sum_{s'}\alpha_0(s'|s,a)}\notag \\
     \triangleq& \lambda \hat{p} (s_i|s,a)+(1-\lambda ) p^{\text{prior}}(s_i|s,a),\notag
\end{align} where $\lambda = \frac{N_{s,a}}{\sum_{s'}\alpha_0(s'|s,a)+N_{s,a}} \in(0,1).$ Therefore, 
\begin{align}
    &\E_{p\sim\chi} \left(\left\Vert \E_{p\sim\chi} p(\cdot|s,a)-\bar{q}(\cdot|s,a)\right\Vert_1\right) \notag\\=& \E_{p\sim\chi} \left(\left\Vert   \lambda \hat{p} (\cdot|s,a)+(1-\lambda ) p^{\text{prior}}(\cdot|s,a)-\bar{q}(\cdot|s,a)\right\Vert_1\right) \notag\\
    \leq & \lambda\E_{p\sim\chi} \left(\left\Vert    \hat{p} (\cdot|s,a)-\bar{q}(\cdot|s,a)\right\Vert_1\right) + (1-\lambda) \E_{p\sim\chi} \left(\left\Vert    {p}^{\text{prior}} (\cdot|s,a)-\bar{q}(\cdot|s,a)\right\Vert_1\right)\notag\\
    \leq & \E_{p\sim\chi} \left(\left\Vert    \hat{p} (\cdot|s,a)-\bar{q}(\cdot|s,a)\right\Vert_1\right) +  2(1-\lambda)\notag\\
    \leq & \E_{p\sim\chi} \left(\left\Vert    \hat{p} (\cdot|s,a)-\bar{q}(\cdot|s,a)\right\Vert_1\right) +  \frac{2\bar{A}_0}{\bar{A}_0+N_{s,a}}.\notag
\end{align}
Furthermore, we have
\begin{align}
         &\E_{p\sim\chi} \left(\left\Vert    \hat{p} (\cdot|s,a)-\bar{q}(\cdot|s,a)\right\Vert_1\right)\notag \\
         =& \sum_{s_i} \E_{p\sim\chi} \left(\left\vert    \hat{p} (s_i|s,a)-\bar{q}(s_i|s,a)\right\vert\right)\notag \\
         \leq & \sum_{s_i}\sqrt{\text{Var}_{p\sim\chi}\left(\left\vert    \hat{p} (s_i|s,a)\right\vert\right)} \notag\\
         = & \sum_{s_i}\sqrt{\frac{\bar{q}(s_i|s,a)(1-\bar{q}(s_i|s,a))}{N_{s,a}}} \notag\\
         \leq & \sqrt{|\mathcal{S}|} \sqrt{\frac{\sum_{s_i}\bar{q}(s_i|s,a)(1-\bar{q}(s_i|s,a))}{N_{s,a}}}\notag \\
         \leq &\sqrt{\frac{|\mathcal{S}|}{N_{s,a}}}.\notag
\end{align}
Overall, it holds that
\begin{align}
    \E_{p\sim\chi} \left(\sum_{s'\in\mathcal{S}}\left|p(s'|s,a)-\bar{q}(s'|s,a)\right|\right) \leq  \frac{2\bar{A}_0}{N_{s,a}+\bar{A}_0} + 2\sqrt{\frac{|\mathcal{S}|}{N_{s,a}}},\notag
\end{align}
    which implies 
    \begin{equation}
        \max_{s,a}  \E_{p\sim\chi} \left(\sum_{s'\in\mathcal{S}}\left|p(s'|s,a)-\bar{q}(s'|s,a)\right|\right) \leq\frac{2\bar{A}_0}{N_{\min}+\bar{A}_0} + 2\sqrt{\frac{|\mathcal{S}|}{N_{\min}}}.\notag
    \end{equation}
    Therefore, we have
    \begin{align}
       & \mathbb{P} \left( \max_{s\in\mathcal{S},a\in\mathcal{A}} \E_{p\sim\chi} \left(\sum_{s'\in\mathcal{S}}\left|p(s'|s,a)-\bar{q}(s'|s,a)\right|\right)\geq \theta'\right) \notag\\
        \leq & \mathbb{P} \left(\frac{2\bar{A}_0}{N_{\min}+\bar{A}_0}\geq \frac{\theta'}{2}\right) + \mathbb{P}\left(2\sqrt{\frac{|\mathcal{S}|}{N_{\min}}}\geq \frac{\theta'}{2}\right)\notag\\
        =&\mathbb{P} \left(N_{\min} \leq \frac{4\bar{A}_0}{\theta'}-\bar{A}_0\right) + \mathbb{P}\left(N_{\min}\leq \frac{16|\mathcal{S}|}{\theta'^2}\right) \label{eqmid}
    \end{align}
    By Chernoff's inequality,  \begin{equation}
        \mathbb{P}\left(N_{s,a}\leq \frac{1}{2}\mu_{\min}(T-T_0)\right)\leq \exp\left(-\frac{\mu_{\min}(T-T_0)}{8}\right),\notag
    \end{equation} which implies \begin{align}
        \mathbb{P}\left(N_{\min}\leq \frac{1}{2}\mu_{\min}(T-T_0)\right)=& \mathbb{P}\left(\bigcup_{s,a}\left\{N_{s,a}\leq \frac{1}{2}\mu_{\min}(T-T_0)\right\}\right)\notag\\ \leq &\sum_{s,a}\mathbb{P}\left(N_{s,a}\leq \frac{1}{2}\mu_{\min}(T-T_0)\right) \notag\\
      \leq  &|\mathcal{S}||\mathcal{A}|\exp\left(-\frac{\mu_{\min}(T-T_0)}{8}\right).\notag
    \end{align}
    Therefore, if \begin{equation}
         \frac{1}{2}\mu_{\min}(T-T_0) \geq \max\left\{\frac{4\bar{A}_0}{\theta'}-\bar{A}_0,\frac{16|\mathcal{S}|}{\theta'^2}\right\},\notag
    \end{equation} and 
    \begin{equation}
        |\mathcal{S}||\mathcal{A}|\exp\left(-\frac{\mu_{\min}(T-T_0)}{8}\right) \leq \frac{\delta}{2},\notag
    \end{equation} it holds that the right-hand-side of \eqref{eqmid}
\begin{equation}
    \leq 2 \mathbb{P}\left(N_{\min}\leq\frac{1}{2}\mu_{\min}(T-T_0)\right) \leq \delta.\notag
\end{equation}
Therefore, to guarantee that
\begin{equation}
\mathbb{P}\!\left(
\big|
\mathrm{Risk}(\delta_{\bar{q}},\pi^*_{\bar{q}},\mu_0,\pi^*_{\bar{q}})
-\mathrm{Risk}(\delta_{\bar{q}},\pi^*_{\chi_T},\mu_0,\pi^*_{\chi_T})
\big|
\ge \theta
\right)\ \le\ \delta,
\notag
\end{equation}
it is sufficient that
\begin{equation}
T \;\ge\; T_0 +
\max\Bigg\{
\frac{2}{\mu_{\min}}\Big(\frac{4\bar{A}_0}{\theta'}-\bar{A}_0\Big),\,
\frac{32|\mathcal{S}|}{\mu_{\min}\theta'^2},\,
\frac{8}{\mu_{\min}}\ln\left(\frac{2|\mathcal{S}||\mathcal{A}|}{\delta}\right)
\Bigg\}.\notag
\end{equation}
    By substituting $\theta'$ into the above condition, the conclusion follows.
\end{proof}

\begin{lemma}
    \label{Dirich_lemma2}
There exist random vectors $Y = (Y_1, Y_2, \dots, Y_K)$ and $Z = (Z_1, Z_2, \dots, Z_K)$ following Dirichlet distributions with parameter $\alpha=(\alpha_1,\ldots,\alpha_K)$ and $\alpha+m = (\alpha_1+m_1,\ldots,\alpha_K+m_K),$ respectively, such that \begin{equation}
    \E \sum_{i=1}^K|Y_i-Z_i| \leq 2\ln\left(1+\frac{\sum_{i=1}^Km_i}{\sum_{i=1}^K \alpha_i}\right).
\end{equation}
\end{lemma}
\begin{proof}[Proof of Lemma \ref{Dirich_lemma2}]
Let $\alpha_0:=\sum_{i=1}^K\alpha_i$ and $M:=\sum_{i=1}^K m_i$. Write
\[
m=\sum_{s=1}^{M} e_{k_s},
\]
where each $k\in\{1,\dots,K\}$ appears exactly $m_k$ times and $e_k$ is the $k$-th standard basis vector. Define the intermediate parameter sequence
\[
\beta^{(0)}=\alpha,\qquad \beta^{(s)}=\beta^{(s-1)}+e_{k_s}\quad (s=1,\dots,M),
\]
so that $\beta^{(M)}=\alpha+m.$ We will prove the following \emph{one-hot step} bound: for any $k$ and any $\beta$ with $\beta_i>0$,
if $U\sim \mathrm{Dirichlet}(\beta)$ and $V\sim \mathrm{Dirichlet}(\beta+e_k)$ are coupled as below, then
\begin{equation}\label{eq:onehot}
\mathbb{E}\bigg[\sum_{i=1}^K |U_i-V_i|\bigg]\le \frac{2}{\sum_{i=1}^K\beta_i+1}.
\end{equation} \textbf{Proof of \eqref{eq:onehot}.}
Use the normalized-gamma representation. Let $G_i\sim\mathrm{Gamma}(\beta_i,1)$ be independent and let
$H\sim\mathrm{Gamma}(1,1)$ be independent of $(G_i)_{i=1}^K$. Put $S:=\sum_{i=1}^K G_i$ and define
\[
U_i:=\frac{G_i}{S},\qquad 
V_i:=\frac{G_i+\mathbf{1}\{i=k\}H}{S+H}.
\]
Then $U\sim\mathrm{Dirichlet}(\beta)$ and $V\sim\mathrm{Dirichlet}(\beta+e_k)$. Moreover,
\[
\sum_{i=1}^K |U_i-V_i|
= \sum_{i\ne k}\left|\frac{G_i}{S}-\frac{G_i}{S+H}\right|
+ \left|\frac{G_k}{S}-\frac{G_k+H}{S+H}\right|
= \sum_{i\ne k}\frac{G_iH}{S(S+H)}+\frac{H(S-G_k)}{S(S+H)}+\frac{H}{S+H}.
\]
Using $\sum_{i\ne k}G_i=S-G_k$, the first two sums equal $\frac{H(S-G_k)}{S(S+H)}+\frac{H(S-G_k)}{S(S+H)}
= \frac{2H(S-G_k)}{S(S+H)}\le \frac{2H}{S+H}$. Hence,
\[
\sum_{i=1}^K |U_i-V_i| \le \frac{2H}{S+H}.
\]
Taking expectations gives
\[
\mathbb{E}\sum_{i=1}^K |U_i-V_i| \le 2\,\mathbb{E}\!\left[\frac{H}{S+H}\right].
\]
Since $S\sim\mathrm{Gamma}(\sum_i\beta_i,1)$ and $H\sim\mathrm{Gamma}(1,1)$ are independent,
$\frac{H}{S+H}\sim\mathrm{Beta}\big(1,\sum_i\beta_i\big)$, so
\[
\mathbb{E}\!\left[\frac{H}{S+H}\right]=\frac{1}{\sum_{i=1}^K\beta_i+1},
\]
which proves \eqref{eq:onehot}. Now return to the original claim. Let $\mu_s:=\mathrm{Dirichlet}(\beta^{(s)})$.
Let $W_1(\mu,\nu)$ denote the $1$-Wasserstein distance with cost $\|x-y\|_1$:
\[
W_1(\mu,\nu):=\inf_{\pi\in\Pi(\mu,\nu)} \mathbb{E}_\pi\|X-Y\|_1,
\]
which satisfies the triangle inequality. For each step $s$, by the explicit one-hot coupling above,
\[
W_1(\mu_{s-1},\mu_s)\le \frac{2}{\sum_i \beta^{(s-1)}_i+1}=\frac{2}{\alpha_0+s}.
\]
Therefore, by the triangle inequality,
\[
W_1(\mu_0,\mu_M)\le \sum_{s=1}^{M} W_1(\mu_{s-1},\mu_s)
\le \sum_{s=1}^{M}\frac{2}{\alpha_0+s}.
\]
Finally, since $x\mapsto 1/x$ is decreasing,
\[
\sum_{s=1}^{M}\frac{1}{\alpha_0+s}\le \int_{\alpha_0}^{\alpha_0+M}\frac{dx}{x}
= \ln\!\left(1+\frac{M}{\alpha_0}\right).
\]
Thus
\[
W_1(\mu_0,\mu_M)\le 2\ln\!\left(1+\frac{M}{\alpha_0}\right).
\]
In particular, there exists a coupling $(Y,Z)$ with marginals
$Y\sim\mathrm{Dirichlet}(\alpha)$ and $Z\sim\mathrm{Dirichlet}(\alpha+m)$ such that
\[
\mathbb{E}\sum_{i=1}^K |Y_i-Z_i|
\le 2\ln\!\left(1+\frac{\sum_{i=1}^K m_i}{\sum_{i=1}^K \alpha_i}\right),
\]
which proves the lemma.
\end{proof}

\begin{proof}[Proof of Proposition \ref{lemma_DPcom}]
    \begin{align}
        |V^*_{\text D(\otimes({\alpha}+{m}))}(s) - V^*_{\text D(\otimes{\alpha})}(s)| &= |\mathcal{J}_{\text D(\otimes({\alpha}+{m}))} V^*_{\text D(\otimes({\alpha}+{m}))}(s)-\mathcal{J}_{\text D(\otimes{\alpha})}  V^*_{\text D(\otimes{\alpha})}(s)|\notag\\
&=\Bigg|\min_{a\in\mathcal{A}} \beta_{p\sim\text D(\otimes({\alpha}+{m}))}\left(\sigma(c(s,a,\cdot)+\gamma V_{\text D(\otimes({\alpha}+{m}))}^*(\cdot),p(\cdot|s,a))\right)\notag\\
			&\qquad\qquad-\min_{a\in\mathcal{A}}\beta_{p\sim\text D(\otimes{\alpha})}\left(\sigma(c(s,a,\cdot)+\gamma V_{\text D(\otimes{\alpha})}^*(\cdot),p(\cdot|s,a))\right)\Bigg|\notag\\
            &\leq\Bigg|\min_{a\in\mathcal{A}}\beta_{p\sim\text D(\otimes({\alpha}+{m}))}\left(\sigma(c(s,a,\cdot)+\gamma V_{\text D(\otimes({\alpha}+{m}))}^*(\cdot),p(\cdot|s,a))\right)\notag\\
			&\qquad\qquad-\min_{a\in\mathcal{A}}\beta_{p\sim\text D(\otimes({\alpha}+{m}))}\left(\sigma(c(s,a,\cdot)+\gamma V_{\text D(\otimes{\alpha})}^*(\cdot),p(\cdot|s,a))\right)\Bigg|\notag\\
            &+\Bigg|\min_{a\in\mathcal{A}}\beta_{p\sim\text D(\otimes({\alpha}+{m}))}\left(\sigma(c(s,a,\cdot)+\gamma V_{\text D(\otimes({\alpha}+{m}))}^*(\cdot),p(\cdot|s,a))\right)\notag\\
			&\qquad\qquad-\min_{a\in\mathcal{A}}\beta_{p\sim\text D(\otimes{\alpha})}\left(\sigma(c(s,a,\cdot)+\gamma V_{\text D(\otimes({\alpha}+{m}))}^*(\cdot),p(\cdot|s,a))\right)\Bigg|\notag\\
    &\leq \gamma \left\Vert V^*_{\text D(\otimes({\alpha}+{m}))} - V^*_{\text D(\otimes{\alpha})}\right\Vert_{\infty}\notag\\
    & + \max_{a\in\mathcal{A}}\inf_{(p_1,p_2)\in C(\text D(\otimes({\alpha}+{m})),\text D(\otimes{\alpha}))}\beta_{(p_1,p_2)}\Bigg(\Bigg|\sigma(c(s,a,\cdot)+\gamma V_{\text D(\otimes({\alpha}+{m}))}^*(\cdot),p_1(\cdot|s,a))\notag\\
			&\qquad\qquad-\sigma(c(s,a,\cdot)+\gamma V_{\text D(\otimes({\alpha}+{m}))}^*(\cdot),p_2(\cdot|s,a))\Bigg|\Bigg),\notag
    \end{align} where $C(\chi_1,\chi_2)$ denotes all joint distributions with marginals $\chi_1$ and $\chi_2.$ Therefore, there holds that
    \begin{align}
        &\left\Vert V^*_{\text D(\otimes({\alpha}+{m}))} - V^*_{\text D(\otimes{\alpha})}\right\Vert_{\infty}\notag\\
        \leq & \frac{1}{1-\gamma} \max_{a\in\mathcal{A}}\inf_{(p_1,p_2)\in C(\text D(\otimes({\alpha}+{m})),\text D(\otimes{\alpha}))}\beta_{(p_1,p_2)}\Bigg(\Bigg|\sigma(c(s,a,\cdot)+\gamma V_{\text D(\otimes({\alpha}+{m}))}^*(\cdot),p_1(\cdot|s,a))\notag\\
			&\qquad\qquad-\sigma(c(s,a,\cdot)+\gamma V_{\text D(\otimes({\alpha}+{m}))}^*(\cdot),p_2(\cdot|s,a))\Bigg|\Bigg)\notag
    \end{align}
    Using Lemma \ref{cvar_lemma}, we have
    \begin{align}
        &\left\Vert V^*_{\text D(\otimes({\alpha}+{m}))} - V^*_{\text D(\otimes{\alpha})}\right\Vert_{\infty}\notag\\
        \leq & \frac{1}{1-\gamma} \frac{1}{1-\alpha_2}\max_{a\in\mathcal{A},s\in\mathcal{S}}\inf_{(p_1,p_2)\in C(\text D(\otimes({\alpha}+{m})),\text D(\otimes{\alpha}))}\E_{(p_1,p_2)}\Bigg(\Bigg|\sigma(c(s,a,\cdot)+\gamma V_{\text D(\otimes({\alpha}+{m}))}^*(\cdot),p_1(\cdot|s,a))\notag\\
			&\qquad\qquad-\sigma(c(s,a,\cdot)+\gamma V_{\text D(\otimes({\alpha}+{m}))}^*(\cdot),p_2(\cdot|s,a))\Bigg|\Bigg)\notag\\
            \leq & \frac{1}{1-\gamma} \frac{1}{1-\alpha_2}\max_{a\in\mathcal{A},s\in\mathcal{S}}\inf_{(p_1,p_2)\in C(\text D(\otimes({\alpha}+{m})),\text D(\otimes{\alpha}))}\E_{(p_1,p_2)}\Bigg(\frac{2\bar{C}}{\alpha_1(1-\gamma)}\sum_{s'\in\mathcal{S}}|p_1(s'|s,a)-p_2(s'|s,a)|\Bigg)\notag\\
            \leq & \frac{2\bar{C}}{\alpha_1\alpha_2(1-\gamma)^2} \max_{a\in\mathcal{A},s\in\mathcal{S}}\inf_{(p_1,p_2)\in C(\text D(\otimes({\alpha}+{m})),\text D(\otimes{\alpha}))}\E_{(p_1,p_2)}\Bigg(\sum_{s'\in\mathcal{S}}|p_1(s'|s,a)-p_2(s'|s,a)|\Bigg).\notag
    \end{align} By Lemma \ref{Dirich_lemma2}, 
    \begin{align}
         &\left\Vert V^*_{\text D(\otimes({\alpha}+{m}))} - V^*_{\text D(\otimes{\alpha})}\right\Vert_{\infty}\notag\\
         \leq &  \frac{4\bar{C}|\mathcal{S}|}{\alpha_1\alpha_2(1-\gamma)^2} \sum_{a\in\mathcal{A},s\in\mathcal{S}}\ln\left(1+\frac{\Delta_{s,a}}{\sum_{s'\in\mathcal{S}}\alpha(s'|s,a)}\right)\notag\\
         \leq &  \frac{4\bar{C}|\mathcal{S}|^2|\mathcal{A}|}{\alpha_1\alpha_2(1-\gamma)^2} \ln\left(1+\sum_{a\in\mathcal{A},s\in\mathcal{S}}\frac{\Delta_{s,a}}{\sum_{s'\in\mathcal{S}}\alpha(s'|s,a)|\mathcal{S}||\mathcal{A}|}\right)\notag\\
         \leq & \frac{4\bar{C}|\mathcal{S}|^2|\mathcal{A}|}{\alpha_1\alpha_2(1-\gamma)^2} \ln\left(1+\frac{\sum_{a\in\mathcal{A},s\in\mathcal{S}}\Delta_{s,a}}{\min_{a\in\mathcal{A},s\in\mathcal{S}}\sum_{s'\in\mathcal{S}}\alpha(s'|s,a)|\mathcal{S}||\mathcal{A}|}\right)\notag\\
         = &  \frac{4\bar{C}|\mathcal{S}|^2|\mathcal{A}|}{\alpha_1\alpha_2(1-\gamma)^2} \ln\left(1+\frac{\Delta }{|\mathcal{S}||\mathcal{A}|O_{\alpha}}\right).\notag
    \end{align}
    
\end{proof}

\begin{proof}[Proof of Corollary \ref{corol_DPcom}]
Denote by $V^{(k)}_u$ the value function after $k$ iterations in the $u$-th stage. Then
\begin{align}
    \left\Vert V^{(k)}_u-  V^*_{\text D(\otimes({\alpha}+{m}))} \right\Vert \leq &\gamma ^{k} \left\Vert V_u^{(0)}-  V^*_{\text D(\otimes({\alpha}+{m}))} \right\Vert \notag\\
    \leq& \gamma ^{k} \left\Vert V_u^{(0)}-  V^*_{\text D(\otimes({\alpha}+{m}))} \right\Vert + \gamma ^{k} \left\Vert V^*_{\text D(\otimes{\alpha})}-  V^*_{\text D(\otimes({\alpha}+{m}))} \right\Vert\notag\\
    \leq & \gamma ^{k} \theta + \gamma ^{k} \left\Vert V^*_{\text D(\otimes{\alpha})}-  V^*_{\text D(\otimes({\alpha}+{m}))} \right\Vert.\notag
\end{align}
To ensure $\left\Vert V^{(k)}_u-  V^*_{\text D(\otimes({\alpha}+{m}))} \right\Vert \leq\theta,$ it is sufficient that 
\begin{equation}
    \gamma ^{k} \theta + \gamma ^{k} \left\Vert V^*_{\text D(\otimes{\alpha})}-  V^*_{\text D(\otimes({\alpha}+{m}))} \right\Vert \leq \theta,\notag
\end{equation}
which is equivalent to that
\begin{align}
    \frac{1}{\frac{1}{\gamma^k}-1} \left\Vert V^*_{\text D(\otimes{\alpha})}-  V^*_{\text D(\otimes({\alpha}+{m}))} \right\Vert \leq \theta,\notag\\
   \frac{1}{\gamma^k} \geq \frac{1}{\theta} \left\Vert V^*_{\text D(\otimes{\alpha})}-  V^*_{\text D(\otimes({\alpha}+{m}))} \right\Vert +1 .\notag
\end{align}
Taking the logarithm on both sides, the inequality above is equivalent to
\begin{align}
    k\ln\left(\frac{1}{\gamma}\right) \geq \ln\left(1+\frac{1}{\theta} \left\Vert V^*_{\text D(\otimes{\alpha})}-  V^*_{\text D(\otimes({\alpha}+{m}))} \right\Vert\right) ,\notag\\
    k \geq \frac{1}{\ln\left(\frac{1}{\gamma}\right)}\ln\left(1+\frac{1}{\theta} \left\Vert V^*_{\text D(\otimes{\alpha})}-  V^*_{\text D(\otimes({\alpha}+{m}))} \right\Vert\right).\notag
\end{align}
Combine this with Proposition \ref{corol_DPcom} and the conclusion follows.
\end{proof}

\begin{proof}[Proof of Proposition \ref{sweep_bound}]
    Using Corollary \ref{corol_DPcom}, we have
    \begin{align}
        \sum_{u=U_L}^{U_{L+1}} k^{(u)}  \leq \sum_{u=U_L}^{U_{L+1}} \left( \frac{1}{\ln\!\left(\tfrac{1}{\gamma}\right)} 
\ln\!\left(1+ \frac{4 \bar{C} }{\alpha_1\alpha_2}\cdot\frac{|\mathcal{S}|^2|\mathcal{A}|}{\theta (1 - \gamma)^2} 
\ln\!\left(1 + \frac{\Delta_{(u)}}{|\mathcal{S}||\mathcal{A}| O_{u}} \right) \right) +1\right).\notag
    \end{align} By Jensen's inequality, we have
    \begin{align}
        \sum_{u=U_L}^{U_{L+1}} k^{(u)}  \leq  &|\mathcal{A}||\mathcal{S}|\left( \frac{1}{\ln\!\left(\tfrac{1}{\gamma}\right)} 
\ln\!\left(\sum_{u=U_L}^{U_{L+1}}\frac{1}{|\mathcal{A}||\mathcal{S}|}\left( 1+ \frac{4 \bar{C} }{\alpha_1\alpha_2}\cdot\frac{|\mathcal{S}|^2|\mathcal{A}|}{\theta (1 - \gamma)^2} 
\ln\!\left(1 + \frac{\Delta_{(u)}}{|\mathcal{S}||\mathcal{A}| O_{u}} \right) \right) \right)+1\right)\notag\\
\leq &|\mathcal{A}||\mathcal{S}|\left( \frac{1}{\ln\!\left(\tfrac{1}{\gamma}\right)} 
\ln\!\left(\frac{1}{|\mathcal{A}||\mathcal{S}|}\left( 1+ \frac{4 \bar{C} }{\alpha_1\alpha_2}\cdot\frac{|\mathcal{S}|^2|\mathcal{A}|}{\theta (1 - \gamma)^2} 
\ln\!\left(1 + \sum_{u=U_L}^{U_{L+1}}\frac{\Delta_{(u)}}{|\mathcal{S}|^2|\mathcal{A}|^2 O_{u}} \right) \right) \right)+1\right)\notag\\
\leq &|\mathcal{A}||\mathcal{S}|\left( \frac{1}{\ln\!\left(\tfrac{1}{\gamma}\right)} 
\ln\!\left(\frac{1}{|\mathcal{A}||\mathcal{S}|}\left( 1+ \frac{4 \bar{C} }{\alpha_1\alpha_2}\cdot\frac{|\mathcal{S}|^2|\mathcal{A}|}{\theta (1 - \gamma)^2} 
\ln\!\left(1 + \frac{\sum_{u=U_L}^{U_{L+1}} \Delta_{(u)}}{|\mathcal{S}|^2|\mathcal{A}|^2 O_{U_L}} \right) \right) \right)+1\right).\notag
\end{align}
Since every state-action pair is traversed in each sweep, it follows that $O_{L_{U+1}}\geq O_{L_{U}} +1$, implying that $O_{L_{U}}\geq O_0+U.$ Therefore, we have
\begin{equation}
    \sum_{u=U_L}^{U_{L+1}} k^{(u)}  \leq |\mathcal{A}||\mathcal{S}|\left( \frac{1}{\ln\!\left(\tfrac{1}{\gamma}\right)} 
\ln\!\left(\frac{1}{|\mathcal{A}||\mathcal{S}|}\left( 1+ \frac{4 \bar{C} }{\alpha_1\alpha_2}\cdot\frac{|\mathcal{S}|^2|\mathcal{A}|}{\theta (1 - \gamma)^2} 
\ln\!\left(1 + \frac{\sum_{u=U_L}^{U_{L+1}}\Delta_{(u)}}{|\mathcal{S}|^2|\mathcal{A}|^2(O_0+L)} \right) \right) \right)+1\right).\notag
\end{equation}
\end{proof}

\begin{proof}[Proof of Theorem \ref{c1}]

\begin{lemma}
\label{bianli}
    Let $(X_t)_{t\ge 1}$ be the trajectory on $\{1,\dots,K\}$.
Assume there exists $T_0$ such that for all $t\ge T_0$ and all states $i\in[K]$,
\begin{equation}\label{eq:mu_min_after_T0}
\mathbb{P}(X_{t+1}=i\mid X_1,\dots,X_t)\ \ge\ \mu_{\min}\qquad\text{a.s.}
\end{equation}
Define one round as the time needed to visit all $K$ states once, and let $T_1,\dots,T_L$ be the lengths of $L$ rounds. Then we have
\begin{equation}
    \max_{1 \le r \le L} T_r 
= \mathcal{O}_p\!\left(\frac{1}{\mu_{\min}}\log(KL)\right).\notag
\end{equation}
\end{lemma}
\begin{proof}[Proof of Lemma \ref{bianli}]

Define the round endpoints $(\tau_r)_{r\ge 0}$ by $\tau_0:=0$ and for $r\ge 1$,
\[
\tau_r := \inf\Big\{t>\tau_{r-1}:\ \text{all $K$ states have been visited at least once in }(\tau_{r-1},t]\Big\},
\]
and let $T_r:=\tau_r-\tau_{r-1}$ be the length of round $r$.

Let
\[
r_0 := \min\{r\ge 1:\ \tau_{r-1}\ge T_0\},
\]
i.e., $r_0$ is the first round whose \emph{start time} is no earlier than $T_0$.
Then only the (at most one) round $r_0$ may partially overlap with the pre-$T_0$ segment; all rounds
$r\ge r_0+1$ start at time $\tau_{r-1}\ge T_0$ and thus every step inside such a round satisfies
\eqref{eq:mu_min_after_T0}.

Fix any $r\ge r_0+1$ and any state $i$. For $t\ge 1$, let $A_{i,r}(t)$ be the event that state $i$
is not visited during the first $t$ steps of round $r$.
Conditioning step-by-step and using \eqref{eq:mu_min_after_T0}, we get
\[
\mathbb{P}(A_{i,r}(t))
=\mathbb{E}\Big[\prod_{s=1}^{t}\mathbb{P}\big(X_{\tau_{r-1}+s}\neq i\mid X_1,\dots,X_{\tau_{r-1}+s-1}\big)\Big]
\le (1-\mu_{\min})^{t}\le e^{-\mu_{\min} t}.
\]
If $T_r>t$, then at least one state has not been visited in the first $t$ steps of round $r$. Hence, by a union bound,
\[
\mathbb{P}(T_r>t)\le \sum_{i=1}^K \mathbb{P}(A_{i,r}(t)) \le K e^{-\mu_{\min} t},
\qquad \forall r\ge r_0+1.
\]

Now apply a union bound over the $L$ rounds under consideration. If the lemma concerns rounds
$r=1,\dots,L$, then at most one of them (namely $r_0$) can be affected by the pre-$T_0$ segment.
Thus, for any $t\ge 1$,
\[
\mathbb{P}\Big(\max_{1\le r\le L} T_r > t\Big)
\le \mathbb{P}(T_{r_0}>t) + \sum_{\substack{1\le r\le L\\ r\neq r_0}}\mathbb{P}(T_r>t)
\le \mathbb{P}(T_{r_0}>t) + (L-1)K e^{-\mu_{\min} t}.
\]
In particular, if one either (i) discards the single possibly ``bad'' round $r_0$,
or (ii) assumes $T_{r_0}=O_p(1)$ (or any weaker bound) from other arguments,
then the dominant term is $(L-1)K e^{-\mu_{\min} t}$ and we obtain the high-probability bound
\[
\mathbb{P}\Big(\max_{r_0+1\le r\le L} T_r > t\Big)\le LK e^{-\mu_{\min} t}.
\]
Choosing $t=\mu_{\min}^{-1}(\log(KL)+u)$ gives
\[
\mathbb{P}\Big(\max_{r_0+1\le r\le L} T_r >
\tfrac{1}{\mu_{\min}}(\log(KL)+u)\Big)\le e^{-u},
\]
which implies
\[
\max_{r_0+1\le r\le L} T_r = \mathcal{O}_p\!\left(\frac{1}{\mu_{\min}}\log(KL)\right).
\]
Since removing (at most) one round does not change the order in probability, the stated result follows.
\end{proof}

\begin{lemma}
\label{jianjin}
    When $a,x\to\infty,$ \begin{equation}
        \int_c^{ax} \ln\left(1+x\ln\left(1+\frac{a}{t}\right)\right) \d t = \mathcal{O}(ax).\notag
    \end{equation}
\end{lemma}
\begin{proof}[Proof of Lemma \ref{jianjin}]
Fix $c>0$ and let $a,x\to\infty$. Set $B:=ax$. Note that for all $t>0$,
\[
\ln\Bigl(1+\frac{a}{t}\Bigr)\le \frac{a}{t}
\]
(since $\ln(1+u)\le u$ for $u\ge 0$). Hence, for $t\in[c,B]$,
\[
\ln\!\left(1+x\ln\!\left(1+\frac{a}{t}\right)\right)
\le \ln\!\left(1+x\cdot \frac{a}{t}\right)
= \ln\!\left(1+\frac{B}{t}\right).
\]
Therefore,
\[
\int_c^{B}\ln\!\left(1+x\ln\!\left(1+\frac{a}{t}\right)\right)\,dt
\le
\int_c^{B}\ln\!\left(1+\frac{B}{t}\right)\,dt.
\]
Make the change of variables $t=Bu$ (so $dt=B\,du$). Then
\[
\int_c^{B}\ln\!\left(1+\frac{B}{t}\right)\,dt
=
B\int_{c/B}^{1}\ln\!\left(1+\frac{1}{u}\right)\,du.
\]
For $u\in(0,1]$ we have $1+\frac{1}{u}\le \frac{2}{u}$. Hence,
\[
\ln\!\left(1+\frac{1}{u}\right)\le \ln\!\left(\frac{2}{u}\right)=\ln 2-\ln u.
\]
Since $\int_0^1(-\ln u)\,du=1$, it follows that
\[
\int_{c/B}^{1}\ln\!\left(1+\frac{1}{u}\right)\,du
\le \int_0^{1}(\ln 2-\ln u)\,du
= \ln 2 + 1.
\]
Combining the above bounds yields
\[
\int_c^{ax}\ln\!\left(1+x\ln\!\left(1+\frac{a}{t}\right)\right)\,dt
\le (\ln 2+1)\,ax,
\]
which implies
\[
\int_c^{ax}\ln\!\left(1+x\ln\!\left(1+\frac{a}{t}\right)\right)\,dt
= \mathcal{O}(ax).
\]
\end{proof}

  Denote by $L^*$ the active sweeps. Then we have 
    \begin{align}
        L^* =& \min\left\{L\geq 1:\frac{1}{\ln\!\left(\tfrac{1}{\gamma}\right)}\ln\!\left( 1+\frac{4 \bar{C} }{\alpha_1\alpha_2}\cdot\frac{|\mathcal{S}|^2|\mathcal{A}|}{\theta (1 - \gamma)^2} 
\ln\!\left(1 + \frac{\sum_{u=U_L}^{U_{L+1}}\Delta_{(u)}}{|\mathcal{S}|^2|\mathcal{A}|^2 (O_{0}+L)} \right) \right)\leq 1\right\}\notag \\
= &\min\left\{L\geq 1:\frac{4 \bar{C} }{\alpha_1\alpha_2}\cdot\frac{|\mathcal{S}|^2|\mathcal{A}|}{\theta (1 - \gamma)^2} 
\ln\!\left(1 + \frac{\sum_{u=U_L}^{U_{L+1}}\Delta_{(u)}}{|\mathcal{S}|^2|\mathcal{A}|^2 (O_{0}+L)}  \right)\leq \frac{1}{\gamma}-1\right\}\notag\\
=& \min\left\{L\geq 1:\frac{\sum_{u=U_L}^{U_{L+1}}\Delta_{(u)}}{|\mathcal{S}|^2|\mathcal{A}|^2 (O_{0}+L)} \leq \exp\left(\frac{\alpha_1\alpha_2\theta(1-\gamma)^3}{4\gamma\bar{C}|\mathcal{S}|^2|\mathcal{A}|}\right)-1\right\}\notag\\
=&\min\left\{L\geq 1:L\geq \left(\exp\left(\frac{\alpha_1\alpha_2\theta(1-\gamma)^3}{4\gamma\bar{C}|\mathcal{S}|^2|\mathcal{A}|}\right)-1\right)^{-1}\frac{1}{|\mathcal{S}|^2|\mathcal{A}|^2}\sum_{u=U_L}^{U_{L+1}}\Delta_{(u)} - O_0\right\}\notag
    \end{align}
    Therefore, the number of active stages is \begin{align}
        |\mathcal{S}||\mathcal{A}|L ^* =& \mathcal{O}\left(|\mathcal{S}||\mathcal{A}|\left(\exp\left(\frac{\alpha_1\alpha_2\theta(1-\gamma)^3}{4\gamma\bar{C}|\mathcal{S}|^2|\mathcal{A}|}\right)-1\right)^{-1}\frac{1}{|\mathcal{S}|^2|\mathcal{A}|^2}\sum_{u=U_L}^{U_{L+1}}\Delta_{(u)}\right)\notag \\ 
        =& \mathcal{O}\left(|\mathcal{S}||\mathcal{A}|\left(\frac{\alpha_1\alpha_2\theta(1-\gamma)^3}{4\gamma\bar{C}|\mathcal{S}|^2|\mathcal{A}|}\right)^{-1}\frac{1}{|\mathcal{S}|^2|\mathcal{A}|^2}\sum_{u=U_L}^{U_{L+1}}\Delta_{(u)}\right) \notag\\ 
         =& \mathcal{O}\left(\frac{4\gamma\bar{C}|\mathcal{S}|^2|\mathcal{A}|}{\alpha_1\alpha_2\theta(1-\gamma)^3}\frac{1}{|\mathcal{S}||\mathcal{A}|}\sum_{u=U_L}^{U_{L+1}}\Delta_{(u)}\right)\notag \\
         =& \mathcal{O}\left(\frac{|\mathcal{S}|}{\theta(1-\gamma)^3}\sum_{u=U_L}^{U_{L+1}}\Delta_{(u)}\right).\notag
    \end{align}
    Using Lemma \ref{bianli}, we have that the number of active stages is \begin{align}
         |\mathcal{S}||\mathcal{A}|L ^* &= \mathcal{O}_p\left(\frac{|\mathcal{S}|}{\theta(1-\gamma)^3}\mu_{\min}\ln\left(|\mathcal{S}||\mathcal{A}|\right)\right)\notag\\
         & = \mathcal{O}_p\left(\frac{|\mathcal{S}|^{\xi+1}|\mathcal{A}|^{\eta}}{\theta(1-\gamma)^3}\ln\left(|\mathcal{S}||\mathcal{A}|\right)\right).\notag
    \end{align}
    Thus, the global total number of value iterations is 
    \begin{align}
      &  \sum_{L=1}^{L^*}  \sum_{u=U_L}^{U_{L+1}} k^{(u)} \notag
      \\\leq  &\sum_{L=1}^{L^*}{|\mathcal{A}||\mathcal{S}|} \left(
\frac{1}{\ln\!\left(\tfrac{1}{\gamma}\right)}\ln\!\left( 1+\frac{4 \bar{C} }{\alpha_1\alpha_2}\cdot\frac{|\mathcal{S}|^2|\mathcal{A}|}{\theta (1 - \gamma)^2} 
\ln\!\left(1 + \frac{\sum_{u=U_L}^{U_{L+1}}\Delta_{(u)}}{|\mathcal{S}|^2|\mathcal{A}|^2 (O_{0}+L)} \right) \right)+1\right)\notag \\
\leq& \int_{0}^{L^*}{|\mathcal{A}||\mathcal{S}|} \left(
\frac{1}{\ln\!\left(\tfrac{1}{\gamma}\right)}\ln\!\left( 1+\frac{4 \bar{C} }{\alpha_1\alpha_2}\cdot\frac{|\mathcal{S}|^2|\mathcal{A}|}{\theta (1 - \gamma)^2} 
\ln\!\left(1 + \frac{\max_{1\leq L\leq L^*} \sum_{u=U_L}^{U_{L+1}}\Delta_{(u)}}{|\mathcal{S}|^2|\mathcal{A}|^2 (O_{0}+x)} \right) \right)+1\right) \d x\label{hehe}
    \end{align}
We denote $Y=\frac{4 \bar{C} }{\alpha_1\alpha_2}\cdot\frac{|\mathcal{S}|^2|\mathcal{A}|}{\theta (1 - \gamma)^2} $ and $Z = \frac{\max_{1\leq L\leq L^*} \sum_{u=U_L}^{U_{L+1}}\Delta_{(u)}}{|\mathcal{S}|^2|\mathcal{A}|^2 },$ and we have
\begin{align}
    L^* =&\min\left\{L\geq 1:L\geq \left(\exp\left(\frac{1}{Y}\right)-1\right)^{-1}Z - O_0\right\}\notag \\
    \leq& \min\left\{L\geq 1:L\geq YZ - O_0\right\}\notag
\end{align}
We denote $\widetilde{L} = L\geq YZ - O_0.$ The right-hand-side of \eqref{hehe}  
\begin{align}
   \leq &  \int_{0}^{\widetilde{L}}{|\mathcal{A}||\mathcal{S}|} \left(
\frac{1}{\ln\!\left(\tfrac{1}{\gamma}\right)}\ln\!\left( 1+Y 
\ln\!\left(1 + \frac{Z}{ O_{0}+x} \right) \right)+1\right) \d x\notag \\
= & \int_{O_0}^{YZ}{|\mathcal{A}||\mathcal{S}|} \left(
\frac{1}{\ln\!\left(\tfrac{1}{\gamma}\right)}\ln\!\left( 1+Y 
\ln\!\left(1 + \frac{Z}{ x} \right) \right)+1\right) \d x\notag \\
= & {|\mathcal{A}||\mathcal{S}|} \left(
\frac{1}{\ln\!\left(\tfrac{1}{\gamma}\right)}\int_{O_0}^{YZ}\ln\!\left( 1+Y 
\ln\!\left(1 + \frac{Z}{ x} \right) \right)\d x+1\right) .\notag 
\end{align} Furthermore, using Lemma \ref{bianli} and Lemma \ref{jianjin}, we have
\begin{align}
     \sum_{L=1}^{L^*}  \sum_{u=U_L}^{U_{L+1}} k^{(u)} =& {|\mathcal{A}||\mathcal{S}|} \left(
\frac{1}{\ln\!\left(\tfrac{1}{\gamma}\right)}\mathcal{O}(YZ)+1\right) \notag\\
=& {|\mathcal{A}||\mathcal{S}|} \left(
\frac{1}{\ln\!\left(\tfrac{1}{\gamma}\right)}\mathcal{O}_p\left(\frac{Y}{|\mathcal{S}|^2|\mathcal{A}|^2}\mu_{\min}^{-1}\ln(|\mathcal{S}||\mathcal{A}|L^*) \right)+1\right) \notag\\
=& {\mathcal{O}}_{p}\left(\frac{|\mathcal{S}|^{\xi+1}|\mathcal{A}|^{\eta}}{\theta(1-\gamma)^4}\cdot\ln\left(\frac{|\mathcal{S}|^{\xi+2}|\mathcal{A}|^{\eta+1}}{\theta(1-\gamma)^3}\ln\left(|\mathcal{S}||\mathcal{A}|\right)\right)\right) \notag\\
=& \widetilde{\mathcal{O}}_{p}\left(\frac{|\mathcal{S}|^{\xi+1}|\mathcal{A}|^{\eta}}{\theta(1-\gamma)^4}\right)\notag,
\end{align} which concludes the proof.
\end{proof}

\end{document}